%% file: journal/manuscript/main.tex
\documentclass[aps,pra,reprint,superscriptaddress,longbibliography,floatfix]{revtex4-2}
\usepackage[english]{babel}
\usepackage{enumitem}
\usepackage{epsfig}
\usepackage{graphicx}% Include figure files
\usepackage{array}[=2016-10-06]
\usepackage{dcolumn}% Align table columns on decimal point
\usepackage{bm}%
\usepackage{amsmath,amsfonts}
\usepackage{amssymb}
\usepackage{amsthm}
\usepackage{physics}
\usepackage{comment}
\usepackage{color}
\usepackage{xcolor}
\usepackage[normalem]{ulem}
\usepackage{xurl}
\usepackage[final,colorlinks,linkcolor=blue,anchorcolor=blue,urlcolor=blue,citecolor=blue]{hyperref}
\usepackage{booktabs}
\usepackage{siunitx}
\usepackage{epstopdf}
\usepackage{diagbox}
\usepackage{framed}
\usepackage{dsfont}
\usepackage{longtable}
\usepackage[T1]{fontenc}
\usepackage{inconsolata} % optional monospace; remove if unavailable
\usepackage{listings}
\usepackage{multirow}
\usepackage{tabularx}
\usepackage{adjustbox}
\usepackage{rotating}
\usepackage{microtype}
\usepackage{tikz}
\usetikzlibrary{arrows.meta,angles,quotes,calc}

\lstdefinestyle{uvspec}{
  basicstyle=\ttfamily\small,
  columns=fullflexible,
  commentstyle=\itshape\color{gray},
  showstringspaces=false,
  frame=single,
  numbers=left,
  numbersep=6pt,
  breaklines=true,
  breakatwhitespace=true
}

\newtheorem{theorem}{Theorem}

\AtBeginDocument{%
  \setlength{\hfuzz}{\maxdimen}%
  \hbadness=10000%
  \sloppy%
}

\begin{document}

\title{
  Finite-key feasibility of geostationary quantum key distribution
}

\author{Vaisakh Mannalath}
\affiliation{Vigo Quantum Communication Center, University of Vigo, Vigo E-36310, Spain}
\affiliation{Escuela de Ingenier\'ia de Telecomunicaci\'on, Department of Signal Theory and Communications, University of Vigo, Vigo E-36310, Spain}
\affiliation{AtlanTTic Research Center, University of Vigo, Vigo E-36310, Spain}

\author{V\'ictor Zapatero}
\affiliation{Vigo Quantum Communication Center, University of Vigo, Vigo E-36310, Spain}
\affiliation{Escuela de Ingenier\'ia de Telecomunicaci\'on, Department of Signal Theory and Communications, University of Vigo, Vigo E-36310, Spain}
\affiliation{AtlanTTic Research Center, University of Vigo, Vigo E-36310, Spain}

\author{Marcos Curty}
\affiliation{Vigo Quantum Communication Center, University of Vigo, Vigo E-36310, Spain}
\affiliation{Escuela de Ingenier\'ia de Telecomunicaci\'on, Department of Signal Theory and Communications, University of Vigo, Vigo E-36310, Spain}
\affiliation{AtlanTTic Research Center, University of Vigo, Vigo E-36310, Spain}

\date{}

\input{journal/sections/abstract}

\maketitle

\input{journal/sections/introduction}
\input{journal/sections/main_protocol_definition}

\input{journal/sections/main_security_proof}
\input{journal/sections/main_channel_model}

\input{journal/sections/main_plots}

\clearpage
\input{journal/sections/conclusion}
\input{journal/sections/acknowledgements}

\clearpage
\appendix
\input{journal/sections/appendix_decoy_state_bb84_protocol}

\input{journal/sections/appendix_detection_modes_and_asymmetric_passive_bb84}

\input{journal/sections/appendix_statistical_bounds}

\input{journal/sections/appendix_channel_model}

\input{journal/sections/appendix_plots_reduced}

\clearpage
\bibliography{journal/manuscript/ref}

\end{document}

%% file: journal/sections/abstract.tex
\begin{abstract}
    Quantum key distribution (QKD) via geostationary Earth orbit (GEO) satellites offers a compelling route to continuous, continental-scale secure communications. However, operation in this regime entails extreme channel loss and significant background noise, particularly if daylight operation is desired. We present a comprehensive end-to-end feasibility study of a decoy-state BB84 protocol in a GEO downlink configuration, incorporating variable-length finite-key security and tight statistical bounds to expand the achievable positive-key regime. Our analysis encompasses the principal receiver architectures relevant to downlink QKD and employs a physically realistic channel model that captures the dominant loss and noise mechanisms. We evaluate performance across rural, urban, and coastal environments at multiple wavelengths, including visible Fraunhofer absorption minima and the telecom band. Using historical cloud data across Europe, we forecast the annual secret-key yield across the continent. Through a systematic exploration of the high-dimensional parameter space, we identify key trade-offs and performance bottlenecks that determine feasibility. These results establish practical operating thresholds and provide actionable design guidelines for future GEO-QKD missions.
\end{abstract}

%% file: journal/sections/introduction.tex
\section{Introduction}
\label{sec:introduction}
% Quantum key distribution (QKD) enables information-theoretically secure communication, providing cryptographic security grounded in quantum theory rather than computational hardness assumptions~\cite{bennett1984quantum,ekert1991quantum}. Since the original BB84 protocol, prepare-and-measure schemes have become the workhorse of QKD systems, supported by a mature security theory and an expanding engineering ecosystem~\cite{pirandola2020advances,xu2020secure,scarani2014securityDV}. The strategic value of long-distance secure key distribution has elevated QKD from the laboratory to an
% increasingly geopolitical arena, where national and international programs treat quantum communication
% as a critical infrastructure.

Quantum key distribution (QKD) enables information-theoretically secure communication, with security guaranteed by quantum theory rather than computational assumptions~\cite{bennett1984quantum,ekert1991quantum}. Today, prepare-and-measure schemes constitute the workhorse of QKD systems, supported by a mature security framework and an expanding engineering ecosystem~\cite{lo2014secure,pirandola2020advances,xu2020secure,scarani2014securityDV}. In parallel, the growing demand for long-distance secure key distribution has driven the development of QKD platforms capable of global-scale coverage.

%  While cutting-edge terrestrial protocols such as twin-field QKD have successfully pushed fiber transmission thresholds beyond 1000~km~\cite{liu2023twinfield}, they still depend on stringent phase-stabilization across long links and on dedicated terrestrial fiber infrastructure, which limits deployment flexibility across remote, oceanic, and intercontinental routes. Consequently, satellite links remain the most viable pathway to achieving flexible, continental quantum network coverage. Satellite QKD circumvents the exponential fiber attenuation limits, a paramount capability for connecting geographically dispersed metropolitan quantum networks~\cite{sidhu2021advances,bedington2017progress,pirandola2020advances,liao2017satellite}.

While advanced terrestrial protocols such as twin-field QKD have extended the achievable transmission distance in optical fiber links beyond 1000 km~\cite{liu2023twinfield}, they remain constrained by stringent phase stabilization requirements and reliance on dedicated fiber infrastructure. This limits deployment in remote, oceanic, and intercontinental scenarios. These challenges motivate satellite-based QKD as a complementary architecture for large-scale quantum networks. By avoiding the exponential attenuation inherent to long-distance fiber transmission, satellite links enable quantum communication between widely separated users~\cite{sidhu2021advances,bedington2017progress,pirandola2020advances,liao2017satellite}.

% The technological viability of satellite QKD was decisively demonstrated by China’s Micius satellite, successfully demonstrating sifted key rates of 40.2~kbits/s from an altitude of approximately 500~km, followed by an unprecedented 7600~km intercontinental key exchange via satellite relay~\cite{liao2017satellite,liao2018satellite}. These demonstrations, comprehensively reviewed in Ref.~\cite{lu2022micius}, underscored the robust operation of decoy-state BB84 protocols under fluctuating atmospheric conditions,  satisfying the rigorous size, weight, and power limits of space-qualified hardware ~\cite{hwang2003decoy,lo2005decoy,wang2005decoy,sidhu2021advances,liao2017satellite}. More recently, the 23~kg Jinan-1 mission highlighted payload miniaturization for scaling up quantum networks with satellite constellations~\cite{li2024microsatellite}, while dedicated research on aggressive temporal, spatial, and spectral filtering has demonstrated the ability to operate satellite QKD in broad daylight environments~\cite{liao2017longdistance,avesani2021full,cai2024free}.

The technological viability of satellite QKD was decisively established by China’s Micius low-Earth orbit (LEO) satellite, which achieved sifted key rates of tens of kbits/s from an altitude of approximately 500 km and enabled an unprecedented 7600 km intercontinental key exchange via satellite relay~\cite{liao2017satellite,liao2018satellite}. These landmark demonstrations confirmed the robust operation of decoy-state BB84 protocols under fluctuating atmospheric conditions, while meeting the stringent size, weight, and power constraints of space-qualified hardware~\cite{hwang2003decoy,lo2005decoy,wang2005decoy,sidhu2021advances,liao2017satellite}. More recently, the Jinan-1 LEO mission has highlighted significant progress in payload miniaturization, paving the way for scalable quantum networks based on satellite constellations~\cite{li2024microsatellite}. In addition, advances in aggressive temporal, spatial, and spectral filtering have demonstrated that satellite QKD can operate even under broad daylight conditions~\cite{liao2017longdistance,avesani2021full,cai2024free}.

Building on the experience gained from LEO QKD, geostationary Earth orbit (GEO) quantum key distribution offers a qualitatively different approach to global quantum communications. LEO links benefit from comparatively lower channel loss and correspondingly high instantaneous key rates. However, they are inherently constrained by orbital dynamics: each ground station has only brief, few-minute satellite passes, creating challenges for rapid tracking, frequent constellation handovers, and stringent scheduling requirements~\cite{bourgoin2013comprehensive,sidhu2023finite,vasylyev2019satellite}. In contrast, positioning a satellite in GEO provides persistent visibility over roughly one-third of the Earth~\cite{liege2026analysis}. This stable geometry enables uninterrupted service windows and represents a compelling advantage for operational deployment. Consequently, GEO-QKD has become the focus of several major international initiatives. China, for instance, plans to launch the “Dawn” GEO mission around 2027, building on earlier quantum communication experiments associated with Shijian-20~\cite{pan2027dawn}. In Europe, feasibility studies of GEO-based architectures have been conducted~\cite{dirks2021geoqkd,klop2021qkd}, including proposals for hybrid networks combining GEO and LEO layers~\cite{qdesign2025}, with various missions currently under development~\cite{abad2022caramuel,garbo2024}.

The principal challenge for realizing GEO-QKD is the extreme channel attenuation, which pushes GEO transmissions into a very sparse-count regime~\cite{P21,pirandola2021limits,dirks2021geoqkd,gunthner2017geostat}. In this scenario, advanced statistical inference and precise parameter estimation are essential. Previous satellite QKD analyses relied on analytical decoy estimators combined with standard statistical bounds~\cite{lim2014concise,serfling1974probability,hoeffding1963probability}, which become prohibitively pessimistic under GEO-QKD conditions. By contrast, linear programming (LP)--based decoy estimation provides tighter predictions for single-photon statistics, enabling positive key rates in regimes where traditional analytical approaches fail~\cite{curty2014finite,zhang2017improved,attema2021optimizing}. Moreover, refined statistical tools substantially reduce finite-size penalties~\cite{bancal2022simple,mannalath2025sharp,mannalath2025elementary}.

% A parallel, critical enhancement targeting volatile satellite channels is the deployment of variable-length protocols. Fixed length operation commits to a target key length before transmission and trades acceptance probability against throughput~\cite{tomamichel2012tight,hayashi2012concise}. Conversely, variable-length protocols adapt the final block extraction to mirror the varying in-situ statistics. This minimizes protocol aborts and improves expected yields despite harsh fluctuations such as rapidly shifting turbulence, background radiance, and atmospheric transparency~\cite{tupkary2025imperfect,Tupkary2024,wang2025phase}.

A parallel and equally important enhancement for volatile satellite channels is the deployment of variable-length protocols. Fixed-length operation commits to a predetermined key length before transmission, trading acceptance probability against throughput~\cite{tomamichel2012tight,hayashi2012concise}. In contrast, variable-length protocols adapt the final block extraction to reflect instantaneous in-situ statistics. This dynamic adjustment minimizes protocol aborts and improves expected key yields, even under challenging conditions such as rapidly fluctuating turbulence, background radiance, and atmospheric transparency~\cite{tupkary2025imperfect,Tupkary2024,wang2025phase}.

% Targeting the impending integration of intercontinental quantum networks, this work introduces a comprehensive, physically realistic GEO-QKD feasibility analysis. 
% % Our underlying methodology combines three main elements: (i) incorporating optimal LP decoy-state estimation within variable-length finite-key extraction utilizing sharp statistical bounds; (ii) integrating a flexible receiver modeling framework that encompasses active, symmetric and asymmetric passive detection paradigms; and (iii) implementing a dynamic atmospheric channel-and-noise profile seamlessly factoring precise terrestrial geometry, turbulence, intrinsic propagation parameters, and sun-position-coupled background via radiative transfer and sky-background tools~\cite{emde2016libradtran,mayer2005libradtran,noll2013cerro,jones2013advanced}. 
% We conduct systematic performance explorations capturing optical wavelengths, telescope apertures, adaptive optics (AO) systems, spectral filtering, and representative  detector configurations (Si APDs, InGaAs APDs, and SNSPDs), and perform explicit annual key-volume forecasting governed by empirical cloud-availability parameters from publicly available datasets~\cite{tzallas2022craas}. The same framework can also be extended to LEO scenarios, although that would require explicitly accounting for satellite motion and the associated dynamical effects.

Targeting the imminent deployment of intercontinental quantum networks, this work presents a comprehensive, physically realistic feasibility analysis of GEO-QKD. We systematically explore system performance across optical wavelengths, telescope apertures, adaptive optics (AO) systems, spectral filtering, and representative detector configurations. Our study provides explicit annual key-volume forecasts, incorporating empirical cloud-availability data drawn from publicly available datasets~\cite{tzallas2022craas}.

% The  paper is organized as follows. Section~\ref{sec:protocol_def} introduces the protocol variants considered in this work, and Section~\ref{sec:security_proof} presents the finite-key decoy-state security analysis. Section~\ref{sec:channel} establishes the unified free-space propagation and noise model. Section~\ref{sec:params} summarizes the simulation assumptions, and Section~\ref{sec:performance} presents the resulting performance maps and annual key-volume forecasts. The appendices then collect the active-protocol derivations (Appendix~\ref{sec:protocol_appendix}), the asymmetric and symmetric passive-receiver analyses (Appendices~\ref{sec:bb84_passive_asym} and~\ref{sec:detection_appendix}), the statistical bounds (Appendix~\ref{sec:bounds}), the detailed channel model (Appendix~\ref{sec:channel_appendix}), and the expanded simulation set (Appendix~\ref{sec:appendix_plots}).

The paper is organized as follows. Section~\ref{sec:protocol_def} introduces the decoy-state QKD protocols considered in this work, covering the most relevant receiver architectures. Section~\ref{sec:security_proof} presents the finite-key decoy-state security analysis, while Section~\ref{sec:channel}  establishes a unified model for free-space propagation and noise. Section~\ref{sec:feas} details the simulation assumptions and presents the resulting performance maps and annual key-volume forecasts. Several appendices provide the full calculations and derivations necessary to reproduce the results.

%% file: journal/sections/main_protocol_definition.tex
\section{Protocol definition}
\label{sec:protocol_def}

We consider a variable-length decoy-state BB84 protocol with phase-randomized weak coherent pulses (PRWCPs). In each transmission round $r$, Alice selects a basis $a_r \in \{Z, X\}$ with probability $p^A_{a_r}$, a bit value $y_r \in \{0,1\}$ uniformly at random, and a signal intensity $k_r \in \mathcal{K}:= \{\mu_0, \dots, \mu_J\}$ with probability $p_{k_r}$. These states are sent through the quantum channel and measured by Bob in a basis $b_r \in \{Z, X\}$ with a quantum receiver. We consider three main receiver architectures for Bob: asymmetric active, asymmetric passive, and symmetric passive.

Following the quantum communication phase, Alice and Bob perform sifting, partitioning the detected events by basis and intensity. Then, they perform parameter estimation, calculate the variable-length key, and execute error correction, error verification, and privacy amplification to distill the final secure key of the calculated length. For a detailed definition of the protocol steps, we refer the reader to Appendix~\ref{sec:active_description} and Appendix~\ref{sec:asym_passive_description} depending on Bob's receiver architecture.

%% file: journal/sections/main_security_proof.tex
\section{Security proof}
\label{sec:security_proof}

To distill a secure key, we apply a variable-length decoy-state analysis that combines the active and passive finite-key treatments of Refs.~\cite{tupkary2025imperfect,wang2025phase} with tight statistical bounds to extend the achievable positive-key-rate regime. For a given primary key basis, a lower bound on the secret-key length that can be extracted is given by
\begin{equation}
    \label{eq:main_keylength}
    \ell =
    \left\lfloor n_{1}^{\mathrm{L}}\left[1 - h\!\left(\phi_1^{\mathrm{U}}\right)\right] - \lambda_{\mathrm{EC}} - \log_2\!\left(\frac{1}{2\epsilon_{\mathrm{PA}}^2\epsilon_{\mathrm{cor}}}\right) \right\rfloor,
\end{equation}
where $n_1^{\mathrm{L}}$ is a lower bound on the key basis detections originating from single-photon rounds, the term $\phi_1^{\mathrm{U}}$ represents an upper bound on the single-photon phase-error rate within these rounds and $h(x)= -x\log_2(x) - (1-x)\log_2(1-x)$ is the Shannon binary entropy function. The parameter $\lambda_{\mathrm{EC}}$ denotes the information leaked during error correction, $\epsilon_{\mathrm{cor}}$ denotes the correctness parameter of the protocol, and $\epsilon_{\mathrm{PA}}$ reflects the secrecy penalty associated with privacy amplification. The overall security parameter is $\epsilon_{\mathrm{tot}}$, with $\epsilon_{\mathrm{tot}}=2\sqrt{\epsilon_{\mathrm{PE}}}+\epsilon_{\mathrm{PA}}+\epsilon_{\mathrm{cor}}$.

The single-photon detection counts and the corresponding phase errors are bounded from the measured count statistics across the chosen decoy intensities $\mu_j$. To do so, we adopt the standard counterfactual decoy-state picture in which Alice is viewed as emitting photon-number states and the intensity labels are assigned a posteriori~\cite{lim2014concise}. This allows us to relate the observed intensity-resolved counts to the underlying photon-number contributions, which we estimate via linear programming~\cite{curty2014finite,zhang2017improved,attema2021optimizing}.
As standard, the core linear constraints relate the intensity-resolved measurement counts to the unknown photon-number counts:
\begin{equation}
    \label{eq:exp_detection_counts}
    \mathbb{E}[n_{\mu_j,\mathcal{B}}] = \sum_{i=0}^{\infty} p_{(\mu_j|i)} \, n_{i,\mathcal{B}},
\end{equation}
where  $\mathbb{E}[\cdot]$ denotes the expectation, $n_{\mu_j,\mathcal{B}}$ is the observed number of detection events in which both Alice and Bob select the basis $\mathcal{B} \in \{Z,X\}$ for pulses with intensity $\mu_j$, $p_{(\mu_j|i)}$ is the conditional probability that an $i$-photon emission is labeled by intensity $\mu_j$ in the counterfactual picture, and $n_{i,\mathcal{B}}$ is the corresponding basis-resolved detection count contributed by the $i$-photon component. Eq.~\eqref{eq:exp_detection_counts} is supplemented by physical and statistical constraints, including photon-number consistency, source statistics, and finite-size confidence intervals. Together, they provide a lower bound $n_{1}^{\mathrm{L}}$. Similar relations can be used to obtain $\phi_{1}^{\mathrm{U}}$.

Finite-size fluctuations are incorporated with sharp concentration inequalities: we bound the intensity count deviations using exact Poisson-Binomial bounds~\cite{bancal2022simple}, and the underlying photon-number events using exact Binomial inversions. In the asymmetric passive architecture, these Binomial bounds also control the vacuum and multi-photon leakage in the key-generation and test data~\cite{wang2025phase}.
Estimation of the phase-error term $\phi_1^{\mathrm{U}}$ relies on the complementary-basis single-photon statistics. For the random sampling bounds required to estimate $\phi_1^{\mathrm{U}}$, we resort to exact hypergeometric confidence intervals based on the Clopper-Pearson method~\cite{mannalath2025sharp,wright1991exact}.
In the asymmetric active and symmetric passive schemes, the standard random sampling argument between the $X$ and $Z$ basis single-photon counts applies directly. In the asymmetric passive scheme, because the effective basis assignment depends on Bob's received photon number, the hypergeometric sampling argument is applied only within the virtual single-photon subspace on Bob's side~\cite{wang2025phase}.

Comprehensive derivations of the exact variable-length bounds, LP constraints, and statistical bounds are detailed in Appendices~\ref{sec:protocol_appendix}, \ref{sec:bb84_passive_asym} and~\ref{sec:bounds}.

%% file: journal/sections/main_channel_model.tex
\section{Channel model}
\label{sec:channel}

Assessing the feasibility of a GEO satellite-to-ground QKD link requires a faithful channel model that captures the main physical mechanisms governing propagation loss and background noise.
This is provided in Appendix~\ref{sec:channel_appendix}. Here, we summarize the principal modeling assumptions.

In particular, we consider a spherical-Earth model and disregard small corrections associated with Earth oblateness and atmospheric refraction, which are negligible at this level. Signal attenuation is built from diffraction, finite receiver aperture, pointing jitter, coupling losses, and atmospheric extinction as well as other system losses. Specifically, diffraction is modeled assuming Gaussian-beam propagation and taking into account the transmitter aperture, launch waist, wavelength, and beam-quality factor~\cite{siegman1986lasers,siegman1993defining}. Turbulence is described with generalized Hufnagel--Valley profiles \cite{hufnagel1964modulation,valley1980isoplanatic,hardy1998adaptive,pugh2020adaptive}. Although this effect only weakly broadens a GEO downlink beam because most of the path is in vacuum, it remains important for wavefront quality and efficient collection into the receiver, especially for single-mode-fiber operation. We model adaptive-optics correction through its improvement of the effective coherence length and Strehl ratio \cite{lanning2021quantum}. Atmospheric absorption and scattering are modeled with \texttt{libRadtran} using the 1976 U.S. Standard Atmosphere together with Shettle--Fenn aerosol models \cite{emde2016libradtran,mayer2005libradtran,united1976us,shettle1979models,shettle1990models}. The particular operating scenarios and numerical values used in the simulations are summarized in Section~\ref{sec:params} and Table~\ref{tab:sim_params2}.

Cloud losses are incorporated using the location-specific CRAAS climatology derived from the CLAAS-2.1 record over Europe. This record provides multi-year mean occurrence fractions for cloud optical thickness (COT) categories on a $1^\circ \times 1^\circ$ grid with 15-minute sampling \cite{tzallas2022craas,EUMETSAT93:online}. In our analysis we group conditions into clear-sky, thin-cloud, and thick-cloud transmission regimes, assign representative COT values, and pass the data to \texttt{libRadtran} to compute the additional zenith-angle-dependent attenuation. The resulting key rates are then weighted by the local occurrence fractions. Higher-COT categories are omitted because their attenuation is too large to support useful GEO-QKD performance. We refer the reader to Appendix~\ref{sec:clouds} for further details.

With these definitions, the overall system transmission can be written as
\begin{equation}
    \eta_{\mathrm{sys}}=\eta_{\mathrm{geo}}\eta_{\mathrm{p}}\eta_{\mathrm{cpl}}\eta_{\mathrm{atm}}\eta_{\mathrm{R}}\eta_{\mathrm{D}},
\end{equation}
where the factors describe geometric collection, pointing, receiver coupling, atmospheric transmission, receiver optics/filter transmission, and detector efficiency, respectively.

Background noise is a key part of the channel model. The mean number of background photons collected within the receiver acceptance window is denoted by $\bar{n}_{\mathrm{B}}$. This quantity is determined by the sky spectral radiance together with the effective geometric-spectral-temporal acceptance factor $\Gamma_{\mathrm{R}}$ (defined explicitly in Appendix~\ref{sec:background_noise}). This background level is mapped to the per-detector per-gate noise-click probability $p_{\mathrm{noise}}$ used in the finite-key analysis. For daytime conditions, we use \texttt{libRadtran} and consider multiple noise scenarios defined by the position of the Sun, while nighttime sky radiance is obtained from the ESO SkyCalc implementation of the Cerro Paranal Advanced Sky Model \cite{ESO_SkyCalc_2025,noll2013cerro,jones2013advanced}. This background model converts the observing geometry, illumination conditions, and receiver acceptance into the mean background-photon level that enters the finite-key performance analysis.

In Section~\ref{sec:feas} we use these models to evaluate the resulting GEO-QKD feasibility across the operating scenarios considered.

%% file: journal/sections/main_plots.tex
\section{Feasibility evaluation}
\label{sec:feas}

\subsection{Simulation assumptions and operating scenarios}
\label{sec:params}

Here we present the simulation settings for performance maps, with Table~\ref{tab:sim_params2} collecting the shared baseline values and parameter sweeps. Unless a parameter is explicitly swept or otherwise specified for a given plot, we use the boldfaced baseline entry from Table~\ref{tab:sim_params2}.

\subsubsection{Geometry and protocol scale}
We consider a GEO satellite at altitude $h_{\mathrm{sat}}=35{,}786$ km and longitude $5^\circ$E, viewed over zenith angles $\theta\in[0^\circ,80^\circ]$, spanning the European area.

The QKD protocol is evaluated in the finite-key regime with source rates $f\in\{10^8,10^9\}$ Hz,  $N\in\{10^{11},10^{12}\}$ transmitted pulses, a total security parameter $\epsilon_{\mathrm{tot}}=10^{-8}$, and an error-correction efficiency factor of $f_{\mathrm{EC}}=1.16$~\cite{sidhu2022finite}.

\subsubsection{Optical hardware and channel settings}
The optical design space covers transmitter apertures $a_T\in[0.2,1.0]$ m and receiver apertures $a_R\in[0.4,2.0]$ m, with launch waist set to $\omega_0=a_T/4$ and a beam-quality factor $M^2=1.2$, consistent with near-diffraction-limited commercial sources summarized in Table~\ref{tab:M2}. Pointing jitter is examined over representative low, moderate, and high cases of $0.5$, $1.0$, and $2.0\,\mu\mathrm{rad}$. We compare no-AO and finite-bandwidth AO configurations~\cite{lanning2021quantum}, with closed-loop bandwidths $f_c\in\{130,200,500\}$ Hz, fixed tip--tilt bandwidth $f_{tc}=60$ Hz, and a turbulent layer extending to an altitude of $h_{\mathrm{atm}}=20$ km.

\subsubsection{Locations, wavelengths, and noise regimes}
We evaluate three representative optical-ground-station (OGS) environments: coastal, urban and rural, with the site-dependent atmospheric and turbulence parameters listed in Table~\ref{tab:location_params}. Across these scenarios we study three signal wavelengths $\{656.448,\,854.445,\,1550.027\}$ nm under both daytime and nighttime operation. Daytime cases cover low, moderate and high solar-background conditions, while nighttime operation uses a full-Moon background model; the corresponding spectral-radiance values are collected in Table~\ref{tab:combined_radiance}. Daytime and nighttime background cases are tied to Figs.~\ref{fig:sky_radiance} and \ref{fig:lunar_radiance}, respectively. Furthermore, the wavelength-dependent transmission and radiance trends are shown in Fig.~\ref{fig:radandtrans}.

\subsubsection{Detector models}
The QKD receiver model includes superconducting nanowire single-photon detectors (SNSPDs) together with wavelength-appropriate avalanche photodiodes (APDs), namely Si APDs in the visible and near infrared and InGaAs/InP APDs at telecom wavelengths. For each detector family we consider two performance tiers, denoted by Spec A and Spec B, with wavelength-dependent efficiency, dark-count rate, and afterpulsing values listed in Table~\ref{tab:detector_params}. Detection is performed with a gate width of $\Delta t=1$ ns together with spectral filters in the 1--30 GHz range, with filter insertion loss included in the receiver optics budget.

\begin{table*}[t]
    \centering
    \caption{List of constants and parameters used in this work. For those quantities that vary across simulations, the default value used, unless otherwise specified, is shown in bold. In the table: SMF, single-mode fiber; FS, free space; FOV, field of view; AO, adaptive optics.}
    \label{tab:sim_params2}
    \small
    \begin{tabularx}{\textwidth}{l c X l}
        \toprule
        \textbf{Parameter}            & \textbf{Symbol}           & \textbf{Value(s)}                                                 & \textbf{Unit}     \\
        \midrule
        \multicolumn{4}{c}{\textit{Physical Constants}}                                                                                                   \\
        \midrule
        Speed of light in vacuum      & $c$                       & $2.998 \times 10^8$                                               & m\,s$^{-1}$       \\
        Reduced Planck constant       & $\hbar$                   & $1.055 \times 10^{-34}$                                           & J\,s              \\
        Earth radius                  & $R_E$                     & 6378                                                              & km                \\
        \midrule
        \multicolumn{4}{c}{\textit{Orbital Geometry}}                                                                                                     \\
        \midrule
        GEO satellite altitude        & $h_{\mathrm{sat}}$        & 35,786                                                            & km                \\
        Satellite longitude           & $\lambda_{\mathrm{sat}}$  & 5$^{\circ}$E                                                      & --                \\
        OGS height above ground level & $h_{0}$                   & 10                                                                & m                 \\
        Viewing zenith angle          & $\theta$                  & [0,80], \textbf{60}                                               & degrees           \\
        % OGS altitude above sea level & $h_{\rm OGS}$ & 0 (Coastal), \textbf{0.2} (Urban), 0.4 (Rural) & km \\
        % Zenith angle & $\theta$ & 0--80, \textbf{45} & degrees \\
        % Slant range (at $\theta=0^{\circ}$) & $\rho$ & $\sim$35,800--42,000 & km \\
        \midrule
        \multicolumn{4}{c}{\textit{Protocol Parameters}}                                                                                                  \\
        \midrule
        Signal wavelength             & $\lambda$                 & 656.448, \textbf{854.445}, 1550.027                               & nm                \\
        % Number of transmission rounds & $R$ & Protocol input & -- \\
        Source repetition rate        & $f$                       & $10^8$, $\mathbf{10^9}$                                           & Hz                \\
        Number of transmission rounds & $N$                       & $10^{11}$, $\mathbf{10^{12}}$                                     & pulses            \\
        Receiver setup                & --                        & Asymmetric active, \textbf{asymmetric passive}, symmetric passive & --                \\
        % Key basis probability & $q_X$ & 0.5 (passive), optimized (active) & -- \\
        % Number of decoy intensities & $J$ & 2 (3-intensity protocol) & -- \\
        % LP truncation photon number & $I$ & 10 & -- \\
        Security parameter            & $\epsilon_{\mathrm{tot}}$ & \textbf{$10^{-8}$}                                                & --                \\
        % PE failure probability & $\epsilon_{\rm PE}$ & $\epsilon_{\rm tot}^2/4$ & -- \\
        % EC failure probability & $\epsilon_{\rm cor}$ & $10^{-3}\epsilon_{\rm tot}$ & -- \\
        % PA failure probability & $\epsilon_{\rm PA}$ & $10^{-3}\epsilon_{\rm tot}$ & -- \\
        Error correction efficiency   & $f_{\mathrm{EC}}$         & 1.16                                                              & --                \\
        Misalignment error rate       & $e_{\mathrm{mis}}$        & \textbf{0.5} (Low), 1.0 (Mod.), 2.0 (High)                        & \%                \\
        \midrule
        \multicolumn{4}{c}{\textit{Transmitter Optics}}                                                                                                   \\
        \midrule
        Transmitter aperture          & $a_T$                     & [0.2,1], \textbf{0.75}                                            & m                 \\
        Beam waist                    & $\omega_0$                & $a_T/4$                                                           & m                 \\
        Beam quality factor           & $M^2$                     & 1.2                                                               & --                \\
        Pointing jitter (angular)     & --                        & \textbf{0.5} (Low), 1.0 (Mod.), 2.0 (High)                        & $\mu\mathrm{rad}$ \\
        \midrule
        \multicolumn{4}{c}{\textit{Atmospheric Channel}}                                                                                                  \\
        \midrule
        Location                      & --                        & Rural, \textbf{Urban}, Coastal                                    & --                \\
        Operation time                & --                        & Day (Low, Moderate, High), \textbf{Night}                         & --                \\
        Atmospheric model             & --                        & U.S. Standard Atmosphere 1976                                     & --                \\
        Top of turbulent layer        & $h_{\mathrm{atm}}$        & 20                                                                & km                \\
        % Surface albedo & -- & 0.2 & -- \\
        % \multicolumn{4}{l}{\quad\textit{Location-dependent parameters:}} \\
        % \quad Visibility (Coastal) & VIS & 5 & km \\
        % \quad Visibility (Urban) & VIS & 10 & km \\
        % \quad Visibility (Rural) & VIS & 23 & km \\
        % \quad Aerosol type & -- & Maritime (Coastal), Urban, Rural & -- \\
        % \multicolumn{4}{l}{\quad\textit{Hufnagel--Valley turbulence coefficients:}} \\
        % \quad Surface layer (Coastal) & $A$ & $5.1 \times 10^{-14}$ & m$^{-2/3}$ \\
        % \quad Surface layer (Urban) & $A$ & $1.7 \times 10^{-14}$ & m$^{-2/3}$ \\
        % \quad Surface layer (Rural) & $A$ & $4.5 \times 10^{-15}$ & m$^{-2/3}$ \\
        % \quad Troposphere (Coastal) & $B$ & $8.1 \times 10^{-16}$ & m$^{-2/3}$ \\
        % \quad Troposphere (Urban) & $B$ & $2.7 \times 10^{-16}$ & m$^{-2/3}$ \\
        % \quad Troposphere (Rural) & $B$ & $9.0 \times 10^{-17}$ & m$^{-2/3}$ \\
        % \quad Tropopause (Coastal) & $C$ & $1.08 \times 10^{-52}$ & m$^{-32/3}$ \\
        % \quad Tropopause (Urban) & $C$ & $3.6 \times 10^{-53}$ & m$^{-32/3}$ \\
        % \quad Tropopause (Rural) & $C$ & $2.0 \times 10^{-53}$ & m$^{-32/3}$ \\
        % \quad Cloud optical thickness & $\tau_{\rm cloud}$ & 0 (Clear), 0.65 (Thin), 2.45 (Thick) & -- \\
        \midrule
        \multicolumn{4}{c}{\textit{Receiver Optics}}                                                                                                      \\
        \midrule
        Receiver aperture             & $a_R$                     & $[0.4,2]$, \textbf{1.5}                                           & m                 \\
        Receiver optics loss          & --                        & 4                                                                 & dB                \\
        Coupling strategy             & --                        & \textbf{SMF}, FS                                                  & --                \\
        SMF intrinsic coupling limit  & $\eta_0$                  & 0.786                                                             & --                \\
        FS intrinsic coupling limit   & $\eta_0$                  & 0.838                                                             & --                \\
        Airy disk FOV half-angle      & $\gamma_{\mathrm{Airy}}$  & $1.22\lambda/a_R$                                                 & rad               \\
        SMF effective FOV             & $\gamma_{\mathrm{SMF}}$   & $2.24\lambda/(\pi a_R)$                                           & rad               \\
        \midrule
        \multicolumn{4}{c}{\textit{Adaptive Optics}}                                                                                                      \\
        \midrule
        AO closed-loop bandwidth      & $f_c$                     & 0 (None), \textbf{130 (Low)}, 200 (Mod.), 500 (High)              & Hz                \\
        Tip-tilt correction bandwidth & $f_{tc}$                  & 60                                                                & Hz                \\
        % OGS wind speed & $v_{\rm OGS}$ & 5 (Rural), 10 (Urban), 25 (Coastal) & km/h \\
        \midrule
        \multicolumn{4}{c}{\textit{Spectral Filtering}}                                                                                                   \\
        \midrule
        Filter bandwidth              & $\Delta\nu$               & \textbf{1} (1550.027 nm), \textbf{10} (656.448/854.445 nm), 30    & GHz               \\
        Filter loss                   & --                        & 1 (included in receiver loss)                                     & dB                \\
        \midrule
        \multicolumn{4}{c}{\textit{Detectors}}                                                                                                            \\
        \midrule
        Detector type                 & --                        & \textbf{SNSPD}, APD                                               & --                \\
        Specification                 & --                        & Spec A, \textbf{Spec B}                                           & --                \\
        % \midrule
        % \multicolumn{4}{c}{\textit{Single-Photon Detectors (SNSPD)}} \\
        % \midrule
        % Detection efficiency (Spec A / B) & $\eta_D$ & 0.85 / 0.95 & -- \\
        % Dark count rate (Spec A / B) & DCR & 10 / 0.1 & Hz \\
        % Afterpulsing probability & $p_{\rm ap}$ & 0 & \% \\
        % \midrule
        % \multicolumn{4}{c}{\textit{Single-Photon Detectors (Si APD, 656/854 nm)}} \\
        % \midrule
        % Detection efficiency (Spec A / B) & $\eta_D$ & 0.50/0.70 (@656), 0.40/0.50 (@854) & -- \\
        % Dark count rate (Spec A / B) & DCR & 100 / 10 & Hz \\
        % Afterpulsing probability (Spec A / B) & $p_{\rm ap}$ & 2.0 / 0.5 & \% \\
        % \midrule
        % \multicolumn{4}{c}{\textit{Single-Photon Detectors (InGaAs APD, 1550 nm)}} \\
        % \midrule
        % Detection efficiency (Spec A / B) & $\eta_D$ & 0.25 / 0.35 & -- \\
        % Dark count rate (Spec A / B) & DCR & 300 / 50 & Hz \\
        % Afterpulsing probability (Spec A / B) & $p_{\rm ap}$ & 3.0 / 1.0 & \% \\
        \midrule
        \multicolumn{4}{c}{\textit{Background Noise}}                                                                                                     \\
        \midrule
        Detection gate width          & $\Delta t$                & 1                                                                 & ns                \\
        \midrule
        % Solar zenith angles & $z_{\rm sun}$ & 0$^{\circ}$ (High), 60$^{\circ}$ (Mod.), 90$^{\circ}$ (Low) & degrees \\
        % Lunar zenith angle (full Moon) & $z_{\rm moon}$ & 30 & degrees \\
        \bottomrule
    \end{tabularx}
\end{table*}

\begin{table*}[t]
    \centering
    \caption{Location parameters for downlink GEO-QKD feasibility analysis. The turbulence parameters $A,B,C$ are defined in Appendix~\ref{sec:turb}.}
    \begin{tabular}{lccc r @{ = } r @{\ } l}
        \hline
        \textbf{Location}        & \textbf{Altitude [km]} & \textbf{Visibility [km]} & \textbf{Ground Wind Speed [m/s]} & \multicolumn{3}{c}{\textbf{Turbulence Parameters}}                                        \\
        \hline
        \multirow{3}{*}{Rural}   & \multirow{3}{*}{0.4}   & \multirow{3}{*}{23}      & \multirow{3}{*}{5}               & $A$                                                & $4.5 \times 10^{-15}$  & m$^{-2/3}$  \\
                                 &                        &                          &                                  & $B$                                                & $9.0 \times 10^{-17}$  & m$^{-2/3}$  \\
                                 &                        &                          &                                  & $C$                                                & $2.0 \times 10^{-53}$  & m$^{-32/3}$ \\
        \hline
        \multirow{3}{*}{Urban}   & \multirow{3}{*}{0.2}   & \multirow{3}{*}{10}      & \multirow{3}{*}{10}              & $A$                                                & $1.7 \times 10^{-14}$  & m$^{-2/3}$  \\
                                 &                        &                          &                                  & $B$                                                & $2.7 \times 10^{-16}$  & m$^{-2/3}$  \\
                                 &                        &                          &                                  & $C$                                                & $3.6 \times 10^{-53}$  & m$^{-32/3}$ \\
        \hline
        \multirow{3}{*}{Coastal} & \multirow{3}{*}{0}     & \multirow{3}{*}{5}       & \multirow{3}{*}{25}              & $A$                                                & $5.1\times 10^{-14}$   & m$^{-2/3}$  \\
                                 &                        &                          &                                  & $B$                                                & $8.1\times 10^{-16}$   & m$^{-2/3}$  \\
                                 &                        &                          &                                  & $C$                                                & $1.08 \times 10^{-52}$ & m$^{-32/3}$ \\
        \hline
    \end{tabular}
    \label{tab:location_params}
\end{table*}

\begin{table*}[t]
    \centering
    \caption{Daytime (solar) and nighttime (lunar) spectral radiance [mW\,m$^{-2}$\,nm$^{-1}$\,sr$^{-1}$] for each location and wavelength, evaluated at various solar ($0^{\circ}$, $60^{\circ}$, $90^{\circ}$) and lunar ($30^{\circ}$) zenith angles. See Appendix~\ref{sec:background_noise} for more details on the background model. }
    \label{tab:combined_radiance}
    \renewcommand{\arraystretch}{1.2}
    \setlength{\tabcolsep}{4pt}

    \begin{tabular}{|c|c||ccc|c|}
        \hline
        \multirow{3}{*}{\textbf{Location}}        &
        \multirow{3}{*}{\textbf{Wavelength (nm)}} &
        \multicolumn{4}{c|}{\textbf{Radiance [mW\,m$^{-2}$\,nm$^{-1}$\,sr$^{-1}$]}}                                                                \\
        \cline{3-6}
                                                  &          &
        \multicolumn{3}{c|}{\textbf{Solar}}       &
        \textbf{Lunar}                                                                                                                             \\
        \cline{3-6}
                                                  &          & \textbf{0° (High)} & \textbf{60° (Mod.)} & \textbf{90° (Low)} & \textbf{30°}        \\
        \hline
        \multirow{3}{*}{Rural}
                                                  & 656.448  & 18.98              & 11.94               & 0.80               & $4.39\times10^{-4}$ \\
                                                  & 854.445  & 6.40               & 4.05                & 0.42               & $1.91\times10^{-4}$ \\
                                                  & 1550.027 & 2.62               & 1.60                & 0.32               & $4.42\times10^{-5}$ \\
        \hline
        \multirow{3}{*}{Urban}
                                                  & 656.448  & 18.62              & 10.40               & 0.50               & $4.52\times10^{-4}$ \\
                                                  & 854.445  & 7.23               & 4.18                & 0.30               & $2.24\times10^{-4}$ \\
                                                  & 1550.027 & 3.90               & 2.29                & 0.24               & $6.99\times10^{-5}$ \\
        \hline
        \multirow{3}{*}{Coastal}
                                                  & 656.448  & 49.92              & 27.22               & 0.69               & $9.44\times10^{-4}$ \\
                                                  & 854.445  & 21.88              & 12.72               & 0.40               & $5.45\times10^{-4}$ \\
                                                  & 1550.027 & 21.22              & 12.90               & 0.39               & $3.34\times10^{-4}$ \\
        \hline
    \end{tabular}
\end{table*}

\begin{table*}[htbp]
    \centering
    \caption{Simulation parameters for single-photon detectors. The specifications are divided into two tiers: Spec A (Standard/Commercial) and Spec B (High-Performance/Research). Parameters listed are detection efficiency ($\eta_{\mathrm{D}}$), dark count rate (DCR), and afterpulsing probability ($p_{\mathrm{ap}}$). The adopted values are chosen based on the detector literature survey summarized in Appendix~\ref{sec:channel_appendix}, Table~\ref{tab:masteRAT_pd}.}
    \label{tab:detector_params}
    \renewcommand{\arraystretch}{1.3}
    \begin{tabular}{lc|ccc|ccc}
        \hline\hline
        \multirow{2}{*}{\textbf{Detector}} & \multirow{2}{*}{\textbf{$\lambda$ [nm]}} & \multicolumn{3}{c|}{\textbf{Spec A}} & \multicolumn{3}{c}{\textbf{Spec B}}                                                                                                                          \\ \cline{3-8}
                                           &                                          & $\boldsymbol{\eta_{\mathrm{D}}}$     & \textbf{DCR [Hz]}                   & $\boldsymbol{p_{\mathrm{ap}}}$ & $\boldsymbol{\eta_{\mathrm{D}}}$ & \textbf{DCR [Hz]} & $\boldsymbol{p_{\mathrm{ap}}}$ \\ \hline
        \multirow{3}{*}{SNSPD}             & 656                                      & 0.85                                 & 10                                  & 0.0                            & 0.95                             & 0.1               & 0.0                            \\
                                           & 854                                      & 0.85                                 & 10                                  & 0.0                            & 0.95                             & 0.1               & 0.0                            \\
                                           & 1550                                     & 0.85                                 & 10                                  & 0.0                            & 0.95                             & 0.1               & 0.0                            \\ \hline
        \multirow{2}{*}{Si APD}            & 656                                      & 0.50                                 & 100                                 & 0.02                           & 0.70                             & 10                & 0.005                          \\
                                           & 854                                      & 0.40                                 & 100                                 & 0.02                           & 0.50                             & 10                & 0.005                          \\ \hline
        InGaAs/InP APD                     & 1550                                     & 0.25                                 & 300                                 & 0.03                           & 0.35                             & 50                & 0.01                           \\ \hline\hline
    \end{tabular}
\end{table*}

\subsection{Simulations}
\label{sec:performance}
Using the source, detector, and channel models defined above, we now translate the simulation parameter set into a compact sequence of feasibility maps. We start by plotting the secret-key rate in the loss-noise plane, then connect the abstract coordinates to sky background and aperture-dependent system loss, and finally project the same constraints onto zenith angle, geography, and long-distance comparisons. At each point of every sweep, the corresponding finite-size secret-key rate is optimized over the relevant protocol parameters. Numerically, each point is obtained by grid search combined with numerical maximization over the admissible parameter region, with neighboring solutions reused as warm starts where appropriate. To keep the main text concise, we show only a representative subset of the full multidimensional parameter space. Expanded sweeps and auxiliary figures are deferred to Appendix~\ref{sec:appendix_plots}.\\

\noindent\textit{Protocol-level feasibility map}

Figure~\ref{fig:loss_noise_rate_grid} provides the protocol-level feasibility map that underlies the rest of the section.
For this asymmetric active BB84 receiver reference plot, the optimized secret-key rate is shown as a function of the two effective channel coordinates that matter most for the finite-key analysis, namely the total system loss and the per-detector per-gate noise-click probability $p_{\mathrm{noise}}$, while the four panels span representative values of misalignment and afterpulsing. The figure can be read as a compact tolerance chart: for any given background level it shows how much loss can be tolerated, and for any given loss it shows how much the background noise must be suppressed.

The contour structure shows that the positive-key region retreats diagonally toward lower loss and lower noise as misalignment and afterpulsing increase, and the high-rate corner shrinks markedly in the most imperfect case. Misalignment and afterpulsing reduce the achievable key rate and markedly remove operating margin in the loss-noise regime. The same map clarifies the practical difference between night and day operation. Nighttime links occupy lower values of $p_{\mathrm{noise}}$, allowing positive key rates at substantially larger losses, whereas daytime links are displaced upward by the higher sky background and therefore require a correspondingly stronger optical channel through larger apertures, better pointing, or stronger turbulence correction. All later figures can be read as physical mechanisms that move the operating point within this plane or determine whether a given region of the plane is accessible at all.\\

\begin{figure}[htbp]
    \includegraphics[width=\linewidth]{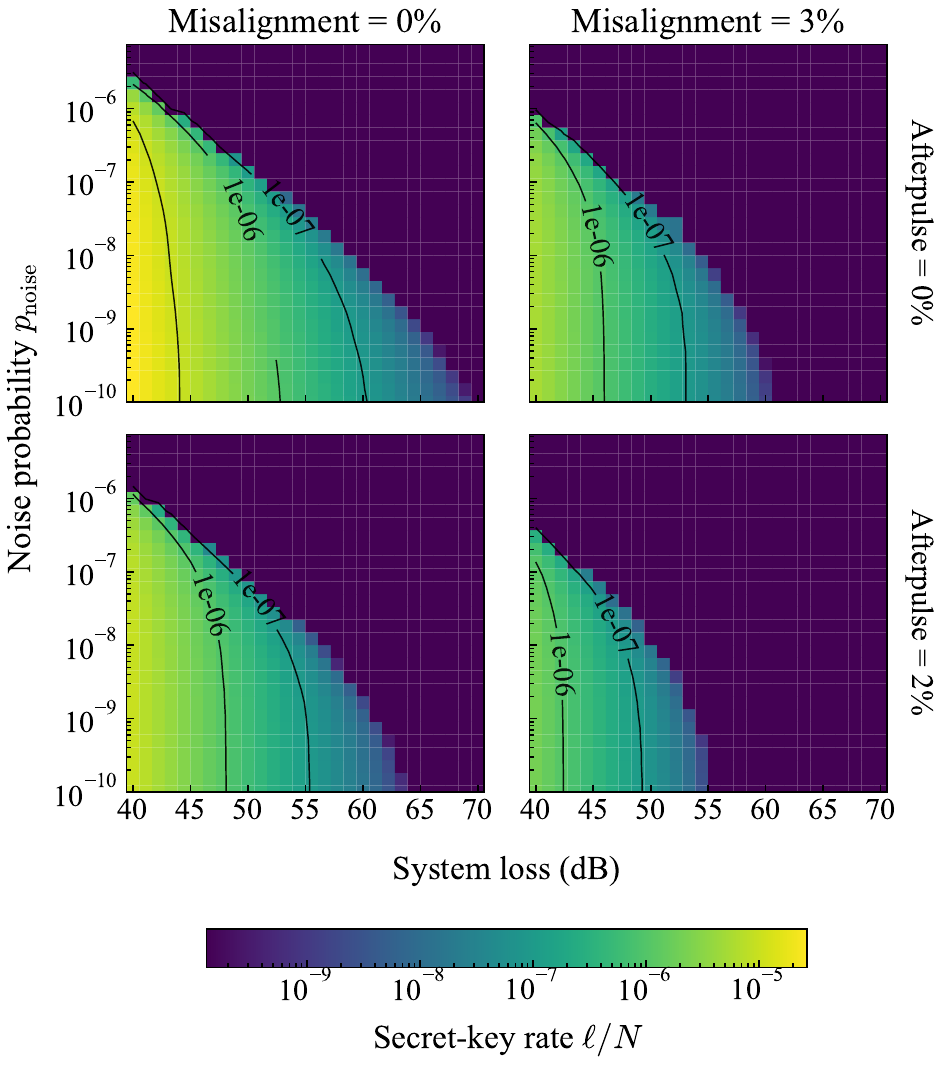}
    \caption{Finite-size secret-key rate $\ell/N$ for an asymmetric active BB84 receiver as a function of total system loss (x axis) and noise-click probability $p_{\mathrm{noise}}$ (y axis). The color map and black contour lines show the finite-size secret-key rate. Columns correspond to misalignment values of $0\%$ and $3\%$, while rows correspond to afterpulsing probabilities of $0\%$ and $2\%$.}
    \label{fig:loss_noise_rate_grid}
\end{figure}

\noindent\textit{Sky radiance, aperture loss and operating points}

Having established the protocol-level noise tolerance, Fig.~\ref{fig:sky_radiance_contours} translates that tolerance into an observational constraint set by diffuse sky radiance in a solar-position map.
The ground station is at the center and the concentric circles indicate evenly spaced zenith angles separated by 30 degrees, with the satellite line of sight fixed at 60 degrees toward the south. Given this configuration, each point in the sky map represents the radiance received by the ground station along the fixed satellite line of sight when the Sun is located at that sky position.
The maps are evaluated for a representative GEO downlink geometry with the satellite line of sight fixed at a zenith angle of $60^\circ$ toward the south. The color scale shows the spectral sky radiance $\tilde{H}_{\lambda}^{\mathrm{sky}}$ across the visible sky for rural and urban scenarios. Through the background model summarized in Appendix~\ref{sec:background_noise}, this radiance is converted to a photon-flux density $H_{\lambda}^{\mathrm{sky}}$, which sets the mean collected background photons $\bar{n}_{\mathrm{B}}=H_{\lambda}^{\mathrm{sky}}\Gamma_{\mathrm{R}}$, with $\Gamma_{\mathrm{R}}$ being the aforementioned effective geometric-spectral-temporal acceptance factor. For the asymmetric active receiver, these parameters determine the per-detector per-gate noise-click probability through $p_{\mathrm{noise}}=1-\exp(-\eta_{\mathrm{R}}\eta_{\mathrm{D}}\bar{n}_{\mathrm{B}}/2)(1-p_{\mathrm{dark}})$, where $p_{\mathrm{dark}}$ is the intrinsic detector dark-count probability. The overlaid contours mark the corresponding critical radiance threshold $\tilde{H}_{\mathrm{crit}}$, namely the sky radiance for which the background alone reaches a prescribed noise level. In the context of Fig.~\ref{fig:loss_noise_rate_grid}, this figure controls how far upward a daytime operating point is pushed along the noise axis.

We define such a critical condition through a background-click probability of $p_{\mathrm{noise}}=10^{-6}$, which is already sufficient to preclude secure-key generation for system losses of about 45 dB or higher. To obtain a simple threshold estimate, we use a diffraction-limited approximation for the receiver field of view, $\Omega_{\mathrm{FOV}}\approx \pi(\lambda/\pi a_R)^2$. With this approximation, for a chosen wavelength and filter bandwidth, the contour $\tilde{H}_{\mathrm{crit}}$ can be interpreted as the maximum tolerable sky radiance before the background alone drives the link beyond the feasible region. The contour lines therefore act as exclusion boundaries: if the background is brighter than the contour associated with a given filter bandwidth, that filter is too broad to support secret-key generation. Narrower filters move these contours inward and thereby enlarge the accessible operating window outside the contour.

\begin{figure}[htbp]
    \includegraphics[width=\linewidth]{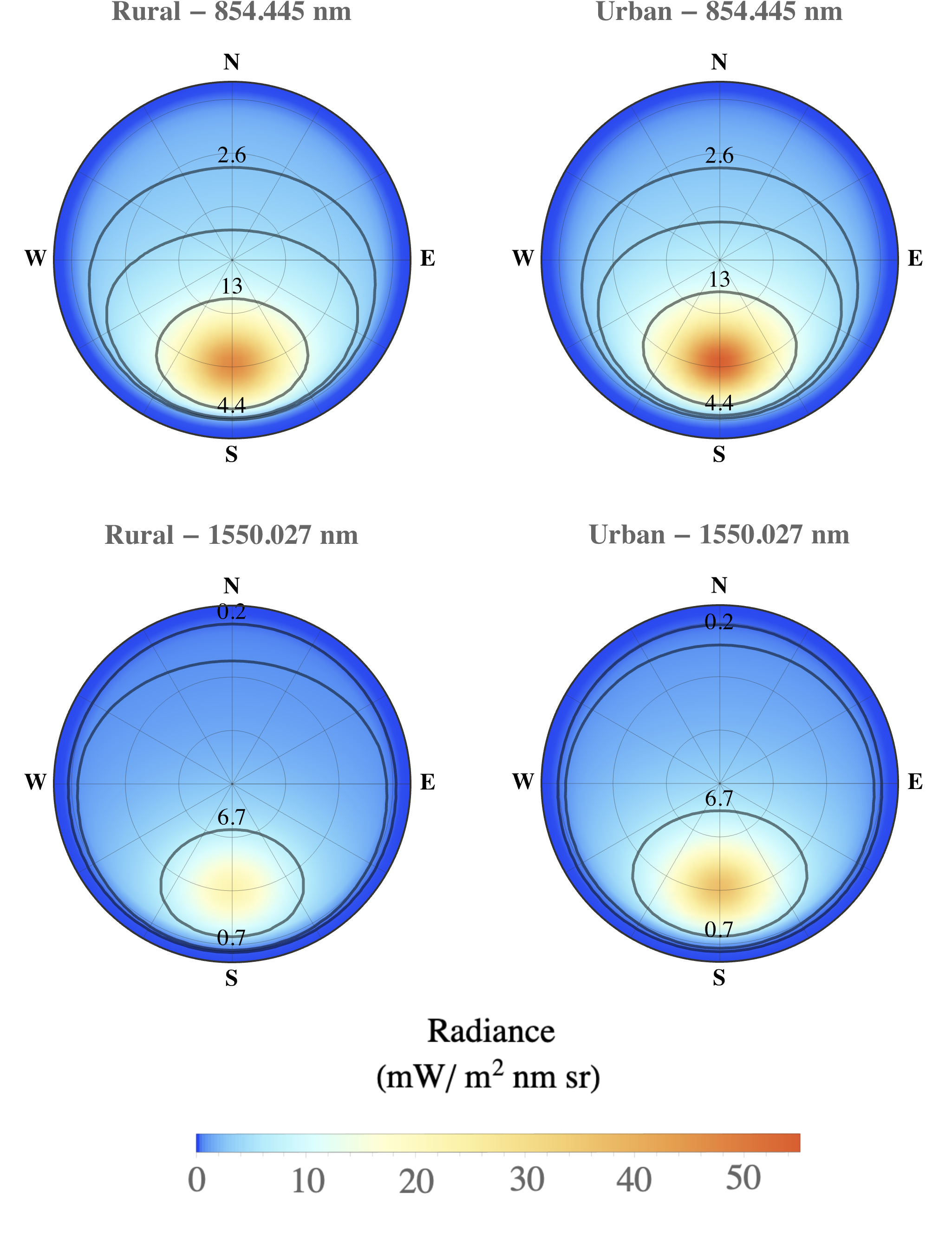}
    \caption{Sky-radiance solar position maps ($\tilde{H}_\lambda^{\mathrm{sky}}$) for rural and urban environments (columns) at the two signal wavelengths shown in the rows, $\lambda=854.445$ nm and $\lambda=1550.027$ nm, for a representative GEO line of sight at zenith angle $60^\circ$ toward the south. The color scale indicates radiance in $\mathrm{mW/ m^{2}\,nm^{1}\,sr^{1}}$. Superimposed contours represent the critical radiance thresholds ($\tilde{H}_{\mathrm{crit}}$) corresponding to a  noise-click probability of $p_{\mathrm{noise}}=10^{-6}$ for representative spectral filter bandwidths. For 854\,nm, the contours correspond to bandwidths of 10, 30, and 50\,GHz (inner to outer). For 1550\,nm, they correspond to 1, 10, and 30\,GHz (inner to outer).}
    \label{fig:sky_radiance_contours}
\end{figure}

Figure~\ref{fig:aperture_loss_contour} isolates the receiver and transmitter-aperture dependence of the deterministic system loss on a fixed urban baseline slice ($\theta=60^\circ$, low pointing jitter, and single-mode fiber (SMF) coupling) for the default SNSPD Spec B receiver. It shows that increasing receiver aperture only becomes effective once the turbulence penalty is corrected with AO. In the no-AO column, the contours are nearly horizontal, so the loss is governed mainly by the transmitter aperture $a_T$: increasing $a_T$ still reduces diffraction loss, but increasing the receiver aperture $a_R$ produces little net change.
This weak receiver-aperture dependence is a near-cancellation built into the model, not an absence of geometric gain. From Eqs.~\eqref{eq:eta_smf} and \eqref{eq:strehl}, the no-AO SMF coupling is $\eta_{\mathrm{cpl}}^{\mathrm{SMF}}=\eta_0[1+(a_R/r_0)^{5/3}]^{-6/5}$, with $r_0$ fixed by the atmosphere for a chosen wavelength and zenith angle. Increasing $a_R$ improves geometric collection while simultaneously worsening the Strehl-limited coupling into the fiber. Over the plotted 854~nm slice, for example, increasing $a_R$ from 0.4 to 2.0~m at fixed $a_T\approx0.75$~m reduces the geometric loss by about 14~dB, but increases the coupling loss by about 13.7~dB, leaving the total loss almost unchanged.
With AO, this cancellation is largely removed. Equation~\eqref{eq:r0cl} makes the corrected coherence length grow approximately in proportion to $a_R$, so the aperture-to-coherence ratio entering the Strehl factor stays nearly constant. In this scenario, the coupling penalty no longer grows markedly with receiver size, and the geometric benefit of a larger receiver reappears as a genuine reduction in total loss. In the loss-noise picture, Fig.~\ref{fig:aperture_loss_contour} mainly explains horizontal motion: aperture and AO determine how far left a configuration can be pushed along the system-loss axis.

Viewed in relation to the feasibility map of Fig.~\ref{fig:loss_noise_rate_grid}, this AO dependence has a direct protocol consequence. In the no-AO column, much of the aperture space remains at losses that are prohibitive once realistic daytime background is included. Strong AO shifts the same configurations into a substantially lower-loss regime, making daytime operation plausible over a much larger part of the design space. Figure~\ref{fig:aperture_loss_contour} therefore highlights AO as a prerequisite for practical GEO-QKD in moderately noisy daytime conditions. The remaining difference between rows is then mainly diffraction-driven, so the 1550~nm case retains a stronger dependence on transmitter aperture than the 854~nm case.

\begin{figure}[htbp]
    \includegraphics[width=\linewidth]{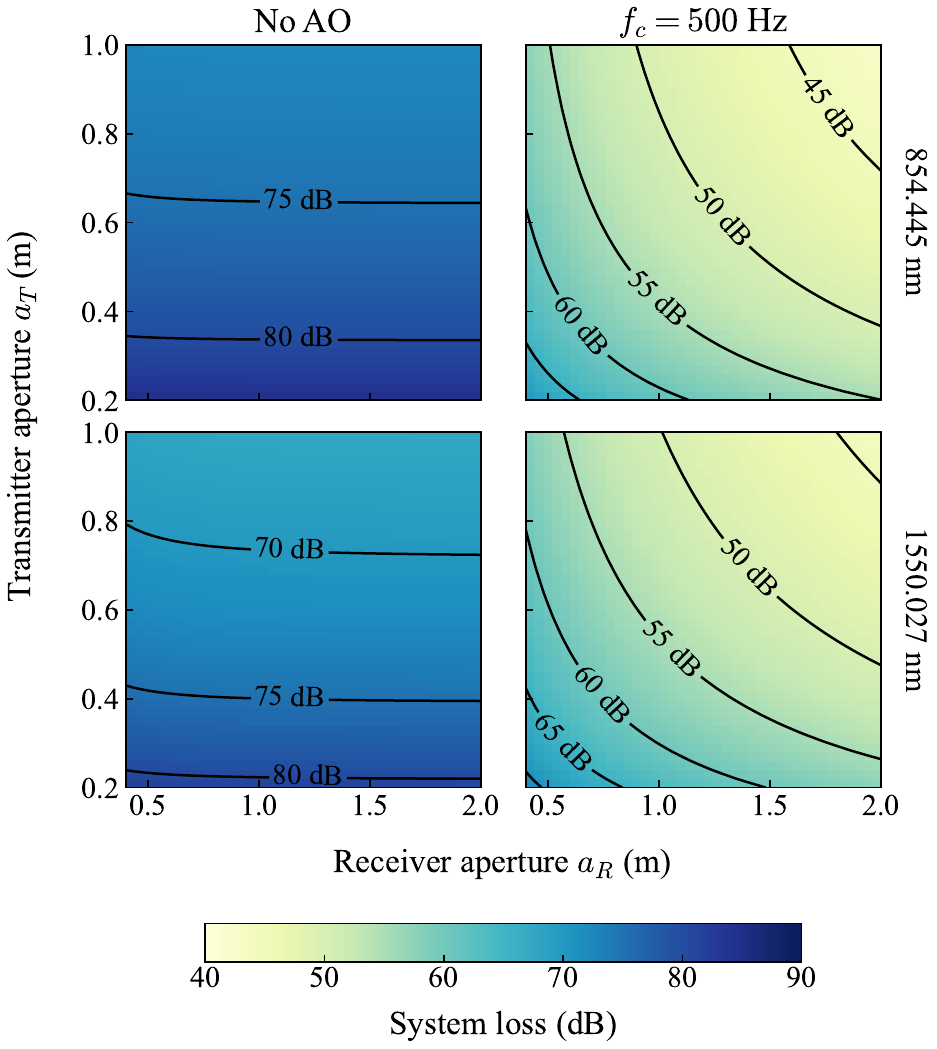}
    \caption{System loss (dB) as a function of receiver aperture $a_R$ (x axis) and transmitter aperture $a_T$ (y axis) on the fixed urban baseline slice with $\theta=60^\circ$, low pointing jitter, SMF coupling, and the default SNSPD Spec B receiver. The heat map and contour labels show iso-loss levels. Columns compare no AO with strong AO correction ($f_c=500$ Hz), while rows correspond to $\lambda=854.445$ nm and $\lambda=1550.027$ nm.}
    \label{fig:aperture_loss_contour}
\end{figure}

Figure~\ref{fig:loss_noise_context_scatter} closes the loop by placing representative asymmetric active-receiver link-model operating points back onto the same system-loss-noise plane. The colored background and black contours show the optimized finite-size secret-key rate on the zero-afterpulsing, zero-misalignment slice, while the overlaid markers place representative operating points on that map for the two wavelengths considered here, the two detector classes in Spec B, two zenith angles, rural and urban sites, and the no-AO and high-AO cases under nighttime and moderate daytime noise scenarios. This figure makes the complementary role of the preceding plots explicit: Fig.~\ref{fig:sky_radiance_contours} sets the vertical displacement through background noise, whereas Fig.~\ref{fig:aperture_loss_contour} sets the horizontal displacement through system loss.

The strongest separation is by background environment and wavefront correction: nighttime points sit much lower on the $p_{\mathrm{noise}}$ axis, while high AO shifts the operating points substantially toward lower effective loss, which is why reliable daytime operation is concentrated in the AO-corrected subsets. Figure~\ref{fig:loss_noise_context_scatter} also establishes the detector hierarchy used throughout the remaining comparisons: under matched channel conditions, SNSPD operating points sit deeper inside the positive-rate region than APD operating points. By contrast, the shifts associated with going from $45^\circ$ to $60^\circ$ zenith angles or from rural to urban conditions are more modest displacements within the same overall bands. The wavelength dependence is likewise not one-dimensional: 854~nm more often lands in the lower-noise positive-key-rate region than 1550~nm in this comparison, but the substantial overlap shows that wavelength choice cannot be inferred from a simple left-right or up-down displacement alone.\\

\begin{figure}[htbp]
    \includegraphics[width=\linewidth]{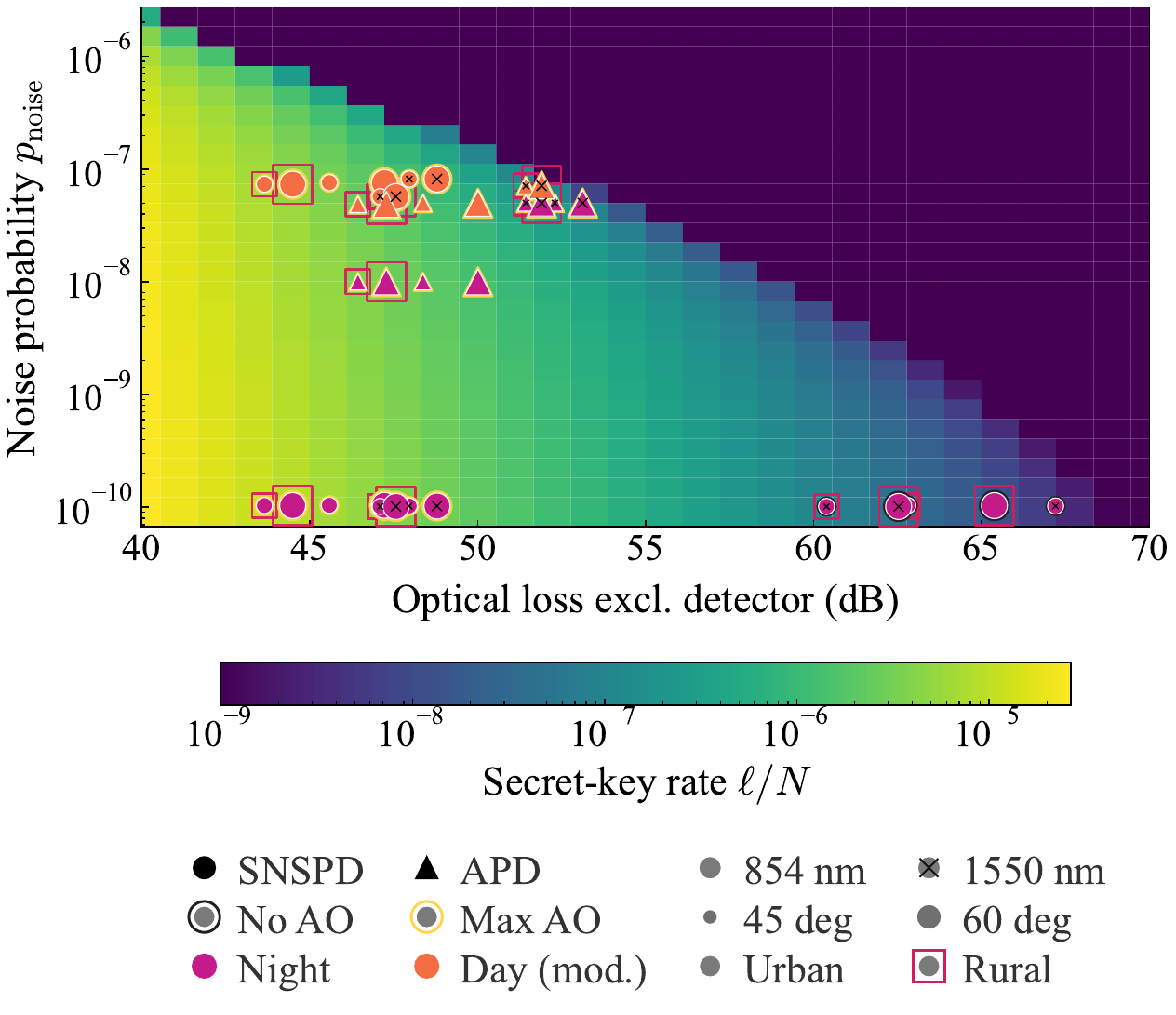}
    \caption{Finite-size secret-key rate map $\ell/N$ for the asymmetric active receiver as a function of total system loss (x axis) and noise-click probability $p_{\mathrm{noise}}$ (y axis) on the zero-afterpulsing, zero-misalignment slice, with representative positive-rate operating points from the link model overlaid. The operating points compare $\lambda=854.445$ nm and $\lambda=1550.027$ nm, zenith angles $45^\circ$ and $60^\circ$, rural and urban sites, no AO and high AO, and Spec B SNSPD/APD receivers under nighttime and moderate-noise daytime conditions.}
    \label{fig:loss_noise_context_scatter}
\end{figure}

\noindent\textit{Aperture, zenith-angle, and finite-key sweeps}

From this point onward, the finite-key rates are computed for the asymmetric passive BB84 receiver described in Appendix~\ref{sec:bb84_passive_asym}. Passive basis selection is generally preferred over active basis switching for its easier implementation, and the asymmetric design provides higher secret-key rates than the symmetric one in the high-loss regime.

Figure~\ref{fig:aperture_rate_contour} illustrates how the finite-size secret-key rate varies across the transmitter-receiver aperture plane for the default urban geometry ($\theta=60^\circ$), low AO correction ($f_c=130$ Hz), and the default SNSPD Spec B receiver.
% When compared  to Fig.~\ref{fig:aperture_loss_contour}, it shows the same aperture trade space after the finite-key rate optimization and the \blue{background model have been restored}. 
The night panels are broadly favorable, but in daytime the zero-rate wedge pushes inward from small apertures and forces operation toward the upper-right corner of the plots. The comparison between rows shows that, under a shared moderate-noise assumption, 854~nm preserves a visibly larger high-rate region than 1550~nm, consistent with the benefit of operating inside a Fraunhofer dip.

\begin{figure}[htbp]
    \includegraphics[width=\linewidth]{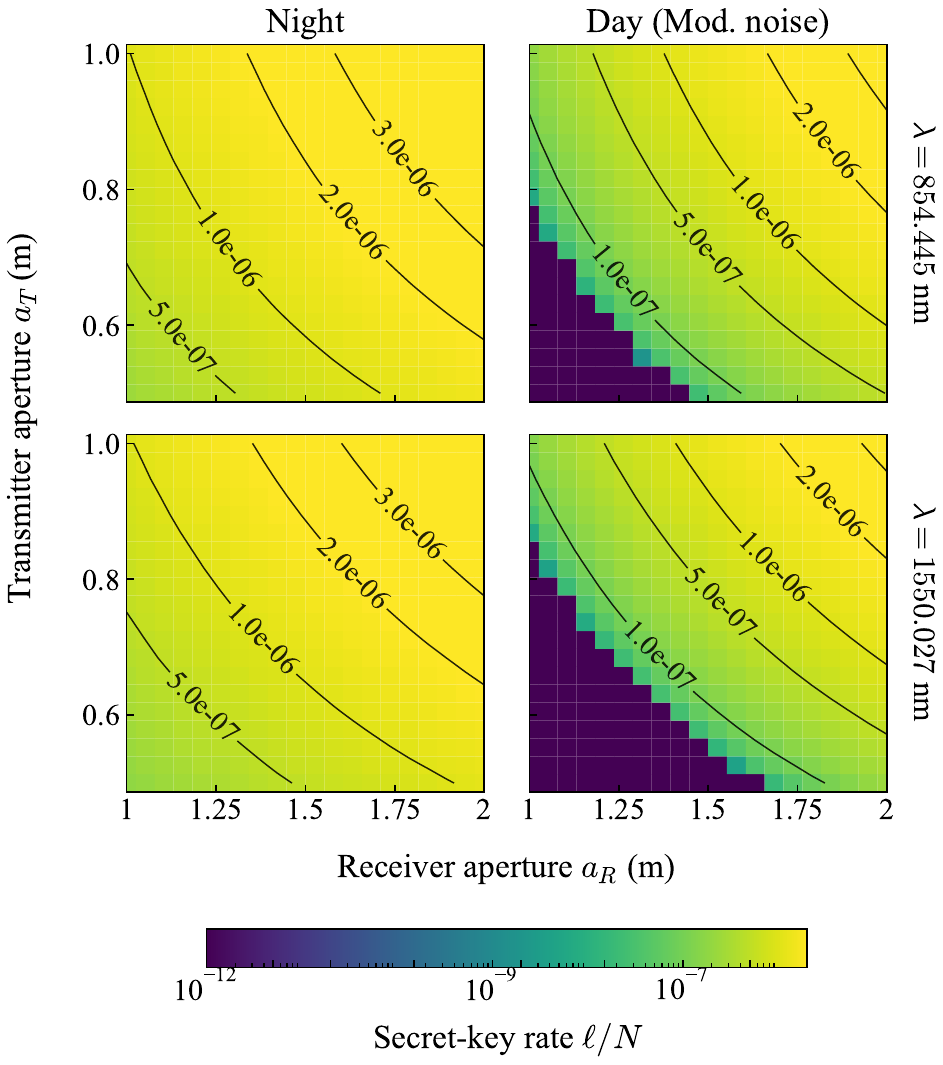}
    \caption{Finite-size secret-key rate $\ell/N$ as a function of receiver aperture $a_R$ (x axis) and transmitter aperture $a_T$ (y axis) for the default urban geometry with $\theta=60^\circ$, low AO correction ($f_c=130$ Hz), and an SNSPD Spec B receiver. The color scale and black contour lines show the finite-size secret-key rate. Columns compare nighttime and daytime operation under moderate background noise, while rows correspond to $\lambda=854.445$ nm and $\lambda=1550.027$ nm.}
    \label{fig:aperture_rate_contour}
\end{figure}

Figure~\ref{fig:zenith_aperture_loss} sweeps the zenith angle at fixed transmitter aperture $a_T=0.75$ m on the same urban low-AO baseline. In effect, it asks how far the aperture compensation identified in Fig.~\ref{fig:aperture_rate_contour} can be pushed before the longer atmospheric path and detector limitations overwhelm the link.
To remain on the same loss contour, the receiver aperture must increase with zenith angle, and this upward trend is visibly steeper at 854~nm than at 1550~nm. This indicates that the shorter-wavelength link is more sensitive to the longer atmospheric path at large zenith angles, so the loss penalty grows faster as the line of sight moves away from zenith.
The detector-threshold curves sharpen the hierarchy visible in Fig.~\ref{fig:loss_noise_context_scatter}. The Spec A APDs are close to unusable over much of the plotted space, and at 1550~nm the APD boundaries retreat so strongly that useful operation is essentially driven by the SNSPD cases.

\begin{figure}[htbp]
    \includegraphics[width=\linewidth]{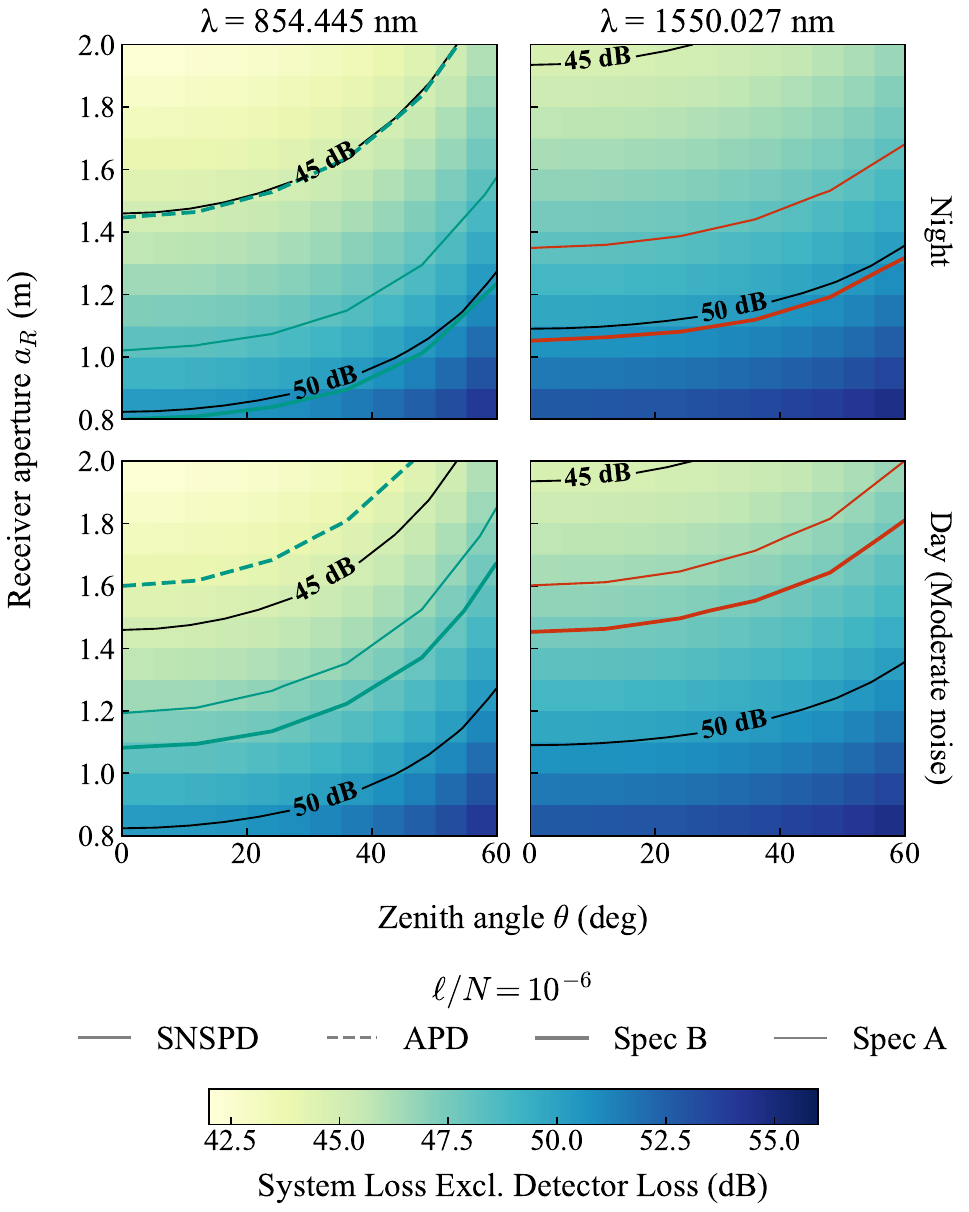}
    \caption{Combined loss and performance map as a function of zenith angle (x axis) and receiver aperture $a_R$ (y axis) at a fixed transmitter aperture $a_T=0.75$ m for the urban low-AO baseline with low pointing jitter and SMF coupling. The background heat map and black contours show system loss excluding detector loss. Colored boundary curves mark the $\ell/N=10^{-6}$ threshold for the different detector models: solid lines denote SNSPDs, dashed lines denote APDs, and thicker lines denote Spec B. Columns correspond to $\lambda=854.445$ nm and $\lambda=1550.027$ nm, and rows compare night with daytime operation at moderate background noise.}
    \label{fig:zenith_aperture_loss}
\end{figure}

Figure~\ref{fig:zenith_rate_finite_key}, on the other hand, shows the corresponding nighttime finite-size secret-key rate behavior once the receiver aperture and site quality are varied directly, still for the default transmitter aperture and low AO correction. Whereas Fig.~\ref{fig:zenith_aperture_loss} emphasizes the threshold geometry, this figure shows the full secret-key-rate decay inside the feasible region. Consistent with Fig.~\ref{fig:loss_noise_context_scatter}, the 1550~nm APD curves become useful for QKD only in the most favorable large-aperture cases. The rural-to-urban transition also affects the shorter wavelengths more strongly than 1550~nm, which is consistent with the stronger turbulence sensitivity of the shorter-wavelength channel together with its larger angle-dependent atmospheric penalty. At small zenith angles, 656~nm and 854~nm can still outperform 1550~nm because the path through the atmosphere is short and the intrinsic long-wavelength diffraction penalty has not yet been compensated by the stronger absorption losses of shorter wavelengths. As the zenith angle increases, however, the atmospheric and turbulence penalties grow more rapidly for the shorter wavelengths, so the 1550~nm curves often decay more slowly and remain competitive deeper into the large-angle regime, especially with SNSPD receivers.

\begin{figure}[htbp]
    \includegraphics[width=\linewidth]{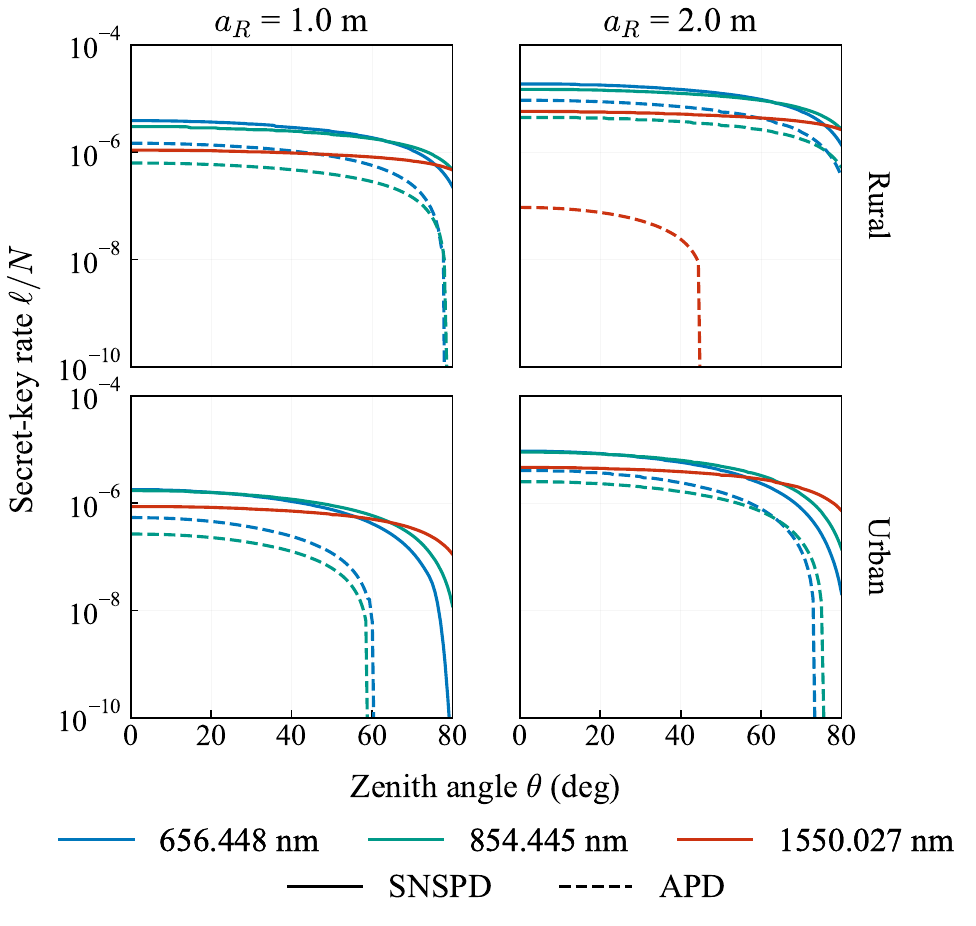}
    \caption{Finite-size secret-key rate $\ell/N$ (y axis) versus zenith angle (x axis) for different receiver apertures and operating environments under nighttime operation, with fixed transmitter aperture $a_T=0.75$ m, low AO correction, and Spec B detectors. Rows correspond to rural and urban conditions, while columns correspond to receiver apertures $a_R$ of $1.0$ m and $2.0$ m. Curve color denotes wavelength $\{656.448,\,854.445,\,1550.027\}$ nm, and line style distinguishes SNSPDs (solid) from APDs (dashed).}
    \label{fig:zenith_rate_finite_key}
\end{figure}

Figure~\ref{fig:min_b} makes the finite-size cost explicit for the default urban daytime architecture with $a_T=0.75$ m, $a_R=1.5$ m, and low AO correction. It is the block-size counterpart of the preceding daytime feasibility plots: channels that are admissible in the loss-noise and zenith-angle maps can still be operationally unattractive if they require impractically many transmission rounds before yielding a positive secret-key rate. The curves stay pinned to zero until a threshold value of $N$ is reached, after which the rate rises gradually. Across all panels, a $45^\circ$ zenith angle outperforms $60^\circ$, while the 1550~nm APD cases are the most demanding, failing to produce a positive key anywhere in the plotted range. Within the daytime slices shown here, the 854~nm curves also tend to cross the positive-key threshold earlier than the 1550~nm curves, especially at $45^\circ$, although earlier figures show that the relative wavelength advantage depends on zenith angle and operating regime rather than being universal. A practically important point is the required accumulation scale: in the favorable cases the onset of positive key is typically around $N\sim 10^{10}$--$10^{11}$ rounds. This corresponds to about $100$--$1000$ s at a $100$ MHz source, or about $10$--$100$ s at a $1$ GHz source.\\

\begin{figure}[htbp]
    \includegraphics[width=\linewidth]{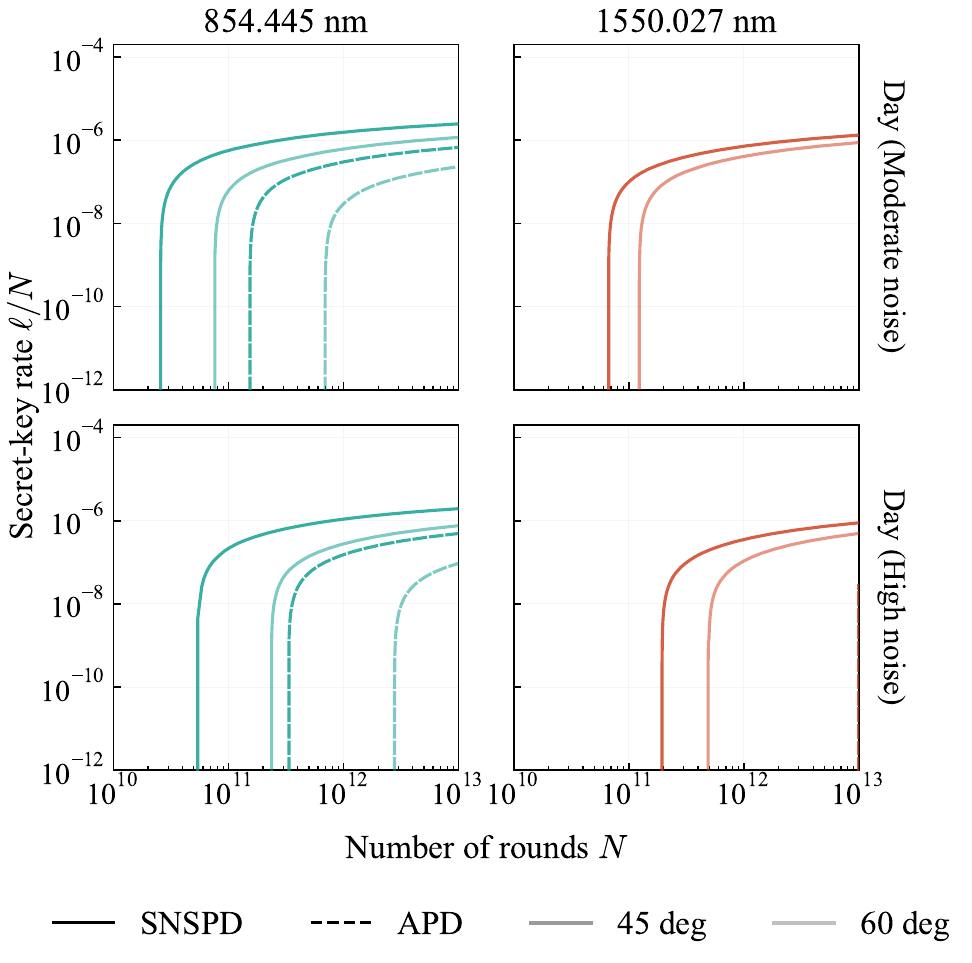}
    \caption{Finite-key rate $\ell/N$ (y axis) versus the number of transmission rounds $N$ (x axis), with both axes on logarithmic scales, for the default urban link with $a_T=0.75$ m, $a_R=1.5$ m, low AO correction, and Spec B detectors. Columns correspond to $\lambda=854.445$ nm and $\lambda=1550.027$ nm, and rows compare daytime operation under moderate and high background noise. Solid curves denote SNSPD receivers, dashed curves denote APD receivers, and the line opacity distinguishes zenith angles of $45^\circ$ and $60^\circ$.}
    \label{fig:min_b}
\end{figure}

\noindent\textit{Geographic annual yield}

We now translate these instantaneous feasibility trends into a geographic annual secret-key-yield prediction by combining the cloud statistics with a day/night averaging model in Fig.~\ref{fig:annual_keyrate_maps}.
For each location and retained cloud regime (clear-sky, thin-cloud, and thick-cloud), we evaluate the finite-size secret-key rate using the local channel loss and background conditions. The annual secret-key yield is then obtained by treating the year as equal daytime and nighttime intervals.
Within the daytime half, motivated by the Sun-track and sky-radiance behavior discussed in Appendix~\ref{sec:background_noise}, we assume that one third of the daylight interval is unusable due to the brightness of the sky when the Sun lies too close to the GEO line of sight. We split the remaining two thirds equally among the low, moderate and high background daytime scenarios. This daylight partition is part of the background model described in Section~\ref{sec:channel} and Appendix~\ref{sec:background_noise}.
We then calculate the secret-key rate, weighted by the local CRAAS cloud fractions, and multiply by the annual number of transmitted pulses. In particular, this map assumes repeated protocol runs with $N=10^{12}$ transmission signals and a source rate of $f = 1$ GHz.

The map shows a clear south-to-north degradation, with the best yields concentrated in the southern and eastern part of the domain and the weakest yields over the cloudier, lower-elevation northwestern region. This confirms that annual GEO-QKD performance is controlled jointly by viewing geometry and climatology, not by clear-sky optics alone.
Across the plotted cells, the geographic variation is substantial: the annual secret-key yield spans $2.7\times10^{9}$--$4.8\times10^{10}$ bits/year at 854.445~nm and $3.2\times10^{9}$--$2.8\times10^{10}$ bits/year at 1550.027~nm. The 854.445~nm yield is larger in about $86\%$ of the plotted cells, with a pointwise 854.445-to-1550.027~nm yield ratio spanning $0.71$--$1.74$. The cells where 854.445~nm underperforms 1550.027~nm are confined mainly to the northern, high-zenith-angle edge of the map, where the shorter wavelength's larger atmospheric and SMF-coupling penalties outweigh its geometric advantage. \\

\begin{figure}[htbp]
    \includegraphics[width=\linewidth]{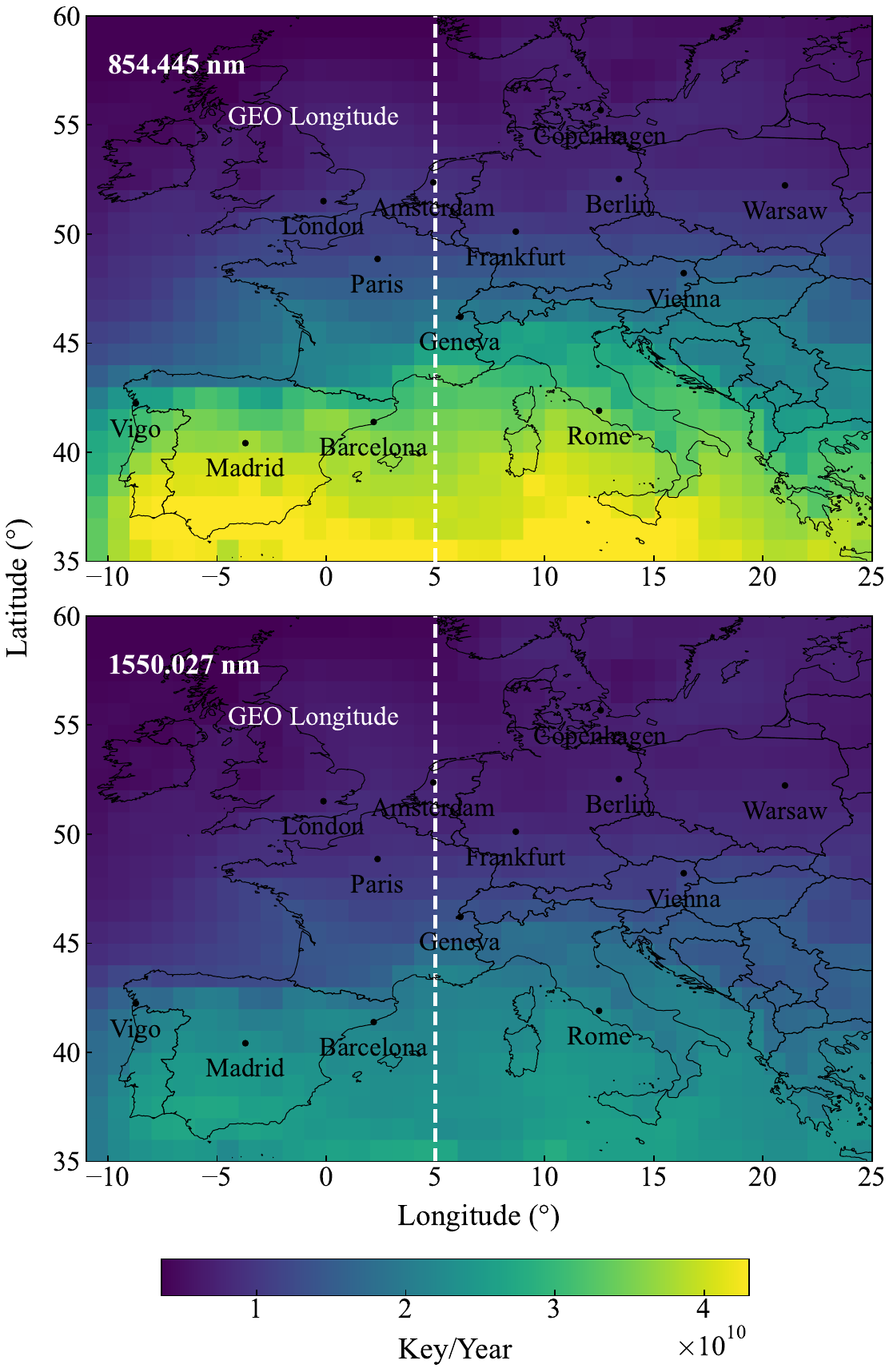}
    \caption{Geographic map of the cloud-weighted annual secret-key yield over the European latitude-longitude window. The top panel shows $\lambda=854.445$ nm and the bottom panel shows $\lambda=1550.027$ nm. The color scale gives the expected annual secret-key yield in bits/year for a GEO downlink, obtained by combining the finite-size secret-key rate with the local cloud-regime fractions and the day/night background model for the default aperture pair $a_T=0.75$ m and $a_R=1.5$ m, the urban site model, SNSPD Spec B, $N=10^{12}$, and $f=1$ GHz. Black markers indicate representative cities, and the dashed vertical line marks the satellite longitude at $5^\circ$E.}
    \label{fig:annual_keyrate_maps}
\end{figure}

\noindent\textit{Long-distance benchmark against fiber}

Finally, Fig.~\ref{fig:plob_bounds_wavelength_grid} provides a system-level benchmark by comparing the GEO link against terrestrial fiber limits.
It translates zenith-angle performance into the ground-arc distance $d_{\mathrm{arc}}$ between two locations on the Earth's surface, using the construction given in Appendix~\ref{sec:ground_separation}. Naturally, this construction accounts for the fact that, in order to establish a key between the ground stations, each of them must establish a key with the satellite in advance. The comparison is made for 854.445 nm and 1550.027 nm, using the urban night and moderate-noise daytime curves from the default $a_T=0.75$ m, $a_R=1.5$ m, and the SNSPD Spec B configuration. As fiber benchmarks, we show the repeaterless Pirandola--Laurenza--Ottaviani--Banchi (PLOB) bound~\cite{pirandola2017fundamental} and the corresponding idealized repeater-assisted bounds for $N_{\mathrm{rep}}=1,5,30$ repeaters \cite{pirandola2019end}, evaluated for a fiber attenuation of $0.2$ dB/km using $R_{\mathrm{fib}}^{\mathrm{rep}}=-\log_2\!\left(1-\eta_{\mathrm{fib}}^{1/(N_{\mathrm{rep}}+1)}\right)$ with $\eta_{\mathrm{fib}}=10^{-\alpha_{\mathrm{fib}} d_{\mathrm{arc}}/10}$.
The figure shows the contrast in long-distance scaling: the GEO curves remain nearly flat out to continental and transoceanic ranges, whereas the repeaterless and low-repeater fiber bounds fall off rapidly on the log-log axes.
Reaching distances of order $10^4$ km in fiber requires on the order of tens of repeaters, as illustrated by the $N_{\mathrm{rep}}=30$ reference curve, leaving aside the fact that quantum repeaters with reliable photon-memory interfaces and long coherence times are not yet available \cite{azuma2023quantum_repeaters}.
More importantly, a fiber-based repeater chain must follow the terrestrial route between the endpoints, which becomes especially restrictive for intercontinental or subsea connections. By contrast, the GEO architecture requires suitable ground stations at the two ends of the link, so its viability is largely independent of the terrain or ocean crossings in between.

\begin{figure}[htbp]
    \includegraphics[width=\linewidth]{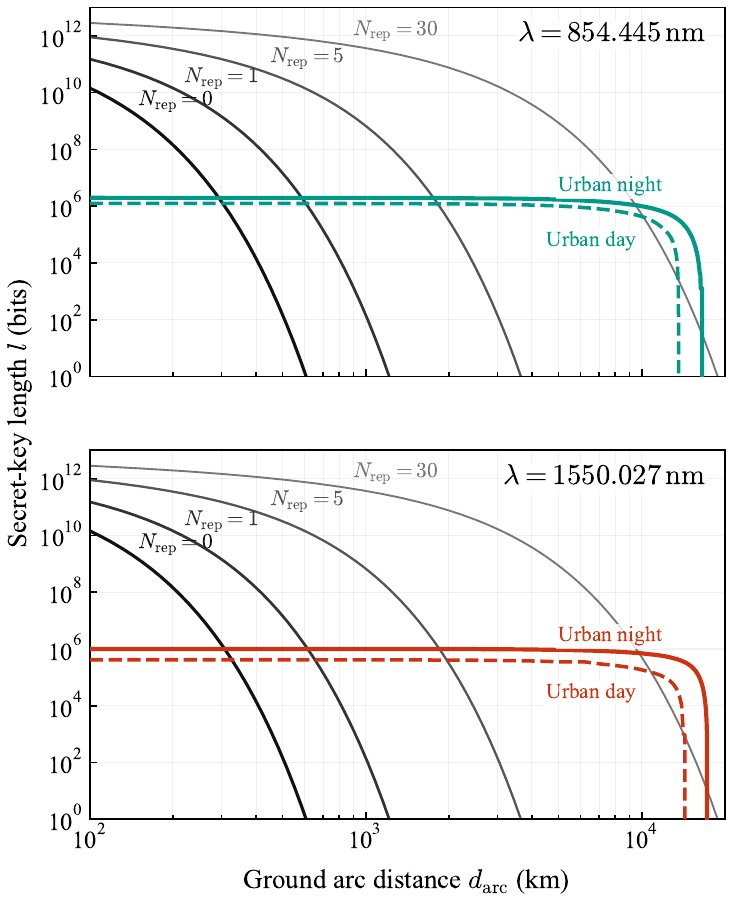}
    \caption{Secret-key length $\ell$ (y axis) versus ground arc distance $d_{\mathrm{arc}}$ (x axis), both on logarithmic scales, for the two signal wavelengths shown in the panel titles, $\lambda=854.445$ nm and $\lambda=1550.027$ nm. In each panel, the colored solid and dashed curves give the urban nighttime and moderate-noise daytime GEO-link results, respectively, for the default aperture pair $a_T=0.75$ m and $a_R=1.5$ m with an SNSPD Spec B receiver. The gray reference curves show the repeaterless PLOB bound and the idealized fiber repeater bounds for $N_{\mathrm{rep}}=1,5,30$ with fiber attenuation coefficient $0.2$ dB/km \cite{pirandola2017fundamental,pirandola2019end}. The same transmission budget, $N=10^{12}$ signals, is used for both the GEO and fiber cases. Since the GEO-QKD requires two satellite--OGS key-generation links, this budget is split evenly between them, allotting $N/2$ transmitted signals per satellite--OGS link.}
    \label{fig:plob_bounds_wavelength_grid}
\end{figure}

%% file: journal/sections/conclusion.tex
\section{Conclusion}
\label{sec:conclusion}

We have presented an end-to-end feasibility analysis of downlink GEO-QKD that combines a variable-length decoy-state BB84 security proof with a physically detailed channel and background model. The framework covers common active and passive receiver architectures, and ties together security analysis, receiver design, atmospheric transmission, and background noise within a single simulation pipeline. Taken together, the results show that GEO-QKD can support finite-key generation under realistic operating conditions and therefore constitutes a credible architecture for continental quantum networking.

At the protocol level, the study shows that finite-key effects have a relevant impact in the GEO regime, given the narrow feasible window in the combined loss-noise regime and the sparseness of the detected counts. The sweeps over the number of transmission rounds make this explicit: even favorable daytime links generally require at least $N\sim 10^{10}$--$10^{11}$ transmitted pulses before a positive key can be distilled, so finite-key performance must be treated as a design constraint from the outset.

The main engineering picture that emerges is likewise clear. GEO-QKD performance is governed by how effectively the system moves toward lower loss and lower background noise at the same time. Under representative urban daytime conditions, bringing the deterministic loss into the $\sim 50$--$60$ dB range requires apertures on the order of a meter, and efficient single-mode operation is not realistically sustained without adaptive optics. Within that design space, 854~nm is especially attractive for daylight operation because the Fraunhofer minimum reduces the solar background, while 1550~nm remains competitive and can become more robust at larger zenith angles because it degrades more slowly with pointing errors and atmospheric path length. Across all wavelengths, adaptive optics also unlock the benefits of using larger receiver apertures, which reduce loss without an equivalent penalty in collected background.

Detector choice sharpens these trends further. SNSPD receivers consistently preserve the broadest feasible operating region and retain positive key rates deeper into high-loss and high-background regimes, whereas APD implementations are much more constrained, especially at 1550~nm. When these instantaneous results are folded into the cloud-weighted annual key-yield analysis, the achievable annual secret-key yield shows a pronounced south-to-north degradation across Europe, demonstrating that GEO-QKD performance is controlled jointly by optical design, satellite elevation, and climatology.

The long-distance comparison with terrestrial fiber benchmarks reinforces the systems-level motivation for GEO links: once feasibility is established, the GEO curves remain comparatively flat with ground distance, whereas terrestrial secret-key capacity falls off rapidly.

Importantly, several aspects deserve further investigation. The present model neglects time-varying effects; it relies on representative turbulence and cloud statistics rather than site-specific forecasting or real-time adaptation. In addition, the network picture considered here is still that of a trusted GEO node. Natural next steps are to place GEO links inside wider QKD network architectures, to compare a persistent GEO backbone service against LEO constellations using explicit orbital visibility and handover dynamics, and to study hybrid GEO--LEO systems in which GEO provides continuous coverage while LEO segments add lower-loss or regionally flexible links.

More broadly, the value of the present study is to turn GEO-QKD from a qualitative possibility into a quantitatively delimited design problem. The operating windows are narrow, but they are now explicit enough to guide the choice of wavelength, detector technology, optical apertures, and observing conditions for future GEO-QKD missions.

%% file: journal/sections/acknowledgements.tex
\section*{Acknowledgments}
The authors specially thank Marcos Reyes, Hannah Thiel, and Alvaro Magdalena for helpful discussions. This work was supported by the European Union’s Horizon Europe Framework Programme under the Marie Sklodowska-Curie Grant No. 101072637 (Project QSI), the Galician Regional Government (consolidation of research units: atlanTTic), the Spanish Ministry of Science, Innovation and Universities (MICIU), the Fondo Europeo de Desarrollo Regional (FEDER) through the grant No. PID2024-162270OB-I00, the “Hub Nacional de Excelencia en Comunicaciones Cuanticas” funded by the Spanish Ministry for Digital Transformation and the Public Service and the European Union NextGenerationEU, the European Union’s Horizon Europe Framework Programme under the project Quantum Secure Networks Partnership (QSNP, grant agreement No 101114043), the European Union under the Project IberianQCI (grant 101249593), and the “Programa de Cooperación Interreg VI-A España–Portugal” (POCTEP) 2021–2027 through the project QUANTUM IBER\_IA.

%% file: journal/sections/appendix_decoy_state_bb84_protocol.tex
\section{Asymmetric active decoy-state BB84 protocol}
\label{sec:protocol_appendix}

\subsection{Protocol description}
\label{sec:active_description}
We consider a variant of the variable-length decoy-state BB84 protocol presented in \cite{tupkary2025imperfect} with an active BB84 receiver. The basis $Z$ is used for key generation and the basis $X$ for testing. Alice and Bob agree on the protocol inputs $N$, $p^A_X$, $J$, $\mathcal{K}:=\left\{\mu_0, \cdots, \mu_J\right\}$, $\left\{p_{\mu_0}, \cdots, p_{\mu_J}\right\}$, $\epsilon_{\mathrm{PE}}$, $\epsilon_{\mathrm{cor}}$, and $\epsilon_{\mathrm{PA}}$, specified in Table~\ref{bb84inputs}.
The protocol runs as follows. For $r=1,\ldots,N$, steps 1 and 2 below are repeated.
\begin{enumerate}
    \item{}\emph{State preparation}: Alice chooses a bit value $y_r$ uniformly at random, a basis setting $a_{r}\in \{Z,X\}$ with probability $p^A_{a_{r}}$,  with $p^A_Z=1-p^A_X$, and  an intensity {setting $k_r\in{}\mathcal{K}$ with probability $p_{k_{r}}$}. Then, she prepares a {phase-randomized weak coherent pulse (PRWCP) encoded with the above settings and sends it to Bob  through the quantum channel.}

    \item{}\emph{Active measurement}: Bob measures the {incident pulse in the basis $b_r\in \{Z,X\}$ with probability $p^B_{Z}=p^A_{Z}$. He records the measurement outcome as $y_r^{\prime}\in\{0,1, \emptyset\}$, where $\emptyset$ refers to the ``no-click'' event. Double clicks are assigned to 0 or 1 uniformly at random.}

    \item{}\emph{{Sifting}}: {The bases and intensity settings selected are publicly revealed and Alice and Bob identify the sets $\mathcal{Z}^{\mathrm{key}}=\bigl\{r: a_r=b_r=Z,\ y_r^{\prime} \neq \emptyset\bigr\}$,
          partitioned as $\mathcal{Z}^{\mathrm{key}}=\cup_{k \in \mathcal{K}} \mathcal{Z}^{\mathrm{key}}_{k}$ with $\mathcal{Z}^{\mathrm{key}}_k=\left\{r\in\mathcal{Z}^{\mathrm{key}}: k_r=k\right\}$ of size  $n_{k, Z}$,
          and $\mathcal{X}^{\mathrm{test}}=\left\{r: a_r=b_r=X,\ y_r^{\prime} \neq \emptyset\right\}$,
          partitioned as $\mathcal{X}^{\mathrm{test}}=\cup_{k \in \mathcal{K}} \mathcal{X}^{\mathrm{test}}_{k}$ with $\mathcal{X}^{\mathrm{test}}_k=\left\{r\in\mathcal{X}^{\mathrm{test}}: k_r=k\right\}$ of size  $n_{k,X}$. }

    \item{}\emph{Parameter estimation} (PE): {For each $k \in \mathcal{K}$, they disclose the bit values in the rounds $\mathcal{X}_k^{\mathrm{test}}$ and compute the corresponding numbers of bit errors, $m_{k, {X}}$.  Using the methods presented in Appendix~\ref{dsa}, they estimate the parameters $n^{\mathrm{L}}_{1,Z}$ (lower bound on the single-photon counts in the $Z$ basis) and $\phi^{\mathrm{U}}_{1,Z}$ (upper bound on the phase-error rate of the single-photon counts in the $Z$ basis), which depend on the requested PE error, $\epsilon_{\mathrm{PE}}$. }

    \item{}\emph{Variable-length decision}:{ Alice and Bob calculate the number of bits to be used for error correction, $\lambda_{\mathrm{EC}}$ \footnote{In practice, the users determine $\lambda_{\mathrm{EC}}$ as a function of the observed data. For our simulations, we set $\lambda_{\mathrm{EC}}=1.16n_{Z}h(e_Z)$, where $e_Z$ is set to the expected quantum bit error rate (QBER) of the channel.}  and the final key length  $\ell$ given by
          \begin{equation}
              \begin{split}
                  \ell =
                  \max\Biggl\{
                  \Bigl\lfloor n^{\mathrm{L}}_{1,Z}\left[1-h(\phi^{\mathrm{U}}_{1,Z})\right] \\
                  - \lambda_{\mathrm{EC}} - \log_2\left(\frac{1}{2\epsilon_{\mathrm{PA}}^2 \epsilon_{\mathrm{cor}}}\right)
                  \Bigr\rfloor,\,0
                  \Biggr\}.
              \end{split}
          \end{equation}
          Aborting is modeled as producing a key of length zero. If they do not abort, they proceed to steps 6,7 and 8.}

    \item{}\emph{Error correction} (EC): {Alice and Bob implement a one-way EC protocol that reveals  $\lambda_{\mathrm{EC}}$ bits  of syndrome information.}

    \item{}\emph{Error verification} (EV): {Alice and Bob perform an EV step based on 2-universal hashing, using EV tags of length $\log(2/\epsilon_{\mathrm{cor}})$ bits at most.  If the EV tags do not match, they abort the protocol. Otherwise, they proceed to privacy amplification.}

    \item{}\emph{Privacy amplification} (PA): {Alice and Bob perform a PA step based on 2-universal hashing, with failure probability at most $\epsilon_{\mathrm{PA}}$. In doing so they obtain an $\epsilon_{\mathrm{tot}}$-secure output key of length $\ell$, with
          \begin{equation}
              \label{eq:eqtot}
              \epsilon_{\mathrm{tot}} = 2\sqrt{\epsilon_{\mathrm{PE}}}+\epsilon_{\mathrm{PA}}+\epsilon_{\mathrm{cor}}.
          \end{equation}

          }

\end{enumerate}

\begin{table}[htbp]
    \centering
    \caption{Decoy-state BB84 protocol inputs.}
    \label{bb84inputs}
    \begin{tabularx}{\columnwidth}{|c|X|}
        \hline
        $N$                                   & Number of transmission rounds                         \\
        \hline
        $p^A_X$                               & Test basis probability                                \\
        \hline
        $\mathcal{K}:=\{\mu_0,\cdots,\mu_J\}$ & Set of intensities                                    \\
        \hline
        $\{p_{\mu_0},\cdots,p_{\mu_J}\}$      & Probabilities of the intensity settings               \\
        \hline
        $\epsilon_{\mathrm{PE}}$              & Failure probability of the parameter estimation step  \\
        \hline
        $\epsilon_{\mathrm{cor}}$             & Failure probability of the error verification step    \\
        \hline
        $\epsilon_{\mathrm{PA}}$              & Failure probability of the privacy amplification step \\
        \hline
    \end{tabularx}
\end{table}

\subsection{Decoy-state analysis}\label{dsa}

Here we describe the linear programs (LP) used to compute the quantities \( n^{\mathrm{L}}_{1,Z} \) and \( \phi^{\mathrm{U}}_{1,Z} \) in the protocol, based on the observed values \( n_{k,Z}, n_{k,X}\) and  \(m_{k,X} \) for all \( k \in \mathcal{K} \). Our analysis is based on the earlier work presented in Ref. \cite{attema2021optimizing}, with further additions.

In the expressions below  \( N(n) \) denotes the number of transmitted (detected) pulses, and \( m \)  the number of observed errors. As for the subscripts, \( i \) denotes the number of photons in a pulse, \( \mu_j \) denotes the intensity setting, and \( \mathcal{B} \) denotes the matched basis selected both by Alice and Bob. For example, \( n_{\mu_j,\mathcal{B}} \) denotes the number of detected pulses prepared and measured in the  \(\mathcal{B}\)-basis  and generated  with intensity $\mu_j$; \( N_{i,\mathcal{B}} \) denotes the number of transmitted \( i \)-photon pulses prepared and measured in the  \( \mathcal{B}\)-basis ; and \( m_X \) denotes the total number of errors observed in the \( X \) basis, aggregated over all intensity settings.

To take into account the effect of finite statistics, we consider the standard counterfactual scenario according to which Alice prepares and sends to Bob photon-number states and the assignment of the intensity settings to the detected signals is performed a posteriori~\cite{lim2014concise}. In this scenario, the expected value of $n_{{\mu_j},\mathcal{B}}$ is given by
\begin{equation}
    \label{countfactequ}
    \mathbb{E}\left[n_{{\mu_j},\mathcal{B}}\right] = \sum_{i=0}^{\infty} p_{\left(\mu_j|i\right)}n_{i,\mathcal{B}},
\end{equation}
where $p_{(\mu_j|i)}$ is the conditional probability that Alice assigns the intensity setting $\mu_j$ to an $i$-photon state.
Due to Bayes' rule, the conditional probability $ p_{\left(\mu_j|i\right)}$ can be written as
\begin{equation}
    p_{\left(\mu_j|i\right)}=\frac{p_{(i|\mu_j)}p_{\mu_j}}{p_i}.
\end{equation}
In the case of PRWCPs, the photon-number statistics $p_{(i|\mu_j)}$ of the source follow a Poisson distribution, meaning that \(
p_{(i|\mu_j)}=e^{-\mu_j}\mu_j^i/i!,
\)
and, in addition, we have that

\begin{equation}
    p_i =
    \sum_{j=0}^Jp_{(i|\mu_j)}p_{\mu_j}=
    \sum_{j=0}^J\frac{p_{\mu_j}e^{-\mu_j} \mu_j^i}{i!},
\end{equation}
where $p_{\mu_j}$ is the probability of Alice selecting the intensity $\mu_j$.

In the counterfactual scenario, the observed variables $n_{{\mu_j},\mathcal{B}}$ arise via independent Bernoulli sampling from $n_{\mathcal{B}}$ and hence, one can apply Theorem~\ref{PoissonBinomial} introduced in Appendix~\ref{sec:bounds} to bound the quantity $\mathbb{E}\left[n_{{\mu_j},\mathcal{B}}\right]$.  In doing so we obtain

\begin{equation}
    \label{multchern}
    \begin{split}
        \Pr \big[ \mathbb{E}[n_{\mu_j,\mathcal{B}}]/n_{\mathcal{B}} \geq \mathcal{F}^{\mathrm{U}}_{n_{\mathcal{B}},\epsilon^{\mathrm{U}}_{\mu_j,n_\mathcal{B}}} \left(n_{\mu_j,\mathcal{B}}/n_{\mathcal{B}}\right) \big] \\ \leq \epsilon^{\mathrm{U}}_{\mu_j,n_\mathcal{B}},
    \end{split}
\end{equation}
\begin{equation}
    \label{multchern2}
    \begin{split}
        \Pr \big[ \mathbb{E}[n_{\mu_j,\mathcal{B}}]/n_{\mathcal{B}} \leq \mathcal{F}^{\mathrm{L}}_{n_{\mathcal{B}},\epsilon^{\mathrm{L}}_{\mu_j,n_\mathcal{B}}}\left(n_{\mu_j,\mathcal{B}}/n_{\mathcal{B}}\right) \big] \\ \leq \epsilon^{\mathrm{L}}_{\mu_j,n_\mathcal{B}},
    \end{split}
\end{equation}
for all $0\leq j \leq J$, $\mathcal{B} \in \{Z,X\}$, and $\epsilon^{\Delta}_{\mu_j,n_\mathcal{B}}>0$ for \(\Delta \in \{\mathrm{U}, \mathrm{L}\}\).
The functions $\mathcal{F}^{\Delta}_{n_{\mathcal{B}},\epsilon^{\Delta}_{\mu_j,n_\mathcal{B}}}$ are given by Eqs.~\eqref{eq:poisslb}, \eqref{eq:poissub}.
Importantly, this allows us to relate  the observed parameters $n_{\mu_j,\mathcal{B}}$ to their expected values as

\begin{equation}
    \label{expectdetected}
    n_{\mu_j,\mathcal{B}}+ \delta_{\mu_j, n_\mathcal{B}}=\mathbb{E}\left[ n_{\mu_j,\mathcal{B}}\right],\end{equation}
with
\(
\delta^{\mathrm{L}}_{\mu_j,n_\mathcal{B}}
\leq
\delta_{\mu_j, n_\mathcal{B}}
\leq
\delta^{\mathrm{U}}_{\mu_j,n_\mathcal{B}},
\)
and where
\(
\delta^{\Delta}_{\mu_j,n_\mathcal{B}} =n_{\mathcal{B}}\mathcal{F}^{\Delta}_{n_{\mathcal{B}},\epsilon^{\Delta}_{\mu_j,n_\mathcal{B}}}\left(n_{\mu_j,\mathcal{B}}/n_{\mathcal{B}}\right)-n_{\mu_j,\mathcal{B}},\nonumber
\)
for \(\Delta\in\{\mathrm{U},\mathrm{L}\}\), except with probability at most $\epsilon^{\mathrm{U}}_{\mu_j,n_\mathcal{B}}+ \epsilon^{\mathrm{L}}_{\mu_j,n_\mathcal{B}}$.

Since $\sum_{j=0}^J\mathbb{E}\left[ n_{\mu_j,\mathcal{B}}\right]=n_{\mathcal{B}}$, by using Eq.~(\ref{expectdetected}) we also have the constraint,
\begin{equation}
    \sum_{j=0}^J \delta_{\mu_j, n_\mathcal{B}} =0.
\end{equation}

Let us now turn our attention to the parameters  $n_{i,\mathcal{B}}$, which are the unknown variables of the LP below.  To start  with, we note that
$n_{i,\mathcal{B}}$ is trivially upper bounded by the number of $i$-photon pulses sent and measured in the $\mathcal{B}$-basis, $N_{i,\mathcal{B}}$. That is,
\begin{equation}
    \label{eq:constraint2}
    n_{i,\mathcal{B}} \leq N_{i,\mathcal{B}},
\end{equation}
for all $i\geq0$ and $\mathcal{B} \in \{Z,X\}$. Moreover, we have that
\begin{equation}
    \mathbb{E}\left[N_{i,\mathcal{B}}\right] = p_i N_{\mathcal{B}}.
\end{equation}
Since the variables \( N_{i,\mathcal{B}} \) are sums of independent and identically distributed Bernoulli random variables, applying Theorem~\ref{Binomial} presented in Appendix~\ref{sec:bounds} provides a bound on \( N_{i,\mathcal{B}} \).

In particular, we obtain

\begin{equation}
    \label{binomupper}
    \Pr\left[N_{i,\mathcal{B}}/N_{\mathcal{B}} \geq
        \mathcal{G}^{\mathrm{U} }_{N_{\mathcal{B}},\epsilon^{\mathrm{U}}_{i,\mathcal{B}}}\left(p_i\right)\right]
    \leq \epsilon^{\mathrm{U}}_{i,\mathcal{B}},
\end{equation}

\begin{equation}
    \label{binomlower}
    \Pr\left[N_{i,\mathcal{B}}/N_{\mathcal{B}} \leq
        \mathcal{G}^{\mathrm{L} }_{N_{\mathcal{B}},\epsilon^{\mathrm{L}}_{i,\mathcal{B}}}\left(p_i\right)\right]
    \leq \epsilon^{\mathrm{L}}_{i,\mathcal{B}},
\end{equation}
for all $i\geq0$, $\mathcal{B} \in \{Z,X\}$, and  $\epsilon^{\{L,U\}}_{i,\mathcal{B}}>0$. The functions $\mathcal{G}^{\Delta}_{N_{\mathcal{B}},\epsilon^{\Delta}_{i,\mathcal{B}}}$ are given by Eqs.~\eqref{eq:binub}, \eqref{eq:binlb}. By combining  Eqs.~\eqref{eq:constraint2}--\eqref{binomlower}, along with the condition $n_{i,\mathcal{B}}\leq n_{\mathcal{B}}$ for all $i\geq 0$, we obtain

\begin{equation}
    \label{photonnumberrelation}
    0 \leq n_{i,\mathcal{B}} \leq \min(N_{i,\mathcal{B}},n_{\mathcal{B}})=\min \left(p_i N_{\mathcal{B}}+\delta_{i,\mathcal{B}},n_\mathcal{B}\right),
\end{equation}
with the parameter $\delta_{i,\mathcal{B}}$ satisfying

\begin{equation}
    \label{delta1bound}\delta^{\mathrm{L}}_{i,\mathcal{B}}\leq\delta_{i,\mathcal{B}}\leq\delta^{\mathrm{U}}_{i,\mathcal{B}},
\end{equation}
except with probability at most $\epsilon^{\mathrm{U}}_{i,\mathcal{B}}+ \epsilon^{\mathrm{L}}_{i,\mathcal{B}}$, where
\(
\delta^{\Delta}_{i,\mathcal{B}}=N_{\mathcal{B}}\mathcal{G}^{\Delta}_{N_{\mathcal{B}},\epsilon^{\Delta}_{i,\mathcal{B}}}\left(p_i\right)-p_i N_{\mathcal{B}}
\)
for \(\Delta\in\{\mathrm{U}, \mathrm{L}\}\).

To enable a finite set of unknowns in the LPs, we define a threshold photon number $I$ and upper-bound the number of detected pulses within $n_{\mathcal{B}}$ that contain more than $I$ photons, i.e., $\sum_{i=I+1}^{\infty} n_{i,\mathcal{B}}$.
Note that the number of detections is trivially upper bounded by the number of transmitted signals, i.e., $\sum_{i=I+1}^{\infty} n_{i,\mathcal{B}}\leq\sum_{i=I+1}^{\infty} N_{i,\mathcal{B}}$.
To upper bound the latter quantity, note that $\sum_{i=I+1}^{\infty} N_{i,\mathcal{B}}$ counts the number of the $N_{\mathcal{B}}$ transmitted signals that contain more than $I$ photons; hence it is a binomial random variable with parameters $N_{\mathcal{B}}$ and $\sum_{i=I+1}^{\infty} p_i$, and expectation $N_{\mathcal{B}}\sum_{i=I+1}^{\infty} p_i$. Therefore, one can apply Theorem~\ref{Binomial} provided in Appendix~\ref{sec:bounds} to obtain

\begin{equation}
    \begin{split}
        \Pr\left[ \left(\sum_{i=I+1}^{\infty} N_{i,\mathcal{B}}\right)/N_{\mathcal{B}}  \geq \mathcal{G}^{\mathrm{U}}_{N_{\mathcal{B}},\epsilon^{\mathrm{U}}_{> I,\mathcal{B}}}\left(\sum_{i=I+1}^{\infty} p_i\right)\right] \\ \leq \epsilon^{\mathrm{U}}_{> I,\mathcal{B}},
    \end{split}
\end{equation}
\begin{equation}
    \begin{split}
        \Pr\left[ \left(\sum_{i=I+1}^{\infty} N_{i,\mathcal{B}}\right)/N_{\mathcal{B}} \leq \mathcal{G}^{\mathrm{L}}_{N_{\mathcal{B}},\epsilon^{\mathrm{L}}_{> I,\mathcal{B}}}\left(\sum_{i=I+1}^{\infty} p_i\right)\right] \\ \leq \epsilon^{\mathrm{L}}_{> I,\mathcal{B}},
    \end{split}
\end{equation}
for $\mathcal{B}\in \{Z,X\}$, and $\epsilon^{\Delta}_{> I,\mathcal{B}}>0$ for $\Delta\in\{\mathrm{U},\mathrm{L}\}$. This leads us to the relation

\begin{equation}
    \begin{split}
        \sum_{i=I+1}^{\infty} n_{i,\mathcal{B}} & \leq \sum_{i=I+1}^{\infty} N_{i,\mathcal{B}}                                 \\
                                                & = N_{\mathcal{B}}\sum_{i=I+1}^{\infty} p_i +\delta_{>I,\mathcal{B}}          \\
                                                & = N_{\mathcal{B}}\left(1-\sum_{i=0}^{I} p_i\right) +\delta_{>I,\mathcal{B}},
    \end{split}
\end{equation}
with the parameter $\delta_{>I,\mathcal{B}}$ satisfying

\begin{equation}
    \delta^{\mathrm{L}}_{>I,\mathcal{B}}\leq\delta_{>I,\mathcal{B}}\leq\delta^{\mathrm{U}}_{>I,\mathcal{B}},
\end{equation}
except with probability at most $\epsilon^{\mathrm{U}}_{> I,\mathcal{B}}+\epsilon^{\mathrm{L}}_{> I,\mathcal{B}}$, with
\begin{equation}
    \begin{split}
        \delta^{\Delta}_{>I,\mathcal{B}} & =N_{\mathcal{B}}\mathcal{G}^{\Delta}_{N_{\mathcal{B}},\epsilon^{\Delta}_{> I,\mathcal{B}}}\left(1-\sum_{i=0}^{I} p_i\right) \\
                                         & \quad -N_{\mathcal{B}}\left(1-\sum_{i=0}^{I} p_i\right) ,
    \end{split}
\end{equation}
for $\Delta\in\{\mathrm{U}, \mathrm{L}\}$. Hence an upper bound on the quantity $\sum_{i=0}^{\infty} p_{\left(\mu_j|i\right)}n_{i,\mathcal{B}} $  follows as
\begin{equation}
    \label{eq:truncation}
    \begin{split}
        \sum_{i=0}^{\infty} p_{\left(\mu_j|i\right)}n_{i,\mathcal{B}} & \leq \sum_{i=0}^{I} p_{\left(\mu_j|i\right)}n_{i,\mathcal{B}}+\sum_{i=I+1}^{\infty} n_{i,\mathcal{B}} \\
                                                                      & \leq \sum_{i=0}^{I} p_{\left(\mu_j|i\right)}n_{i,\mathcal{B}}                                         \\
                                                                      & +N_{\mathcal{B}}\left(1-\sum_{i=0}^{I} p_i\right)+\delta_{>I,\mathcal{B}}.
    \end{split}
\end{equation}
We also have the following lower bound,
\begin{equation}
    \label{eq:trunclow}
    \sum_{i=0}^{\infty} p_{\left(\mu_j|i\right)}n_{i,\mathcal{B}}\geq \sum_{i=0}^{I} p_{\left(\mu_j|i\right)}n_{i,\mathcal{B}}.
\end{equation}

Importantly, Eqs.~\eqref{eq:truncation} and~\eqref{eq:trunclow} bound the quantity $\sum_{i=0}^{\infty} p_{\left(\mu_j|i\right)}n_{i,\mathcal{B}}$ with a finite number of terms and make the LPs below numerically feasible. They respectively form an upper and lower bound for Eq.~\eqref{expectdetected} via Eq.~\eqref{countfactequ}. Specifically, using Eqs.~\eqref{countfactequ},
\eqref{expectdetected}, and~\eqref{eq:truncation}, we have that
\begin{multline}
    n_{\mu_j,\mathcal{B}}+ \delta_{\mu_j, n_\mathcal{B}} = \mathbb{E}\left[n_{{\mu_j},\mathcal{B}}\right] \\
    = \sum_{i=0}^{\infty} p_{\left(\mu_j|i\right)}n_{i,\mathcal{B}} \\
    \leq \sum_{i=0}^{I} p_{\left(\mu_j|i\right)}n_{i,\mathcal{B}}+N_{\mathcal{B}}\left(1-\sum_{i=0}^{I} p_i\right)+\delta_{>I,\mathcal{B}},
\end{multline}
and using Eqs.~\eqref{countfactequ},
\eqref{expectdetected}, and~\eqref{eq:trunclow}, we also have
\begin{equation}
    \begin{split}
        n_{\mu_j,\mathcal{B}}+ \delta_{\mu_j, n_\mathcal{B}} & = \mathbb{E}\left[n_{{\mu_j},\mathcal{B}}\right] = \sum_{i=0}^{\infty} p_{\left(\mu_j|i\right)}n_{i,\mathcal{B}} \\
                                                             & \geq \sum_{i=0}^{I} p_{\left(\mu_j|i\right)}n_{i,\mathcal{B}}.
    \end{split}
\end{equation}

The truncation provided by Eq.~\eqref{eq:truncation} also imposes a constraint on the variables $\delta_{i,\mathcal{B}}$ that appear in Eq.~\eqref{photonnumberrelation}. In particular, by using the fact that $N_{i,\mathcal{B}}= p_i N_{\mathcal{B}}+\delta_{i,\mathcal{B}}$ (except with probability $\epsilon^{\mathrm{U}}_{i,\mathcal{B}}+ \epsilon^{\mathrm{L}}_{i,\mathcal{B}}$) and  applying Theorem~\ref{Binomial}, we have that

\begin{equation}
    \begin{split}
        \Pr\bigg[ \frac{1}{N_{\mathcal{B}}} \left( \sum_{i=0}^I p_iN_\mathcal{B}+\delta_{i,\mathcal{B}}\right) \qquad \qquad \\
            \geq \mathcal{G}^{\mathrm{U}}_{N_{\mathcal{B}},\epsilon^{\mathrm{U}}_{\leq I,\mathcal{B}}}\left(\sum_{i=0}^I p_i\right) \bigg] \leq \epsilon^{\mathrm{U}}_{\leq I,\mathcal{B}},
    \end{split}
\end{equation}
\begin{equation}
    \begin{split}
        \Pr\bigg[ \frac{1}{N_{\mathcal{B}}} \left( \sum_{i=0}^I p_iN_\mathcal{B}+\delta_{i,\mathcal{B}}\right) \qquad \qquad \\
            \leq \mathcal{G}^{\mathrm{L}}_{N_{\mathcal{B}},\epsilon^{\mathrm{L}}_{\leq I,\mathcal{B}}}\left(\sum_{i=0}^I p_i\right) \bigg] \leq \epsilon^{\mathrm{L}}_{\leq I,\mathcal{B}}.
    \end{split}
\end{equation}
This implies that the parameters $\delta_{i,\mathcal{B}}$ satisfy
\begin{equation}
    \delta^{\mathrm{L}}_{\leq I,\mathcal{B}}\leq \sum_{i=0}^I \delta_{i,\mathcal{B}}\leq \delta^{\mathrm{U}}_{\leq I,\mathcal{B}},\end{equation}
for
\begin{equation}
    \begin{split}
        \delta^{\mathrm{U}}_{\leq I,\mathcal{B}} & =N_{\mathcal{B}}\mathcal{G}^{\mathrm{U}}_{N_{\mathcal{B}},\epsilon^{\mathrm{U}}_{\leq I,\mathcal{B}}}\left(\sum_{i=0}^I p_i\right)-\sum_{i=0}^I p_iN_{\mathcal{B}}, \\
        \delta^{\mathrm{L}}_{\leq I,\mathcal{B}} & =N_{\mathcal{B}}\mathcal{G}^{\mathrm{L}}_{N_{\mathcal{B}},\epsilon^{\mathrm{L}}_{\leq I,\mathcal{B}}}\left(\sum_{i=0}^I p_i\right)-\sum_{i=0}^I p_iN_{\mathcal{B}},
    \end{split}
\end{equation}
except with probability at most $\sum_{i=0}^I(\epsilon^{\mathrm{U}}_{i,\mathcal{B}}+ \epsilon^{\mathrm{L}}_{i,\mathcal{B}})+\epsilon^{\mathrm{U}}_{\leq I,\mathcal{B}}+\epsilon^{\mathrm{L}}_{\leq I,\mathcal{B}}$.

Moreover, taking into account that $\sum_{i=0}^{\infty}N_{i,\mathcal{B}}=N_{\mathcal{B}}$ we find that $\sum_{i=0}^{\infty}\delta_{i,\mathcal{B}}=0$. This means that we also have the following constraint,
\begin{equation}
    \sum_{i=0}^{I}\delta_{i,\mathcal{B}}+\delta_{>I,\mathcal{B}}=0
\end{equation}
except with probability at most $\sum_{i=0}^I(\epsilon^{\mathrm{U}}_{i,\mathcal{B}}+ \epsilon^{\mathrm{L}}_{i,\mathcal{B}})+\epsilon^{\mathrm{U}}_{> I,\mathcal{B}}+\epsilon^{\mathrm{L}}_{> I,\mathcal{B}}$.

Now, we have all the necessary constraints to construct the LPs to evaluate  $n^{\mathrm{L}}_{1, Z}$ and $\phi^{\mathrm{U}}_{1, Z}$, as detailed in the following subsections.

\subsubsection{Lower bound on \texorpdfstring{$n_{1, Z}$}{n1,Z}}
Combining the constraints from the previous discussion, a lower bound $n_{1,Z}^{\mathrm{L}}$ on the parameter $n_{1,Z}$ that holds except with probability at most
\begin{equation}
    \begin{split}
        \epsilon_{1,Z}= & \sum_{i = 0}^{I} \left(\epsilon^{\mathrm{U}}_{i,Z}+\epsilon^{\mathrm{L}}_{i,Z}\right) + \sum_{j=0}^J\left(\epsilon^{\mathrm{U}}_{\mu_j,n_Z}+\epsilon^{\mathrm{L}}_{\mu_j,n_Z}\right) \\
                        & +\epsilon^{\mathrm{L}}_{\leq I,Z}+\epsilon^{\mathrm{U}}_{\leq I,Z}+\epsilon^{\mathrm{L}}_{>I,Z}+\epsilon^{\mathrm{U}}_{> I,Z},
    \end{split}
\end{equation}
is given by the following LP,
\begin{equation}
    \label{eq:LP_Z0}
    \begin{aligned}
        n_{1,Z}^{\mathrm{L}} = &  &  & \min        &  & n_{1,Z},                                                  \\
                               &  &  & \text{s.t.} &  & n_{\mu_j, Z} + \delta_{\mu_j,n_Z}
        \geq \sum_{i=0}^{I} p_{(\mu_j|i)} n_{i,Z},                                                                \\
                               &  &  &             &  & n_{\mu_j, Z} + \delta_{\mu_j,n_Z}
        \leq \sum_{i=0}^{I} p_{(\mu_j|i)} n_{i,Z}                                                                 \\
                               &  &  &             &  & +N_{Z}\!\left(1-\sum_{i=0}^{I} p_i\right)+\delta_{>I,Z},  \\
                               &  &  &             &  & 0 \leq n_{i,Z}
        \leq \min\!\left(p_iN_Z+\delta_{i,Z},\, n_Z\right),                                                       \\
                               &  &  &             &  & \sum_{j=0}^J \delta_{\mu_j,n_Z} = 0,                      \\
                               &  &  &             &  & \delta^{\mathrm{L}}_{\mu_j,n_Z}
        \leq \delta_{\mu_j,n_Z} \leq \delta^{\mathrm{U}}_{\mu_j,n_Z},                                             \\
                               &  &  &             &  & \delta^{\mathrm{L}}_{>I,Z}
        \leq \delta_{>I,Z} \leq \delta^{\mathrm{U}}_{>I,Z},                                                       \\
                               &  &  &             &  & \delta^{\mathrm{L}}_{i,Z}
        \leq \delta_{i,Z} \leq \delta^{\mathrm{U}}_{i,Z},                                                         \\
                               &  &  &             &  & \delta^{\mathrm{L}}_{\leq I,Z}
        \leq \sum_{i=0}^I \delta_{i,Z} \leq \delta^{\mathrm{U}}_{\leq I,Z},                                       \\
                               &  &  &             &  & \sum_{i=0}^{I}\delta_{i,Z}+\delta_{>I,Z}=0                \\
                               &  &  &             &  & \forall \, 0 \leq j \leq J,\; \forall \, 0 \leq i \leq I.
    \end{aligned}
\end{equation}

\subsubsection{Upper bound on \texorpdfstring{$\phi_{1, Z}$}{phi1,Z}}
An upper bound $\phi^{\mathrm{U}}_{1,Z}$ on the phase error rate of the single-photon events in the $Z$ basis, $\phi_{1,Z}$, can be obtained by casting it in terms of a random sampling problem, in which one observes $m_{1,X}$ errors in $n_{1,X}$ instances, randomly sampled from a population of size $N_{\mathrm{H}}:=n_{1,X}+n_{1,Z}$, with test-sample size $n_{\mathrm{H}}:=n_{1,X}$. Defining $e_{1,X}=m_{1,X}/n_{1,X}$, the single-photon bit error rate in the $X$ basis, and utilizing Theorem~\ref{Hyper} provided in Appendix~\ref{sec:bounds},  we have that
\begin{equation}
    \label{phiproper}
    \begin{split}
        \phi_{1,Z} \leq \frac{1}{n_{1,Z}} \bigl[ & (n_{1,Z}+n_{1,X}) \times                                                                                                \\
                                                 & \mathcal{H}^{\mathrm{U}}_{N_{\mathrm{H}},\,n_{\mathrm{H}},\,\epsilon_{\mathrm{S}}}\left(e_{1,X}\right) -m_{1,X} \bigr],
    \end{split}
\end{equation}
except with probability at most $\epsilon_{\mathrm{S}}$. The function $\mathcal{H}^{\mathrm{U}}_{N,n,\epsilon}$ is given by Eq.~\eqref{eq:hyperub}. Note that the parameters on the right-hand side are not observed during the experiment. In order to relate them to the observables, we define the quantity,

\begin{equation}
    \label{phiupperbound}
    \begin{split}
        \phi^{\mathrm{U}}_{1,Z} = \frac{1}{n_{1,Z}^{\mathrm{L}}} \bigl[ & (n_{1,Z}^{\mathrm{U}}+n_{1,X}^{\mathrm{U}}) \times                                                                                                                          \\
                                                                        & \mathcal{H}^{\mathrm{U}}_{N_{\mathrm{H}}^{\mathrm{U}},\,n_{\mathrm{H}}^{\mathrm{L}},\,\epsilon_{\mathrm{S}}}\left(e^{\mathrm{U}}_{1,X}\right) -m_{1,X}^{\mathrm{L}} \bigr],
    \end{split}
\end{equation}
where  \(N_{\mathrm{H}}^{\mathrm{U}}:=n_{1,Z}^{\mathrm{U}}+n_{1,X}^{\mathrm{U}}\) and \(n_{\mathrm{H}}^{\mathrm{L}}:=n_{1,X}^{\mathrm{L}}\) , \(e_{1,X}^{\mathrm{U}}:=m_{1,X}^{\mathrm{U}}/n_{1,X}^{\mathrm{L}}\), and the bounds \(m_{1,X}^{\mathrm{L}}\), \(m_{1,X}^{\mathrm{U}}\), \(n_{1,X}^{\mathrm{L}}\), \(n_{1,X}^{\mathrm{U}}\), and \(n_{1,Z}^{\mathrm{U}}\) entering Eq.~\eqref{phiupperbound} are constructed immediately below.
Following a similar reasoning as before, we find that an upper bound $m_{1,X}^{\mathrm{U}}$ on the number of single-photon errors in the $X$ basis, $m_{1,X}$, that holds except with probability at most
\begin{equation}
    \begin{split}
        \epsilon_{m,X} & = \sum_{i = 0}^{I} \left(\epsilon^{\mathrm{U}}_{i,X}+\epsilon^{\mathrm{L}}_{i,X}\right)
        + \sum_{j=0}^J\left(\epsilon^{\mathrm{U}}_{\mu_j,m_X}+\epsilon^{\mathrm{L}}_{\mu_j,m_X}\right)                                                  \\
                       & +\epsilon^{\mathrm{U}}_{\leq I,X}+\epsilon^{\mathrm{L}}_{\leq I,X}+\epsilon^{\mathrm{U}}_{>I,X}+\epsilon^{\mathrm{L}}_{> I,X},
    \end{split}
\end{equation}
is given by the following LP,
\begin{equation}
    \label{eq:LP_E}
    \begin{aligned}
        m_{1,X}^{\mathrm{U}} = &  &  & \max        &  & m_{1,X}                                                                                          & \\
                               &  &  & \text{s.t.} &  & m_{\mu_j,X}+ \delta_{\mu_j,m_X} \geq \sum_{i=0}^{I} p_{\left(\mu_j|i\right)} m_{i,X},              \\
                               &  &  &             &  & m_{\mu_j,X}+ \delta_{\mu_j,m_X} \leq \sum_{i=0}^{I} p_{\left(\mu_j|i\right)} m_{i,X}               \\
                               &  &  &             &  & + N_{X}\left(1-\sum_{i=0}^{I} p_i\right) +\delta_{>I,X},                                           \\
                               &  &  &             &  & 0 \leq m_{i,X} \leq \min\left(p_iN_X+\delta_{i,X}, m_X \right),                                    \\
                               &  &  &             &  & \sum_{j=0}^J \delta_{\mu_j,m_X} =0,                                                                \\
                               &  &  &             &  & \delta^{\mathrm{L}}_{\mu_j,m_X}\leq \delta_{\mu_j,m_X} \leq \delta^{\mathrm{U}}_{\mu_j,m_X},       \\
                               &  &  &             &  & \delta^{\mathrm{L}}_{>I,X}\leq \delta_{>I,X} \leq \delta^{\mathrm{U}}_{>I,X},                      \\
                               &  &  &             &  & \delta^{\mathrm{L}}_{i,X}\leq \delta_{i,X} \leq \delta^{\mathrm{U}}_{i,X},                         \\
                               &  &  &             &  & \delta^{\mathrm{L}}_{\leq I,X}\leq \sum_{i=0}^I \delta_{i,X}\leq \delta^{\mathrm{U}}_{\leq I,X},   \\
                               &  &  &             &  & \sum_{i=0}^{I}\delta_{i,X}+\delta_{>I,X}=0                                                         \\
                               &  &  &             &  & \forall \, 0 \leq j \leq J,\; \forall \, 0 \leq i \leq I.
    \end{aligned}
\end{equation}

A lower bound $m^{\mathrm{L}}_{1,X}$ on the number of single-photon errors in the $X$ basis, $m_{1,X}$, that holds except with probability $\epsilon_{m,X}$ can also be obtained by minimizing the LP given by Eq.~\eqref{eq:LP_E}. Similarly, a lower bound $n_{1,X}^{\mathrm{L}}$ on the parameter $n_{1,X}$ that holds except with probability at most
\begin{equation}
    \begin{split}
        \epsilon_{1,X} & = \sum_{i = 0}^{I} \left(\epsilon^{\mathrm{U}}_{i,X}+\epsilon^{\mathrm{L}}_{i,X}\right)                                       \\
                       & + \sum_{j=0}^J\left(\epsilon^{\mathrm{U}}_{\mu_j,n_X}+\epsilon^{\mathrm{L}}_{\mu_j,n_X}\right)                                \\
                       & +\epsilon^{\mathrm{U}}_{\leq I,X}+\epsilon^{\mathrm{L}}_{\leq I,X}+\epsilon^{\mathrm{U}}_{>I,X}+\epsilon^{\mathrm{L}}_{> I,X}
    \end{split}
\end{equation}
is given by the LP obtained by changing the basis from $Z$ to $X$ in Eq.~(\ref{eq:LP_Z0}).

The parameters $m^{\mathrm{U}}_{1,X}$ and $n^{\mathrm{L}}_{1,X}$ can be used to estimate an upper bound $e_{1,X}^{\mathrm{U}}$ on $e_{1,X}$,  as
\begin{equation}
    \label{erroreqn}
    e_{1,X} \leq e_{1,X}^{\mathrm{U}}=  \frac{m_{1,X}^{\mathrm{U}}}{n_{1,X}^{\mathrm{L}}},
\end{equation}
that holds except with probability at most
\begin{equation}
    \begin{split}
        \epsilon_{e,X} & = \epsilon_{m,X}+\epsilon_{1,X}                                                                                                                                  \\
                       & = \sum_{i = 0}^{I} (\epsilon^{\mathrm{U}}_{i,X}+\epsilon^{\mathrm{L}}_{i,X}) + \sum_{j=0}^J(\epsilon^{\mathrm{U}}_{\mu_j,n_X}+\epsilon^{\mathrm{L}}_{\mu_j,n_X}) \\
                       & + \sum_{j=0}^J(\epsilon^{\mathrm{U}}_{\mu_j,m_X}+\epsilon^{\mathrm{L}}_{\mu_j,m_X})                                                                              \\
                       & +\epsilon^{\mathrm{U}}_{\leq I,X}+\epsilon^{\mathrm{L}}_{\leq I,X}+\epsilon^{\mathrm{U}}_{>I,X}+\epsilon^{\mathrm{L}}_{> I,X}.
    \end{split}
\end{equation}
As for the parameters $n^{\mathrm{U}}_{1,Z}$ and $n^{\mathrm{U}}_{1,X}$ in Eq.~(\ref{phiupperbound}), they can be evaluated by maximizing the  LP used to obtain the corresponding lower bounds. The failure probability is equal to that associated with the corresponding lower bound.

Now, we have all the ingredients needed to evaluate $\phi^{\mathrm{U}}_{1,Z}$. Comparing Eqs.~(\ref{phiproper}) and (\ref{phiupperbound}), we have that from Boole's inequality~\cite{bonferroni1936teoria},
\begin{equation}
    \phi_{1,Z}\leq \phi^{\mathrm{U}}_{1,Z}
\end{equation}
except with probability at most $\epsilon_{\mathrm{PE}}=\epsilon_{1,Z}+\epsilon_{e,X}+\epsilon_{\mathrm{S}}$.

\subsubsection{Security parameter allocation}
The total security parameter of the protocol, $\epsilon_{\mathrm{tot}}$, is  related to $\epsilon_{\mathrm{PE}}$, $\epsilon_{\mathrm{cor}}$, and $\epsilon_{\mathrm{PA}}$ via Eq.~\eqref{eq:eqtot}. For our simulations, we  set for simplicity $\epsilon_{\mathrm{cor}}= \epsilon_{\mathrm{PA}} = 10^{-3}\epsilon_{\mathrm{tot}}$ and allocate the remaining budget to the failure probability of the parameter estimation step, resulting in $\epsilon_{\mathrm{PE}} \approx \epsilon_{\mathrm{tot}}^2/4$. This failure probability $\epsilon_{\mathrm{PE}}$ must be further distributed among the $S_n$ independent statistical bounds used to construct the LPs. For simplicity, we set the failure probability for each individual constraint to $\epsilon_{\mathrm{ind}} = \epsilon_{\mathrm{PE}}/S_n$.

For a full implementation of the LPs described in Eqs.~\eqref{eq:LP_Z0} and \eqref{eq:LP_E}, considering a truncation level $I$ and $J+1$ intensities, the total number of statistical constraints includes:
$4(I+1)$ bounds for the photon number deviations $\delta_{i,Z}$ and $\delta_{i,X}$;
$4(J+1)$ bounds for the intensity count deviations $\delta_{\mu_j, n_Z}$ and $\delta_{\mu_j, n_X}$;
$2(J+1)$ bounds for the intensity error deviations $\delta_{\mu_j, m_X}$;
eight bounds for the truncation terms ($\delta_{>I,X}$, $\delta_{>I,Z}$, $\delta_{\leq I,X}$, and $\delta_{\leq I,Z}$); and one additional constraint for the hypergeometric bound. Summing these contributions for both bases, the total number of constraints scales as $S_n = 4I + 31$ for the relevant case $J=2$.
%For a typical truncation of $I=10$, this results in over 70 constraints, significantly penalizing $\epsilon_{\rm ind}$.

In the high-loss regime characteristic of geostationary Earth orbit (GEO) satellite links, the total number of detections is generally small. In this specific regime, it turns out that the explicit constraints on the individual photon number deviations $\delta_{i,\mathcal{B}}$ provide negligible tightening of the feasible region of $n_{i,\mathcal{B}}$ (and $m_{i,\mathcal{B}}$) while consuming a large portion of the epsilon budget. In addition, one can  use $\delta^{\mathrm{U}}_{>I,\mathcal{B}}$ directly instead of relying on the variable $\delta_{>I,\mathcal{B}}$. Hence, for our simulations, we employ a \emph{reduced} LP formulation that omits the  constraints on $\delta_{i,\mathcal{B}}$ and the lower constraint on $\delta_{>I,\mathcal{B}}$. This simplification is strictly conservative as it only relaxes the LP (potentially yielding a lower key rate) and reduces the constraint count to a constant $S_n = 21$, independent of $I$. This reduction in $S_n$ allows for a tighter $\epsilon_{\mathrm{ind}}$, improving the overall key rate performance in the finite-size regime.

%% file: journal/sections/appendix_detection_modes_and_asymmetric_passive_bb84.tex
% Combined appendix sections: active detection, passive detection, and asymmetric passive BB84.
% Improve line breaking for dense two-column appendix text.
\setlength{\emergencystretch}{1.2em}
\sloppy

% Include-safe local helpers used by the asymmetric passive appendix.
\providecommand{\rv}[1]{\bm{#1}} % random variable (bold), including subscripts/superscripts
\providecommand{\mc}{\mathrm{mc}}
\makeatletter
\@ifundefined{c@lemma}{\newtheorem{lemma}[theorem]{Lemma}}{}
\makeatother

\section{Asymmetric active BB84 receiver}
\label{sec:asym_active_BB84_rec}

Assuming unpolarized background light and an asymmetric active BB84 receiver with two detectors, the total mean number of background photons collected within a temporal window $\Delta t$ and spectral bandwidth $\Delta\lambda$ is $\bar{n}_{\mathrm{B}}$ (given in Eq.~\eqref{eq:bgphot}). Accounting for the receiver efficiency, the mean photon number admitted to the detection path is $\eta_{\mathrm{R}}\eta_{\mathrm{D}}\bar{n}_{\mathrm{B}}$, where $\eta_{\mathrm{R}}$ is the receiver optics/filter transmission and $\eta_{\mathrm{D}}$ is the detection efficiency of Bob's detectors, which we assume is the same for both of them.

Because the background field is unpolarized and split equally by the polarizing beam splitter (PBS), each detector receives on average half the number of  photons. The resulting per-detector per-gate noise-click probability is therefore
\begin{equation}
    p_{\mathrm{noise}} = 1 - \exp\!\left(-\frac{\eta_{\mathrm{R}} \eta_{\mathrm{D}}\bar{n}_{\mathrm{B}}}{2}\right)\,(1 - p_{\mathrm{dark}}),
    \label{eq:pec}
\end{equation}
where $p_{\mathrm{dark}}$ is the intrinsic dark-count probability per gate of the detector. We assume that $p_{\mathrm{dark}}$ is the same for both detectors.
% The exponential form arises from the Poisson statistics of photon arrivals: the probability of detecting no background photons is $\exp(-\eta_R\bar{n}_{\rm B}/2)$, so the complement gives the probability of one or more background detections.

Given a signal intensity $\mu$, the joint probability of a detection event arises from either a signal or noise.
The probability that at least one signal photon is detected is $1 - \exp(-\eta_{\mathrm{sys}} \mu)$, where $\eta_{\mathrm{sys}}$ is the overall system transmission efficiency, given in Eq.~\eqref{eq:totloss}.
The total probability of a detection (excluding afterpulses) at Bob's receiver is therefore
\begin{equation}
    D_\mu = 1 - (1 - p_{\mathrm{noise}})^2\,\exp(-\eta_{\mathrm{sys}} \mu).
    \label{eq:Dmu}
\end{equation}

Afterpulses are spurious detection clicks triggered by carriers trapped during a previous detection event; consequently, the probability of an afterpulse in a certain time window depends on whether a detection occurred in the preceding window. Given a previous click occurring with probability $D_\mu$, an afterpulse is produced with conditional probability $p_{\mathrm{ap}}$.
The total probability of no afterpulse click in the current window is therefore $1 - D_\mu p_{\mathrm{ap}}$.
This modifies the no-click probabilities for each of the two detectors.

For the \emph{correct} detector (i.e., the one aligned with Alice's transmitted polarization state), the no-click probability from signal and dark counts is
\begin{equation}
    p_{\mathrm{nc},c} = (1 - p_{\mathrm{noise}})\,\exp\!\left[-\eta_{\mathrm{sys}}\mu\, (1 - e_{\mathrm{mis}})\right],
\end{equation}
where $e_{\mathrm{mis}}$ is the optical misalignment error rate.
This latter parameter represents the fraction of signal detections that yields an incorrect bit value despite the detection of a genuine photon.
In satellite-based quantum key distribution (QKD) channels, this intrinsic error arises from several physical mechanisms:
(i) imperfect alignment between the transmitter and receiver polarization bases due to residual calibration errors;
(ii) time-dependent polarization rotation caused by relative satellite motion and the changing orientation of the optical link;
(iii) finite polarization extinction ratio and contrast of beam splitters and polarizing optics;
and (iv) birefringence introduced by telescope mirrors, optical windows, and fibers along the path.
These effects cause deviations between the transmitted and measured polarization states, even when noise is negligible. See Appendix~\ref{sec:misalignment} for a  discussion.

Including the afterpulse contribution, the total no-click probability becomes
\begin{equation}
    P_{\mathrm{nc},c} = p_{\mathrm{nc},c}\,(1 - D_\mu p_{\mathrm{ap}}).
    \label{eq:Pnc_correct}
\end{equation}

Similarly, for the \emph{wrong} detector (i.e., the one associated with a signal orthogonal to Alice's transmitted state), the no-click probability is
\begin{equation}
    p_{\mathrm{nc},w} = (1 - p_{\mathrm{noise}})\,\exp\!\left(-\eta_{\mathrm{sys}}\, \mu\, e_{\mathrm{mis}}\right),
\end{equation}
with the total no-click probability (including the effect of afterpulses) given by
\begin{equation}
    P_{\mathrm{nc},w} = p_{\mathrm{nc},w}\,(1 - D_\mu p_{\mathrm{ap}}).
    \label{eq:Pnc_wrong}
\end{equation}

The corresponding click probabilities are $P_{\mathrm{c},c} = 1 - P_{\mathrm{nc},c}$ and $P_{\mathrm{c},w} = 1 - P_{\mathrm{nc},w}$.

The detection yield (i.e., the probability of observing a  valid detection event in a given time window at Bob's receiver) is given by
\begin{equation}
    Q_\mu = 1 - P_{\mathrm{nc},c}\,P_{\mathrm{nc},w},
    \label{eq:Qmu_detailed}
\end{equation}
accounting for the possibility that either or both detectors click.

Errors arise when either (i) the wrong detector clicks while the correct detector does not (always an error), or (ii) both detectors click simultaneously, in which case the bit is assigned randomly, resulting in an error probability of $1/2$.
The error probability per detection is therefore given by
\begin{equation}
    E_\mu = P_{\mathrm{c},w}\,P_{\mathrm{nc},c} + \tfrac{1}{2}\,P_{\mathrm{c},w}\,P_{\mathrm{c},c}
    = \tfrac{1}{2}\,P_{\mathrm{c},w}\,(1 + P_{\mathrm{nc},c}),
    \label{eq:Emu_detailed}
\end{equation}
with the quantum bit error rate (QBER) given by the ratio $E_\mu/Q_\mu$.

\section{Asymmetric passive BB84 protocol}
\label{sec:bb84_passive_asym}

Here, we particularize the methods introduced in Ref. \cite{wang2025phase} to the case where Bob uses an asymmetric passive receiver with a fixed beam-splitting ratio $s\in(0,1)$ to the $X$ arm (and $1-s$ toward the $Z$ arm), together with four threshold memoryless single-photon detectors with identical detection efficiency $\eta_{\mathrm{D}}$ and per-gate dark-count probability $p_{\mathrm{dark}}$. The key difficulty for asymmetric passive receivers remains that the probability of being routed to the $X$ or $Z$ arm depends on the photon number entering Bob's device and is therefore, in general, under Eve's control. As a result, the standard BB84 random-sampling argument for phase-error estimation does not apply here directly.

The security proof strategy presented in Ref. \cite{wang2025phase} circumvents this by working within Bob's side photon-number subspaces. Because Bob's POVM is block-diagonal in the Fock basis, one may insert at Bob's receiver a virtual QND measurement without changing any observed classical statistics. We denote the QND measurement outcome by $\rv{\tilde{\imath}}\in\tilde{A}$, where $\tilde{A}=\{0, 1, {>}1\}$. The final key length calculation requires two ingredients: a lower bound on the number of key-generation rounds when $\rv{\tilde{\imath}}=1$, and an upper bound on the phase-error rate on that same set of rounds. We reuse the same protocol structure and security reduction from \cite{wang2025phase}, and adapt it to the scenario considered here: (a) fixed detector parameters $(s,\eta_{\mathrm{D}},p_{\mathrm{dark}})$, (b) explicit conditioning on the photon number subspace that Bob receives, and (c) exact binomial/hypergeometric inversions for the finite-size bounds used to determine the key length. Table~\ref{tab:asym_passive_notation} summarizes the notation used in this section.

\begin{table}[htbp]
    \centering
    \caption{Notation used in the asymmetric passive finite-key analysis. Here \(\mathcal{B}\in\{X,Z\}\), \(k\) is the intensity setting, \(1_A\) denotes Alice-emitted single photons, \(1_B\) denotes the single-photon outcome in the virtual quantum non-demolition (QND) measurement on Bob's side, and tildes denote virtual quantities.}
    \label{tab:asym_passive_notation}
    \renewcommand{\arraystretch}{1.1}
    \scriptsize
    \begin{tabularx}{\columnwidth}{|>{\raggedright\arraybackslash}p{0.31\columnwidth}|>{\raggedright\arraybackslash}X|}
        \hline
        \textbf{Symbol}                         & \textbf{Meaning}                                                                                                                                                                                                                                \\
        \hline
        \(\rv{n_{k,\mathcal{B}}}\)              & Single-click counts in basis \(\mathcal{B}\) at intensity \(k\).                                                                                                                                                                                \\
        \hline
        \(\rv{m_{k,X}}\)                        & Error counts in \(X\)-basis test rounds at intensity \(k\).                                                                                                                                                                                     \\
        \hline
        \(\rv{n_{1_A,\mathcal{B}}}\)            & Single-click counts in \(\mathcal{B}\)-basis rounds in which Alice emitted a single photon; superscript \(L/U\) denotes lower/upper-bound counts.                                                                                               \\
        \hline
        \(\rv{m_{1_A,X}}\)                      & Error counts in \(X\)-basis test rounds in which Alice emitted a single photon; superscript \(L/U\) denotes lower/upper-bound error counts.                                                                                                     \\
        \hline
        \(\rv{\tilde n_{1_A,1_B,\mathcal{B}}}\) & Virtual counts of \(\mathcal{B}\)-basis single-click rounds in which Alice emitted one photon (\(1_A\)) and Bob's virtual QND measurement has a single-photon outcome (\(1_B\)); untilded superscript \(L/U\) denotes lower/upper-bound counts. \\
        \hline
        \(\rv{\tilde m_{1_A,1_B,X}}\)           & Virtual error counts in \(X\)-basis test rounds in which Alice emitted one photon and Bob's virtual QND measurement has a single-photon outcome; untilded superscript \(L/U\) denotes lower/upper-bound error counts.                           \\
        \hline
        \(\rv{\tilde \phi_{1_A,1_B,Z}}\)        & Virtual phase-error rate for \(Z\)-basis rounds in which Alice emitted one photon and Bob's virtual QND measurement has a single-photon outcome; untilded superscript \(U\) denotes its upper bound.                                            \\
        \hline
        \(\rv{\tilde n_0}\)                     & Counts of Bob-side zero-photon QND outcomes; \(\rv{\tilde n_{0,\mathcal{B}}}\) restricts to \(\mathcal{B}\)-basis single-click rounds.                                                                                                          \\
        \hline
        \(\rv{\tilde n_{>1}}\)                  & Counts of Bob-side more-than-one-photon QND outcomes; adding the subscript ``\(\mc\)'' restricts to multi-click outcomes.                                                                                                                       \\
        \hline
        \(\rv{n_{\mc}}\)                        & Observed total multi-click counts, where two or more detectors click in the same round.                                                                                                                                                         \\
        \hline
    \end{tabularx}
    \renewcommand{\arraystretch}{1.0}
\end{table}

\subsection{Protocol description}
\label{sec:asym_passive_description}
Like in the case of the asymmetric active receiver analyzed in Appendix~\ref{sec:protocol_appendix}, we shall consider that the $Z$ basis is used for key generation and the $X$ basis for testing.

Alice fixes the basis probabilities $p_X^A$ and $p_Z^A:=1-p_X^A$, and a set of decoy intensities $\mathcal{K}:=(\mu_0,\ldots,\mu_J)$ with probabilities $\{p_{\mu_0},\ldots,p_{\mu_J}\}$. Bob fixes the beam-splitting ratio $s\in(0, 1)$. In addition to the privacy amplification and error verification failure  parameters $\epsilon_{\mathrm{PA}}$ and $\epsilon_{\mathrm{cor}}$, respectively,  we allocate the following failure budgets: $\epsilon_{0Z}$ (associated with the vacuum contribution in the key set), $\epsilon_{0X}$ (associated with the vacuum contribution in the test set), $\epsilon_{\mathrm{PNE}}$ (associated with the photon-number estimation via multi-click events), and $\epsilon_{\mathrm{S}}$ (associated with the hypergeometric sampling for phase errors). Like in Appendix~\ref{sec:protocol_appendix}, we obtain the decoy-state bounds on Alice-emitted photon-number states using LP. For convenience, we aggregate all statistical failure terms entering the LP constraints into a single failure parameter $\epsilon_{\mathrm{decoy}}$. In addition, we combine this into a single parameter-estimation failure probability
\begin{equation}
    \epsilon_{\mathrm{PE}}
    :=
    \epsilon_{0Z} + \epsilon_{0X} + \epsilon_{\mathrm{PNE}}  + \epsilon_{\mathrm{decoy}} + \epsilon_{\mathrm{S}}.
    \label{eq:epsilon_PE_def}
\end{equation}

Below, tilded quantities denote unobserved random variables conditioned on the virtual QND measurement outcomes, denoted by the subscripts $0,1, {>}1$. For both tilded and untilded quantities, we distinguish random variables from their realized values using bold and plain symbols, respectively.

The protocol runs as follows. For $r=1,\ldots,N$, steps 1 and 2 below are repeated.
\begin{enumerate}
    \item \emph{State preparation}:
          Alice chooses a bit value $\rv{y_r}$ uniformly at random, a basis value $\rv{a_r}\in\{Z,X\}$ with probability $p_{\rv{a_r}}^A$,
          and the intensity value $\rv{k_r}\in\mathcal{K}$ with probability $p_{\rv{k_r}}$.
          She generates a PRWCP encoded with the settings $(\rv{a_r},\rv{y_r})$ at intensity $\rv{k_r}$ and sends it to Bob.

    \item \emph{Asymmetric passive  measurement}:
          Bob’s passive module outputs a click pattern on its four detectors.
          Define the coarse-grained outcome
          $\rv{o_r}\in\{(X,0),(X, 1),(Z,0),(Z, 1),\emptyset,\mc\}$ where
          $\emptyset$ refers to no click,
          $(X,y)$ corresponds to exactly one click in the $X$ arm giving the bit value $y\in\{0, 1\}$,
          $(Z,y)$ indicates exactly one click in the $Z$ arm giving the bit value $y\in\{0, 1\}$ and
          $\mc$ refers to two or more detectors producing a click.
          If $\rv{o_r}\in\{X,Z\}\times\{0, 1\}$, write $\rv{o_r}=(\rv{b_r},\rv{y_r'})$ with $\rv{b_r}\in\{X,Z\}$ and $\rv{y_r'}\in\{0, 1\}$; otherwise set $\rv{y_r'}=\emptyset$.

    \item \emph{Sifting}:
          Alice reveals $\{\rv{a_r}\}$ and $\{\rv{k_r}\}$.
          Bob reveals whether the round outcome is $\emptyset$ (no click), $\mc$ (multi-click), or a valid single click. In addition, in this latter case he reveals the basis label $\rv{b_r}\in\{X,Z\}$.
          They define the key/test sets as  follows
          \begin{equation}
              \begin{aligned}
                  \mathcal{Z}^{\mathrm{key}}
                   & := \{r:\ \rv{a_r}=Z,                 \\
                   & \qquad \rv{o_r}\in\{(Z,0),(Z,1)\}\}, \\
                  \mathcal{X}^{\mathrm{test}}
                   & := \{r:\ \rv{a_r}=X,                 \\
                   & \qquad \rv{o_r}\in\{(X,0),(X,1)\}\}.
              \end{aligned}
          \end{equation}

          that is, both sets contain basis-matched single-click rounds.
          In addition, they define the intensity-resolved subsets
          $\mathcal{Z}^{\mathrm{key}}_{k}:=\{r\in\mathcal{Z}^{\mathrm{key}}:\rv{k_r}=k\}$ and
          $\mathcal{X}^{\mathrm{test}}_{k}:=\{r\in\mathcal{X}^{\mathrm{test}}:\rv{k_r}=k\}$ with sizes
          $n_{k,Z}:=|\mathcal{Z}^{\mathrm{key}}_{k}|$ and $n_{k,X}:=|\mathcal{X}^{\mathrm{test}}_{k}|$.
          They also record the total multi-click counts as
          \begin{equation}
              \rv{n_{\mc}}:=|\{r:\ \rv{o_r}=\mc\}|.
          \end{equation}

    \item \emph{Parameter estimation} (PE):
          For each $k\in\mathcal{K}$, Alice and Bob disclose the bit values on $\mathcal{X}^{\mathrm{test}}_{k}$ and compute the test-bit errors $\rv{m_{k,X}}$. Using the methods presented in Appendix~\ref{sec:estimation}, they calculate $\rv{n^L_{1_A,1_B,Z}}$ (lower bound on the key rounds in which Alice emitted and Bob's virtual QND measurement output a single photon) and $\rv{\phi_{1_A,1_B,Z}^{U}}$ (upper bound on the phase error rate in those rounds), which depend on the requested PE error, $\epsilon_{\mathrm{PE}}$.
          % They also randomly choose a  subset $\mathcal{Z}^{\rm EC}\subseteq \mathcal{Z}^{\rm keep}$ to estimate the bit error rate. That is, they compute the observed $Z$-basis error rate $e_Z$ on $\mathcal{Z}^{\rm EC}$, and define the post-revelation key set
          % \begin{equation}
          %     \mathcal{Z}^{\rm key}:=\mathcal{Z}^{\rm keep}\setminus \mathcal{Z}^{\rm EC}.
          % \end{equation}

          % For each intensity $k$, define $\mathcal{Z}^{\rm key}_{k}:=\{r\in\mathcal{Z}^{\rm key}:\rv{k_r}=k\}$ and $\rv {n_{k,Z}}:=|\mathcal{Z}^{\rm key}_{k}|$.

    \item \emph{Variable-length decision}:
          Alice and Bob calculate the number of bits to be used for error correction, $\lambda_{\mathrm{EC}}$. The final key $\ell$ is given by
          \begin{align}
              \rv{\ell}
               & :=
              \max\Biggl\{
              \Bigl\lfloor
              \rv{n^L_{1_A,1_B,Z}}\bigl[1-h(\rv{\phi_{1_A,1_B,Z}^{U}})\bigr]
              - \lambda_{\mathrm{EC}} \notag                                                              \\
               & \qquad - \log_2\!\left(\frac{1}{2\epsilon^2_{\mathrm{PA}}\epsilon_{\mathrm{cor}}}\right)
              \Bigr\rfloor
              ,\,0\Biggr\}.
              \label{eq:keylen_passive_exact}
          \end{align}

          Aborting is modeled as producing a key of length zero. If they do not abort, they proceed to steps 6, 7 and 8.

    \item{}\emph{Error correction} (EC): Alice and Bob implement a one-way EC protocol that reveals  $\lambda_{\mathrm{EC}}$ bits  of syndrome information.

    \item{}\emph{Error verification} (EV): Alice and Bob perform an EV step based on 2-universal hashing, using EV tags of length $\log(2/\epsilon_{\mathrm{cor}})$ bits at most.  If the EV tags do not match, they abort the protocol. Otherwise, they proceed to privacy amplification.

    \item{}\emph{Privacy amplification} (PA): Alice and Bob perform a PA step based on 2-universal hashing, with failure probability at most $\epsilon_{\mathrm{PA}}$. In doing so they obtain an $\epsilon_{\mathrm{tot}}$-secure output key of length $\ell$, with
          \begin{equation}
              \label{eq:asymeqtot}
              \epsilon_{\mathrm{tot}} = 2\sqrt{\epsilon_{\mathrm{PE}}}+\epsilon_{\mathrm{PA}}+\epsilon_{\mathrm{cor}}.
          \end{equation}

\end{enumerate}

\subsection{Decoy-state analysis}
\label{sec:estimation}
Let $\rv{n_{1_A,Z}}$ and $\rv{n_{1_A,X}}$ be the random variables counting the key/test rounds in which Alice emitted a single-photon, and $\rv{m_{1_A,X}}$  the  random variable counting the $X$-basis bit errors among the single-photon rounds. Using the LPs presented in Appendix~\ref{dsa}, Alice and Bob estimate lower and upper bounds of these random variables,
\begin{align}
     & \rv{n_{1_A,Z}^{L}} \le \rv{n_{1_A,Z}} \le \rv{n_{1_A,Z}^{U}}, \\
     & \rv{n_{1_A,X}^{L}} \le \rv{n_{1_A,X}} \le \rv{n_{1_A,X}^{U}}, \\
     & \rv{m_{1_A,X}^{L}} \le \rv{m_{1_A,X}} \le \rv{m_{1_A,X}^{U}},
\end{align}
that hold except with joint failure probability at most $\epsilon_{\mathrm{decoy}}$.

\medskip

\paragraph{Upper bound on the vacuum contributions.}
Let $\rv{\tilde n_{0}}$ denote the  number of rounds in which Bob's virtual QND measurement outputs zero photons, and let $\rv{\tilde n_{0,Z}}$ denote the  number of events in the key set arising from those vacuum rounds.
Let $q_Z$ denote the vacuum single-click parameter, i.e., the maximum probability that an event is assigned to the key set when the virtual QND measurement outputs zero photons. When all detectors have the same dark count probability $p_{\mathrm{dark}}$, we have that
\begin{equation}
    q_Z = p_Z^A\,2p_{\mathrm{dark}}(1-p_{\mathrm{dark}})^3.
    \label{eq:qZ_identical}
\end{equation}
Then, it is shown in Ref.~\cite{wang2025phase} that for all $\tilde n_{0}\in\{0,\ldots,N\}$ and all $b\in\{\lceil \tilde n_0 q_Z \rceil+1,\ldots,N\}$, it is satisfied that
\begin{equation}
    \Pr\!\left(\rv{\tilde n_{0,Z}} \ge b \mid \rv{\tilde n_{0}}=\tilde n_{0}\right)
    \le
    \Pr\!\left[\mathrm{Bin}(\tilde n_{0},q_Z)\ge b\right].
    \label{eq:vacuum_key_domination}
\end{equation}
Since $\rv{\tilde n_{0}}\le N$, we have that $\Pr\left[\mathrm{Bin}(\tilde n_{0},q_Z)\ge b\right]\le \Pr\left[\mathrm{Bin}(N,q_Z)\ge b\right]$, and therefore for any $b\in\{\lceil N q_Z\rceil+1,\ldots,N\}$,
\begin{multline}
    \Pr\!\left(\rv{\tilde n_{0,Z}} \ge b\right)
    \le
    \Pr\!\left[\mathrm{Bin}(N,q_Z)\ge b\right] \\
    =
    \sum_{i=b}^{N}\binom{N}{i}q_Z^{i}(1-q_Z)^{N-i}.
    \label{eq:vacuum_key_tail}
\end{multline}
Fixing a target failure probability $\epsilon_{0Z}$, one can invert the exact upper tail of the binomial distribution and define the smallest threshold $b$ for which $\Pr\left[\mathrm{Bin}(N,q_Z)\ge b\right]\le \epsilon_{0Z}$. This yields the exact quantile

\begin{multline}
    B_0^{\mathrm{ex}}(N,q,\epsilon)
    :=
    \min\Biggl\{
    b\in\{\lceil N q\rceil+1,\dots,N+1\}: \\
    \sum_{i=b}^{N}\binom{N}{i}q^{i}(1-q)^{N-i}
    \le \epsilon
    \Biggr\}.
    \label{eq:B0_exact_def}
\end{multline}
By definition, $B_0^{\mathrm{ex}}(N,q,\epsilon)$ is the smallest integer threshold such that $\Pr[\mathrm{Bin}(N,q)\ge B_0^{\mathrm{ex}}(N,q,\epsilon)]\le \epsilon$. Combining this with Eq.~\eqref{eq:vacuum_key_tail}, one obtains
\begin{equation}
    \Pr\!\left[\rv{\tilde n_{0,Z}} \ge B_0^{\mathrm{ex}}(N,q_Z,\epsilon_{0Z})\right]\le \epsilon_{0Z}.
    \label{eq:vacuum_key_upper_exact}
\end{equation}
The analogous bound for the $X$-arm single-click operator can be obtained by exchanging $Z\leftrightarrow X$, taking into account that
\begin{equation}
    q_X = p_X^A\,2p_{\mathrm{dark}}(1-p_{\mathrm{dark}})^3.
    \label{eq:qX_identical}
\end{equation}
Precisely, let $\rv{\tilde n_{0,X}}$ denote the random variable counting the $X$-arm single-click outcomes arising from Bob-side vacuum signals. We have that
\begin{equation}
    \Pr\!\left[\rv{\tilde n_{0,X}} \ge B_0^{\mathrm{ex}}(N,q_X,\epsilon_{0X})\right]\le \epsilon_{0X}.
    \label{eq:vacuum_test_upper_exact}
\end{equation}

\medskip

\paragraph{Upper bound on the multi-photon contributions.}

Let $\rv{\tilde n_{>1}}$ denote the number of pulses with more than one photon that arrive at Bob's side, and let $\rv{n_{\mc}}$ denote the number of multi-click outcomes. Also, let $\lambda_{\min}\in(0,1]$ be a lower bound on the multi-click probability conditioned on Bob receiving a pulse with more than one photon. Precisely, one can take
\begin{equation}
    \lambda_{\min} := \eta_{\mathrm{D}}^2\,2s(1-s),
    \label{eq:lambda_min_closed_form}
\end{equation}
which corresponds to the case where Bob receives a two photon pulse~\cite{wang2025phase}.
Let $\rv{\tilde n_{\mc,>1}}$ be the  number of multi-click outcomes originating from the ${>}1$-photon subspace. This quantity trivially satisfies
\begin{equation}
    \rv{n_{\mc}} \ge \rv{\tilde n_{\mc,>1}}.
    \label{eq:mc_domination}
\end{equation}
In this context it can be shown that for all
$t\in\{0,\ldots,N\}$ and $k\in\{0,\ldots,N\}$, we have \cite{wang2025phase}
\begin{align}
    \Pr(\rv{n_{\mc}}\le k \mid \rv{\tilde n_{>1}}=t)
     & \le
    \Pr(\rv{\tilde n_{\mc,>1}}\le k \mid \rv{\tilde n_{>1}}=t) \notag \\
     & \le
    \Pr\!\left[\mathrm{Bin}(t,\lambda_{\min})\le k\right] \notag      \\
     & =
    \sum_{i=0}^{k} \binom{t}{i}\lambda_{\min}^{i} (1-\lambda_{\min})^{t-i}.
    \label{eq:pne_binom_cdf}
\end{align}

Now define the exact inversion (i.e., an upper confidence bound on the number of trials) by
\begin{multline}
    V^{\mathrm{ex}}(k,\epsilon,\lambda)
    := \\
    \min\!\Biggl(
    \left\{
    t\in\{k,\dots,N\}:\;
    \sum_{i=0}^{k} \binom{t}{i}\lambda^{i}(1-\lambda)^{t-i}
    \le \epsilon
    \right\} \\
    \cup \{N+1\}
    \Biggr).
    \label{eq:V_exact_def}
\end{multline}
In words, $V^{\mathrm{ex}}(k,\epsilon,\lambda)$ is the smallest trial count whose binomial cumulative distribution function at $k$ is at most $\epsilon$; any larger value is also valid but looser. Now we state an exact upper bound on the multi-photon contributions. The bound is exact in the sense that it inverts the binomial tail directly rather than relying on looser exponential bounds, i.e., Chernoff bounds.
\begin{lemma}[Exact upper bound on the multi-photon contributions]
    \label{lem:pne_exact}
    With the parameter $V^{\mathrm{ex}}(k,\epsilon,\lambda)$ as defined in Eq.~\eqref{eq:V_exact_def}, we have that
    \begin{equation}
        \Pr\!\left[\rv{\tilde n_{>1}} \ge V^{\mathrm{ex}}(\rv{n_{\mc}},\epsilon_{\mathrm{PNE}},\lambda_{\min})\right]\le \epsilon_{\mathrm{PNE}}.
        \label{eq:pne_upper_exact}
    \end{equation}
\end{lemma}
\begin{proof}
    Fix any $\tilde n_{>1}\in\{0,\ldots,N\}$. Define
    \begin{multline}
        k_{\tilde n_{>1}}:=\max\Biggl\{k\in\{-1,\ldots,N\}: \\
        \Pr[\mathrm{Bin}(\tilde n_{>1},\lambda_{\min})\le k]\le \epsilon_{\mathrm{PNE}}\Biggr\},
    \end{multline}
    and let $F_{\tilde n_{>1}}(k):=\Pr[\mathrm{Bin}(\tilde n_{>1},\lambda_{\min})\le k]$. By definition of $V^{\mathrm{ex}}$, for every integer $k$,
    \begin{multline}
        \tilde n_{>1}\ge V^{\mathrm{ex}}(k,\epsilon_{\mathrm{PNE}},\lambda_{\min}) \\
        \Rightarrow F_{\tilde n_{>1}}(k)\le \epsilon_{\mathrm{PNE}}
        \Rightarrow k\le k_{\tilde n_{>1}}.
    \end{multline}
    Hence
    \begin{equation}
        \left\{\tilde n_{>1}\ge V^{\mathrm{ex}}(\rv{n_{\mc}},\epsilon_{\mathrm{PNE}},\lambda_{\min})\right\}
        \subseteq
        \left\{\rv{n_{\mc}}\le k_{\tilde n_{>1}}\right\}.
    \end{equation}
    Conditioning on $\rv{\tilde n_{>1}}=\tilde n_{>1}$ and using Eq.~\eqref{eq:pne_binom_cdf}, we obtain
    \begin{align}
         & \Pr\!\left[\tilde n_{>1}\ge V^{\mathrm{ex}}(\rv{n_{\mc}},\epsilon_{\mathrm{PNE}},\lambda_{\min})\mid \rv{\tilde n_{>1}}=\tilde n_{>1}\right] \notag \\
         & \quad\le \Pr[\rv{n_{\mc}}\le k_{\tilde n_{>1}}\mid \rv{\tilde n_{>1}}=\tilde n_{>1}] \notag                                                         \\
         & \quad\le \Pr[\mathrm{Bin}(\tilde n_{>1},\lambda_{\min})\le k_{\tilde n_{>1}}]
        \le \epsilon_{\mathrm{PNE}}.
    \end{align}
    Finally taking the expectation over $\rv{\tilde n_{>1}}$ yields Eq.~\eqref{eq:pne_upper_exact}.
\end{proof}

\medskip
\noindent
\paragraph{Lower bound on the single-photon key rounds.}
Let $\rv{\tilde n_{1_A,1_B,Z}}$ denote the  random variable counting the number of rounds in which simultaneously Alice emits a single-photon pulse, Bob's virtual QND measurement outputs a single-photon pulse, and the round is kept and assigned to the $Z$-basis key set. A lower bound on this random variable can be obtained by subtracting worst-case vacuum and multi-photon contributions from a lower bound on $\rv{n_{1_A,Z}}$ obtained with the decoy-state method \cite{wang2025phase}. The vacuum contribution is controlled by the vacuum single-click parameter $q_Z$, while the multi-photon contribution is controlled by the observed multi-click count $\rv{n_{\mc}}$ together with the lower bound $\lambda_{\min}$ on the multi-click probability conditioned on Bob receiving more than one photon. Below we replace the Hoeffding relaxations employed in Ref.~\cite{wang2025phase} by the exact bounds provided by Eqs.~\eqref{eq:vacuum_key_upper_exact} and \eqref{eq:pne_upper_exact}. In particular define the exact-statistics by
\begin{multline}
    \rv{n^L_{1_A,1_B,Z}}
    :=
    \max\!\Bigl\{0,\; \rv{n_{1_A,Z}^{L}} - B_0^{\mathrm{ex}}(N,q_Z,\epsilon_{0Z}) \\
    - V^{\mathrm{ex}}(\rv{n_{\mc}},\epsilon_{\mathrm{PNE}},\lambda_{\min}) \Bigr\},
    \label{eq:B1_exact}
\end{multline}
where $n_{1_A,Z}^L$ is the lower bound on $n_{1_A,Z}$.
The two subtracted terms are explicit upper bounds on the undesired contributions (vacuum and multi-photon rounds) that are not certified by the single-photon sampling argument.
This means that, by a union bound over Eqs.~\eqref{eq:vacuum_key_upper_exact},  and \eqref{eq:pne_upper_exact}, and the decoy-state LP failure probability, we have
\begin{equation}
    \Pr\!\left(\rv{\tilde n_{1_A,1_B,Z}} < \rv{n^L_{1_A,1_B,Z}}\right)\le \epsilon_{0Z}+\epsilon_{\mathrm{PNE}}+\epsilon_{\mathrm{decoy}}.
    \label{eq:B1_exact_fail}
\end{equation}

For the sampling argument we also need bounds on the corresponding  test rounds. Let $\rv{\tilde n_{1_A,1_B,X}}$ denote the  random variable counting the number of rounds in which simultaneously Alice emits a single-photon pulse, Bob's virtual QND measurement outputs a single-photon pulse, and the round is kept and assigned to the $X$-basis test set. Define
\begin{multline}
    \rv{n_{1_A,1_B,X}^{L}}
    :=
    \max\!\Bigl\{0,\; \rv{n_{1_A,X}^{L}} - B_0^{\mathrm{ex}}(N,q_X,\epsilon_{0X}) \\
    - V^{\mathrm{ex}}(\rv{n_{\mc}},\epsilon_{\mathrm{PNE}},\lambda_{\min}) \Bigr\}.
    \label{eq:n1X_L}
\end{multline}
Then, by the same subtraction logic discussed above and a union bound over Eqs.~\eqref{eq:vacuum_test_upper_exact} and \eqref{eq:pne_upper_exact}, together with the LP failure probability associated with the estimation of $\rv{n_{1_A,X}^{L}}$, we have
\begin{equation}
    \Pr\!\left(\rv{\tilde n_{1_A,1_B,X}} < \rv{n_{1_A,1_B,X}^{L}}\right)\le \epsilon_{0X}+\epsilon_{\mathrm{PNE}}+\epsilon_{\mathrm{decoy}}.
    \label{eq:n1X_fail}
\end{equation}

\medskip
\noindent

\paragraph{Hypergeometric Clopper-Pearson phase-error rate bound.}
Here we use the exact hypergeometric Clopper-Pearson inversion stated in
Theorem~\ref{Hyper} presented in Appendix~\ref{sec:bounds}.
This sampling step is applied only after conditioning both on Alice emitting a single photon and on Bob's virtual QND outcome being \(1_B\), because on this restricted population the passive beam splitter acts as an effective active basis selector.

We first state the sampling relation in terms of the random variables. Specifically, consider the rounds in which Alice emits a single-photon pulse, Bob's QND measurement outputs a single-photon pulse, and Alice's and Bob's basis match. Let $\rv{\tilde N_{\mathrm{H}}}:=\rv{\tilde n_{1_A,1_B,X}}+\rv{\tilde n_{1_A,1_B,Z}}$ be its size, partitioned into an $X$-arm test subset of size $\rv{\tilde n_{\mathrm{H}}}:=\rv{\tilde n_{1_A,1_B,X}}$ and a $Z$-arm key subset of size $\rv{\tilde n_{1_A,1_B,Z}}$. Let $\rv{\tilde m_{1_A,1_B,X}}$ denote the number of $X$-basis bit errors within the test subset. Since the test subset is a uniform sample without replacement from the population given by $\rv{\tilde n_{1_A,1_B,X}}+\rv{\tilde n_{1_A,1_B,Z}}$, Theorem~\ref{Hyper} yields the following relation for the phase-error rate.
\begin{multline}
    \Pr\!\bigl[
        \rv{\tilde \phi_{1_A,1_B,Z}}
        \ge
        \frac{1}{\rv{\tilde n_{1_A,1_B,Z}}}
        \bigl(
        \rv{\tilde n_{1_A,1_B,X}}+\rv{\tilde n_{1_A,1_B,Z}}
        \bigr) \\
        \times \mathcal{H}^{\mathrm{U}}_{\rv{\tilde N_{\mathrm{H}}},\,\rv{\tilde n_{\mathrm{H}}},\,\epsilon_{\mathrm{S}}}\!\left(
        \frac{\rv{\tilde m_{1_A,1_B,X}}}{\rv{\tilde n_{1_A,1_B,X}}}
        \right)
        - \frac{\rv{\tilde m_{1_A,1_B,X}}}{\rv{\tilde n_{1_A,1_B,Z}}}
        \bigr]
    \le
    \epsilon_{\mathrm{S}}.
    \label{eq:phi_ideal}
\end{multline}

Next, we choose the relevant bounds for the individual terms. The decoy-state LP bounds the total $X$-test single-photon error counts $\rv{m_{1_A,X}}$ (for Alice-emitted \emph{single} photons). On the other hand, the random variable $\rv{\tilde m_{1_A,1_B,X}}$ counts only those errors within the received single-photon subspace. Since the received vacuum and multi-photon contributions can each contribute at most one error per round, a lower bound on $\rv{\tilde m_{1_A,1_B,X}}$ can be obtained by subtracting from $\rv{m_{1_A,X}^{L}}$ the upper bounds on these non-single-photon contributions:
\begin{multline}
    \rv{m_{1_A,1_B,X}^{L}}
    := \\
    \max\!\Bigl\{
    0,\;
    \rv{m_{1_A,X}^{L}}-B_0^{\mathrm{ex}}(N,q_X,\epsilon_{0X}) \\
    -V^{\mathrm{ex}}(\rv{n_{\mc}},\epsilon_{\mathrm{PNE}},\lambda_{\min})
    \Bigr\}.
    \label{eq:m1X_bounds}
\end{multline}
Define
\begin{align}
    \rv{\phi_{1_A,1_B,Z}^{U}}
     & :=
    \frac{1}{\rv{n^L_{1_A,1_B,Z}}}
    \Biggl[
        \left(\rv{n_{1_A,1_B,X}^{U}}+\rv{n_{1_A,1_B,Z}^{U}}\right)
    \notag           \\
     & \qquad \times
        \mathcal{H}^{\mathrm{U}}_{\rv{N_{\mathrm{H}}^{U}},\,\rv{n_{\mathrm{H}}^{L}},\,\epsilon_{\mathrm{S}}}\!\left(
        \frac{\rv{m_{1_A,1_B,X}^{U}}}{\rv{n_{1_A,1_B,X}^{L}}}
        \right)
        -
        \rv{m_{1_A,1_B,X}^{L}}
        \Biggr],
    \label{eq:phi_bound}
\end{align}
where the upper bounds on the different quantities are given by
\begin{align}
    \rv{m_{1_A,1_B,X}^{U}} & := \rv{m_{1_A,X}^{U}}, \\
    \rv{n_{1_A,1_B,X}^{U}} & := \rv{n_{1_A,X}^{U}}, \\
    \rv{n_{1_A,1_B,Z}^{U}} & := \rv{n_{1_A,Z}^{U}},
\end{align}
and
\begin{align}
    \rv{N_{\mathrm{H}}^{U}}
     & :=
    \rv{n_{1_A,1_B,X}^{U}}+\rv{n_{1_A,1_B,Z}^{U}},
     &
    \rv{n_{\mathrm{H}}^{L}}
     & :=
    \rv{n_{1_A,1_B,X}^{L}}.
    \label{eq:hypergeom_shorthand_bounds}
\end{align}
Using the hypergeometric sampling statement provided by Eq.~\eqref{eq:phi_ideal} together with a union bound over the auxiliary events needed to justify these bounds, we obtain
\begin{equation}
    \Pr\!\left(\rv{\tilde \phi_{1_A,1_B,Z}} > \rv{\phi_{1_A,1_B,Z}^{U}}\right)\le \epsilon_{0Z}+\epsilon_{0X}+\epsilon_{\mathrm{PNE}}+\epsilon_{\mathrm{decoy}}+\epsilon_{\mathrm{S}}.
    \label{eq:phi_fail}
\end{equation}
\fussy
\setlength{\emergencystretch}{0pt}

\section{Asymmetric passive BB84 receiver}
\label{sec:detection_appendix}

Here we provide the model that we use to simulate the detection statistics in an asymmetric passive BB84 receiver with splitting ratio $s\in (0,1)$.

In general, in a passive BB84 receiver, Bob employs a static beam splitter that routes a fraction $s\in(0,1)$ of the incident intensity to the $X$ arm and the remaining fraction $1-s$ to the $Z$ arm. A round assigned to basis $\mathcal{B}\in\{Z,X\}$ is accepted only when the corresponding arm produces a valid single click and the complementary arm remains silent. This single-click veto rule is the common criterion used throughout the passive receiver architectures
discussed in this manuscript.

For a sifted $Z$-basis event, the matching arm receives intensity $(1-s)\mu$ and the veto arm (i.e., the arm associated with the $X$-basis) receives intensity $s\mu$. For a sifted $X$-basis event, these roles are reversed: the matching arm receives intensity $s\mu$ and the veto arm receives intensity $(1-s)\mu$. Within the matching arm, the correct detector receives a fraction $(1-e_{\mathrm{mis}})$ of the signal intensity and the wrong detector receives a fraction $e_{\mathrm{mis}}$ due to optical misalignment.

\subsection{Click probabilities}

As in the asymmetric active receiver (see Appendix~\ref{sec:asym_active_BB84_rec}), we denote by $p_{\mathrm{noise}}$ the per-detector per-gate probability of a noise-induced click, namely a click caused by background photons and/or detector dark counts in the absence of signal.

To express this quantity from the basic receiver parameters, let $\bar{n}_{\mathrm{B}}$ denote the total mean number of background photons arriving at Bob's receiver in one gate. In the asymmetric passive scenario, the beam splitter routes a fraction $(1-s)$ of this background light to the $Z$ arm and a fraction $s$ to the $X$ arm, after which the PBS divides the incoming background light equally between its two detectors. Therefore, the mean background-photon number incident on a single detector, which we shall denote by $\bar{n}_{\mathrm{B,det}}^Z$ and $\bar{n}_{\mathrm{B,det}}^X$ depending on the basis, is
\begin{equation}
    \bar{n}_{\mathrm{B,det}}^{Z} = \frac{(1-s)\,\eta_{\mathrm{R}}\eta_{\mathrm{D}}\,\bar{n}_{\mathrm{B}}}{2},
\end{equation}
\begin{equation}
    \bar{n}_{\mathrm{B,det}}^{X} = \frac{s\,\eta_{\mathrm{R}}\eta_{\mathrm{D}}\,\bar{n}_{\mathrm{B}}}{2},
\end{equation}
where $\eta_{\mathrm{R}}$ is the receiver optics efficiency and $\eta_{\mathrm{D}}$ is the detector efficiency.
Hence the corresponding per-detector per-gate noise-click probabilities are
\begin{equation}
    p_{\mathrm{noise}}^{Z} = 1 - \exp\!\left(-\frac{(1-s)\,\eta_{\mathrm{R}}\eta_{\mathrm{D}}\,\bar{n}_{\mathrm{B}}}{2}\right)(1-p_{\mathrm{dark}}),
\end{equation}
\begin{equation}
    p_{\mathrm{noise}}^{X} = 1 - \exp\!\left(-\frac{s\,\eta_{\mathrm{R}}\eta_{\mathrm{D}}\,\bar{n}_{\mathrm{B}}}{2}\right)(1-p_{\mathrm{dark}}),
\end{equation}
where the superscript indicates the basis.

For a valid detection assigned to basis $\mathcal{B}$, the intensities arriving at the matching-arm detectors are $\eta_{\mathrm{sys}}(1-s)\mu(1-e_{\mathrm{mis}})$ and $\eta_{\mathrm{sys}}(1-s)\mu e_{\mathrm{mis}}$ in the $Z$ arm, and $\eta_{\mathrm{sys}}s\mu(1-e_{\mathrm{mis}})$ and $\eta_{\mathrm{sys}}s\mu e_{\mathrm{mis}}$ in the $X$ arm, for the correct and wrong detectors respectively. The corresponding probability that one of the detectors in the matching arm produced at least one click is
\begin{equation}
    D_{\mu,Z} = 1 - (1 - p_{\mathrm{noise}}^{Z})^2 \exp(-\eta_{\mathrm{sys}}(1-s)\mu),
\end{equation}
\begin{equation}
    D_{\mu,X} = 1 - (1 - p_{\mathrm{noise}}^{X})^2 \exp(-\eta_{\mathrm{sys}}s\mu).
\end{equation}

The probability of observing no click (due to signal or dark counts) in the matching arm detectors is given by
\begin{align}
    p_{\mathrm{nc},c}^{Z} & = (1 - p_{\mathrm{noise}}^{Z}) \exp\bigl[-\eta_{\mathrm{sys}}(1-s)\mu(1 - e_{\mathrm{mis}})\bigr], \\
    p_{\mathrm{nc},w}^{Z} & = (1 - p_{\mathrm{noise}}^{Z}) \exp\bigl[-\eta_{\mathrm{sys}}(1-s)\mu e_{\mathrm{mis}}\bigr],      \\
    p_{\mathrm{nc},c}^{X} & = (1 - p_{\mathrm{noise}}^{X}) \exp\bigl[-\eta_{\mathrm{sys}}s\mu(1 - e_{\mathrm{mis}})\bigr],     \\
    p_{\mathrm{nc},w}^{X} & = (1 - p_{\mathrm{noise}}^{X}) \exp\bigl[-\eta_{\mathrm{sys}}s\mu e_{\mathrm{mis}}\bigr],
\end{align}
where the subscripts $c$ and $w$ indicate whether the click happens in the correct or in the wrong detector, respectively.

Accounting for afterpulsing, we model the total probability of no click as
\begin{align}
    P_{\mathrm{nc},c}^{Z} & = p_{\mathrm{nc},c}^{Z} (1 - D_{\mu,Z}  p_{\mathrm{ap}}), \\
    P_{\mathrm{nc},w}^{Z} & = p_{\mathrm{nc},w}^{Z} (1 - D_{\mu,Z}  p_{\mathrm{ap}}), \\
    P_{\mathrm{nc},c}^{X} & = p_{\mathrm{nc},c}^{X} (1 - D_{\mu,X}  p_{\mathrm{ap}}), \\
    P_{\mathrm{nc},w}^{X} & = p_{\mathrm{nc},w}^{X} (1 - D_{\mu,X}  p_{\mathrm{ap}}),
\end{align}
where $p_{\mathrm{ap}}$ denotes the afterpulsing probability.

For the complementary veto arm, the condition for a valid detection requires both of its detectors to remain silent. The corresponding no-click probability is
\begin{equation}
    P_{\mathrm{veto}}^{(Z)} = (1 - p_{\mathrm{noise}}^{X})^2 \exp(-\eta_{\mathrm{sys}}s\mu) (1 - D_{\mu,X}  p_{\mathrm{ap}})^2,
\end{equation}
\begin{equation}
    P_{\mathrm{veto}}^{(X)} = (1 - p_{\mathrm{noise}}^{Z})^2 \exp(-\eta_{\mathrm{sys}}(1-s)\mu) (1 - D_{\mu,Z}  p_{\mathrm{ap}})^2.
\end{equation}

Consequently, in the asymmetric passive BB84 receiver a valid detection requires exactly one click in the matching arm together with silence in the veto arm. The corresponding basis-resolved gains for each basis are therefore
\begin{equation}
    Q_{\mu,Z} = \Bigl[\bigl(1 - P_{\mathrm{nc},c}^{Z}\bigr) P_{\mathrm{nc},w}^{Z} + P_{\mathrm{nc},c}^{Z}\bigl(1 - P_{\mathrm{nc},w}^{Z}\bigr)\Bigr]  P_{\mathrm{veto}}^{(Z)},
\end{equation}
\begin{equation}
    Q_{\mu,X} = \Bigl[\bigl(1 - P_{\mathrm{nc},c}^{X}\bigr) P_{\mathrm{nc},w}^{X} + P_{\mathrm{nc},c}^{X}\bigl(1 - P_{\mathrm{nc},w}^{X}\bigr)\Bigr]  P_{\mathrm{veto}}^{(X)}.
\end{equation}

The multi-click probability needed to evaluate \(\rv{n_{\mc}}\) includes both double clicks within one arm and cross-clicks between the two arms. Let
\begin{equation}
    P_{\emptyset,\mu}
    =
    P_{\mathrm{nc},c}^{Z}P_{\mathrm{nc},w}^{Z}
    P_{\mathrm{nc},c}^{X}P_{\mathrm{nc},w}^{X}
    \label{eq:passive_all_silent_probability}
\end{equation}
be the probability that all four detectors remain silent. Since the no-click event, the two single-click events with probabilities \(Q_{\mu,Z}\) and \(Q_{\mu,X}\), and the multi-click event are disjoint and exhaustive, the total multi-click probability is
\begin{equation}
    Q_{\mu,\mc}
    =
    1-P_{\emptyset,\mu}-Q_{\mu,Z}-Q_{\mu,X}.
    \label{eq:passive_multiclick_probability}
\end{equation}

Similarly, the probability of observing an error can be obtained by requiring that the wrong detector clicks, the correct detector remains silent, and the veto condition is satisfied:
\begin{equation}
    E_{\mu,Z} = \bigl(1 - P_{\mathrm{nc},w}^{Z}\bigr) P_{\mathrm{nc},c}^{Z}  P_{\mathrm{veto}}^{(Z)},
\end{equation}
\begin{equation}
    E_{\mu,X} = \bigl(1 - P_{\mathrm{nc},w}^{X}\bigr) P_{\mathrm{nc},c}^{X}  P_{\mathrm{veto}}^{(X)}.
\end{equation}
In the finite-key security analysis presented in Appendix~\ref{sec:bb84_passive_asym}, the vacuum single-click parameters $q_Z$ and $q_X$ are computed from the detector dark-count probability alone, because background-induced clicks are already accounted for in the arm-dependent single- and multi-click statistics through $p_{\mathrm{noise}}^{Z}$ and $p_{\mathrm{noise}}^{X}$.

\section{Symmetric passive BB84 receiver}

We devote this subsection to the \emph{symmetric} passive BB84 receiver, obtained by setting $s=1/2$. In this scenario, we consider that the secret key is extracted from both bases. The two basis arms of Bob's receiver are statistically identical, so for the simulations we derive the key length for one basis and then invoke symmetry for the other.

In contrast to Appendix~\ref{sec:detection_appendix}, in the symmetric passive receiver, matching-arm double clicks are retained and assigned a random bit value. Furthermore, the photon-number-subspace analysis based on Bob's virtual QND outcome and the associated vacuum/multi-photon subtraction terms are not necessary. Instead, parameter estimation follows essentially the same methodology as in the case of the asymmetric active BB84 receiver described in Appendix~\ref{sec:protocol_appendix} with a modification. In particular, because the secret key is now extracted from both bases, neither basis can be used entirely for parameter estimation: in each basis, a random fraction of the sifted detections is publicly disclosed for parameter estimation, and only the remaining undisclosed fraction is kept for key generation. Consequently, the bit strings used for key generation in both $Z$ and $X$ are shorter than the corresponding sifted strings. The finite-key post-processing is then carried out separately for the two basis-specific bit strings, and privacy amplification is applied separately to each of them.

For a total of $N$ transmitted rounds, the detection counts and error counts for each intensity are computed using these basis-resolved yields and error rates. Let $P_{\mathrm{nc},c}$ and $P_{\mathrm{nc},w}$ denote the correct- and wrong-detector no-click probabilities in the matching arm, and let $P_{\mathrm{veto}}$ denote the probability that both detectors in the complementary arm remain silent, all evaluated at $s=1/2$. Then, the overall gain associated with all matching-arm events with at least one click in one basis is given by
\begin{equation}
    Q_{\mu} = \bigl(1 - P_{\mathrm{nc},c} P_{\mathrm{nc},w}\bigr) P_{\mathrm{veto}},
\end{equation}
while the corresponding overall error probability includes both explicit wrong-detector single clicks and matching-arm double clicks with random bit assignment:
\begin{equation}
    E_{\mu} = \frac{1}{2}\bigl(1 - P_{\mathrm{nc},w}\bigr)\bigl(1 + P_{\mathrm{nc},c}\bigr) P_{\mathrm{veto}}.
\end{equation}
For the simulations we consider the expected symmetric scenario in which $Q_{\mu,Z}=Q_{\mu,X}=Q_{\mu}$ and $E_{\mu,Z}=E_{\mu,X}=E_{\mu}$. As already mentioned, a fraction $r$ of the detected signals in each basis is publicly disclosed for parameter estimation, while the remaining fraction $(1-r)$ is used for key generation. Denoting the key and test portions of the $Z$-basis detections as $n_{Z}^{\mathrm{key}}$ and $n_{Z}^{\mathrm{test}}$, respectively, we have
\begin{align}
    n_{Z}^{\mathrm{key}}  & = n_Z  (1 - r), \\
    n_{Z}^{\mathrm{test}} & = n_Z  r,
\end{align}
with analogous expressions for the $X$ basis. Here $n_Z$ is the size of the sifted key in the $Z$ basis.

For the $Z$-basis key extraction, the secret key is derived from the undisclosed $Z$-basis detections ($n_{Z}^{\mathrm{key}}$), and parameter estimation is performed using the disclosed test portion of the complementary $X$-basis ($n_{X}^{\mathrm{test}}$, $m_{X}^{\mathrm{test}}$). Here, $m_X^{\mathrm{test}}$ is the number of bit errors in $n_X^{\mathrm{test}}$. The analysis proceeds using the same decoy-state LPs as in the asymmetric active receiver (see Appendix~\ref{dsa}). Specifically, the LP in Eq.~\eqref{eq:LP_Z0} is applied to the $Z$-basis key portion to yield $n_{1,Z}^{\mathrm{key},L}$, and $n_{1,Z}^{\mathrm{key},U}$, and to the $X$-basis test portion to yield $n_{1,X}^{\mathrm{test},L}$ and $n_{1,X}^{\mathrm{test},U}$. Likewise, the LP in Eq.~\eqref{eq:LP_E} is applied to the $X$-basis test errors to obtain $m_{1,X}^{\mathrm{test},L}$ and $m_{1,X}^{\mathrm{test},U}$.

An upper bound on the single-photon bit error rate in the $X$-test set is given by $e_{1,X}^{\mathrm{U}} = m_{1,X}^{\mathrm{test},U}/n_{1,X}^{\mathrm{test},L}$. Using the complementarity principle, the phase error rate in the $Z$-key is upper bounded via Eq.~\eqref{phiupperbound}, with \(N_{\mathrm{H}}^{\mathrm{U}}:=n_{1,Z}^{\mathrm{key},U}+n_{1,X}^{\mathrm{test},U}\) and \(n_{\mathrm{H}}^{\mathrm{L}}:=n_{1,X}^{\mathrm{test},L}\):
\begin{equation}
    \label{eq:passive_phi}
    \begin{split}
        \phi^{\mathrm{U}}_{1,Z}=\frac{1}{n_{1,Z}^{\mathrm{key},L}} \bigg[ \left(n_{1,Z}^{\mathrm{key},U}+n_{1,X}^{\mathrm{test},U}\right) \\
            \times \mathcal{H}^{\mathrm{U}}_{N_{\mathrm{H}}^{\mathrm{U}},\,n_{\mathrm{H}}^{\mathrm{L}},\,\epsilon_{\mathrm{S}}}\left(e^{\mathrm{U}}_{1,X}\right)-m_{1,X}^{\mathrm{test},L} \bigg].
    \end{split}
\end{equation}

% The error correction cost requires an estimate of the QBER in the $Z$-key portion. This is obtained from the disclosed $Z$-test set using hypergeometric statistics (Theorem~\ref{Hyper}):
% \begin{equation}
%     e_Z^{\rm U} = \frac{n_Z\,\mathcal{H}^{\rm U}_{n_Z,\, n_{Z}^{\rm test},\, \epsilon_{\rm S}}\!\left(\frac{m_{Z}^{\rm test}}{n_{Z}^{\rm test}}\right) - m_{Z}^{\rm test}}{n_{Z}^{\rm key}}.
% \end{equation}

A lower bound on the secret-key length that can be extracted from the $Z$ basis is then given by
\begin{equation}
    \label{eq:passive_keylength}
    \begin{split}
        \ell_Z =
        \max\Biggl\{
        \bigg\lfloor n_{1,Z}^{\mathrm{key},L}\left[1-h(\phi^{\mathrm{U}}_{1,Z})\right]
        - \lambda_{\mathrm{EC}} \\
        - \log_2\left(\frac{1}{2\epsilon_{\mathrm{PA}}^2 \epsilon_{\mathrm{cor}}}\right)
        \bigg\rfloor,\,0
        \Biggr\}.
    \end{split}
\end{equation}

By symmetry, for the simulations we assume that a key of length $\ell_X$ can be extracted from the $X$ basis using the same procedure with the roles of $Z$ and $X$ interchanged: the key comes from the undisclosed $X$-basis portion, and parameter estimation uses the disclosed $Z$-test data. The total secret-key length is
\begin{equation}
    \ell_{\mathrm{total}} = \ell_Z + \ell_X.
\end{equation}

\section{Comparing receiver architectures}

Figure~\ref{fig:receiver_mode_comparison} compares the secret-key rate of the three receiver architectures as a function of the overall system loss after optimizing the parameters of each protocol at every loss point. The optimized parameters include the basis and intensity probabilities and the signal and decoy intensities for all architectures, with additional disclosure-rate and beam-splitting-ratio optimization for the symmetric and asymmetric passive receivers, respectively. As expected, the asymmetric active receiver keeps the highest rate and tolerates the highest loss in all three subfigures. Within the passive variants, the asymmetric design remains above the symmetric one across the plotted range for higher loss values. In the low-loss regime, the penalty induced by the modified analysis of the asymmetric case results in a better performance of the symmetric architecture.

\begin{figure}[htbp]
    \includegraphics[width=\linewidth]{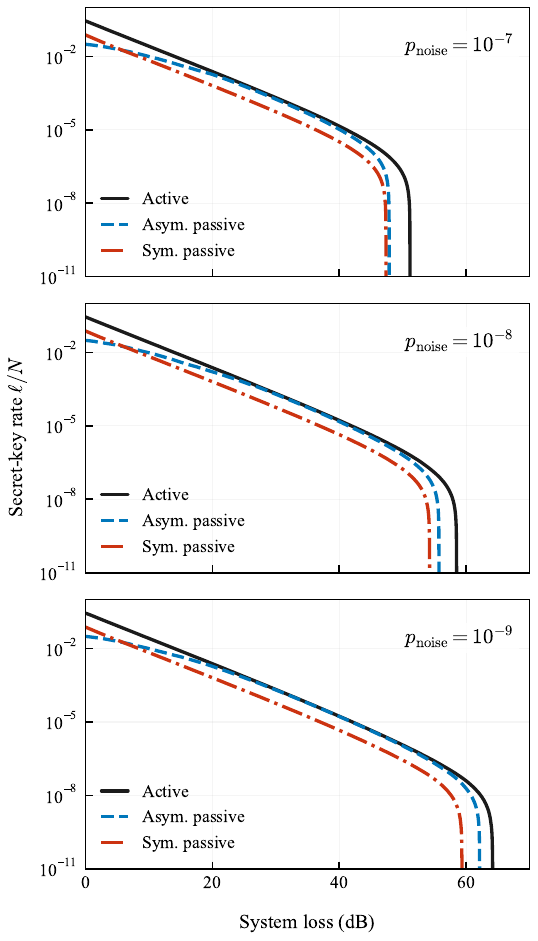}
    \caption{Finite-size secret-key rate $\ell/N$ (y axis) versus overall system loss (x axis) for asymmetric active, asymmetric passive, and symmetric passive BB84 receiver architectures. The three subfigures correspond, respectively, to background-noise probabilities $p_{\mathrm{noise}}=10^{-7}$, $10^{-8}$, and $10^{-9}$ per detector per gate. In all cases the optimization uses $N=10^{12}$, $\epsilon_{\mathrm{tot}}=10^{-8}$, misalignment $e_{\mathrm{mis}}=0.5\%$, and no afterpulsing.}
    \label{fig:receiver_mode_comparison}
\end{figure}

%% file: journal/sections/appendix_statistical_bounds.tex
\section{Statistical bounds}
\label{sec:bounds}

\begin{theorem}
    \label{PoissonBinomial}
    \cite{bancal2022simple}
    Let $\hat{p}$ be the average of $n$ independent Bernoulli variables, with expected value $\mathbb{E}\left[\hat{p}\right]=p$. Then, for all $\epsilon\in(0,1/4]$, we have that $\Pr[p<\mathcal{F}^{\mathrm{L}}_{n,\epsilon}(\hat{p})]\leq{\epsilon}$ for
    \begin{equation}
        \label{eq:poisslb}
        \mathcal{F}^{\mathrm{L}}_{n,\epsilon}(\hat{p})=\left\{
        \begin{array}{ll}
            0                                                                   & \mathrm{if}\hspace{.2cm}n\hat{p}=0,             \\
            \hat{p}-\frac{1-\epsilon}{n\left[1-\epsilon^{*}(n\hat{p},n)\right]} & \mathrm{if}\hspace{.1cm}n\hat{p}>0, \text{ and} \\
                                                                                & \epsilon^{*}\leq{}\epsilon\leq{}1,              \\
            I^{-1}_{\epsilon}(n\hat{p},n')                                      & \mathrm{if}\hspace{.1cm}n\hat{p}>0, \text{ and} \\
                                                                                & 0\leq{}\epsilon\leq{}\epsilon^{*},
        \end{array}
        \right.
    \end{equation}
    where $\epsilon^{*}=\epsilon^{*}(n\hat{p},n)=I_{(n\hat{p}-1)/n}(n\hat{p},n-n\hat{p}+1)$ with $n'=n-n\hat{p}+1$, and  where $I^{-1}_{\epsilon}(a,b)$ denotes the inverse regularized incomplete beta function, such that $I_{I^{-1}_{\epsilon}(a,b)}(a,b)=\epsilon$ for
    \begin{equation}
        \label{eq:betafunc}
        I_{x}(a,b)=\frac{\displaystyle{\int_{0}^{x}t^{a-1}(1-t)^{b-1}dt}}{\displaystyle{\int_{0}^{1}t^{a-1}(1-t)^{b-1}dt}}.
    \end{equation}
    Complementarily, for all $\epsilon\in(0,1/4]$, we have
    \begin{equation}
        \Pr[p>\mathcal{F}^{\mathrm{U}}_{n,\epsilon}(\hat{p})]\leq{\epsilon},
    \end{equation}
    where
    \begin{equation}
        \label{eq:poissub}
        \mathcal{F}^{\mathrm{U}}_{n,\epsilon}(\hat{p})
        =1-\mathcal{F}^{\mathrm{L}}_{n,\epsilon}(1-\hat{p}).
    \end{equation}
\end{theorem}

\begin{theorem}
    \label{Binomial}
    Let $\hat{p}$ be the average of $n$ independent and identically distributed Bernoulli variables, with expected value $p$. Then for all $\epsilon>0$,
    \begin{equation}
        \Pr[\hat{p}\geq \mathcal{G}^{\mathrm{U}}_{n,\epsilon}(p)]\leq\epsilon,
    \end{equation}
    where,
    \begin{multline}
        \label{eq:binub}
        \mathcal{G}^{\mathrm{U}}_{n,\epsilon}(p)=\frac{1}{n}\min\Bigl\{k\in\{-1, \dots, n+1\} : \\
        I_p(k, n-k+1)\leq\epsilon\Bigr\}.
    \end{multline}
    Similarly,
    \begin{equation}
        \Pr[\hat{p}\leq \mathcal{G}^{\mathrm{L}}_{n,\epsilon}(p)]\leq\epsilon,
    \end{equation}
    where
    \begin{multline}
        \label{eq:binlb}
        \mathcal{G}^{\mathrm{L}}_{n,\epsilon}(p)=\frac{1}{n}\max\Bigl\{k\in\{-1, \dots, n+1\} : \\
        I_{1-p}(n-k, k+1)\leq\epsilon\Bigr\},
    \end{multline}
    and  $I_x(a,b)$ is given by Eq.~\eqref{eq:betafunc}.
\end{theorem}

\begin{proof}
    For the proof, it suffices to observe that
    $n\mathcal{G}^{\mathrm{U}}_{n,\epsilon}(p)$ is the smallest integer threshold
    $k$ for which the binomial upper tail
    $\Pr[X\ge k]$, with $X\sim\mathrm{Bin}(n,p)$, is at most $\epsilon$.
    Let
    \begin{equation}
        \begin{split}
            f_{n, \epsilon} (p) = \min \bigg\{ k & \in \{-1, \dots, n+1\} \;\bigg|                                  \\
                                                 & \sum_{j=k}^n \binom{n}{j} p^j (1-p)^{n-j} \leq \epsilon \bigg\}.
        \end{split}
    \end{equation}

    By definition of $f_{n, \epsilon}(p)$, we have that
    $$ \Pr \left[ \hat{X} \geq f_{n, \epsilon} (p) \right] \leq \epsilon \quad \text{for } \hat{X} = n\hat{p}. $$

    Equivalently,
    $$ \Pr \left[ \hat{p} \geq f_{n, \epsilon} (p)/n \right] \leq \epsilon, $$
    and the claim follows identifying that
    $$ \sum_{j=k}^n \binom{n}{j} p^j (1-p)^{n-j} = I_p(k, n-k+1). $$
    The lower bound follows similarly by defining $g_{n, \epsilon} (p) = \max \{ k \in \{-1, \dots, n+1\} \mid \sum_{j=0}^k \binom{n}{j} p^j (1-p)^{n-j} \leq \epsilon \}$ and identifying the sum as $I_{1-p}(n-k, k+1)$.
\end{proof}

\begin{theorem}
    \label{Hyper}
    \cite{mannalath2025sharp}
    A test sample of size $n$ is drawn at random from a binary population with $N$ elements. Let $\hat{p}$ ($\hat{q}$) denote the frequency of ones in the test (complementary) sample. Then for all $\epsilon>0$,
    \begin{equation}
        \Pr\left[\hat{q}\geq q^{\mathrm{U}}_{N,n,\epsilon}(\hat{p})\right]\leq{}\epsilon
    \end{equation}
    for
    \begin{equation}\label{q_th}
        q^{\mathrm{U}}_{N,n,\epsilon}(x)=\frac{N\mathcal{H}^{\mathrm{U}}_{N,n,\epsilon}(x)-nx}{N-n},
    \end{equation}
    where $\mathcal{H}^{\mathrm{U}}_{N,n,\epsilon}(x)$ is defined as,
    \begin{equation}
        \label{eq:hyperub}
        \begin{split}
            \mathcal{H}^{\mathrm{U}}_{N,n,\epsilon}(x)=\min\{1\geq y\geq x\big\rvert{} \hat{Y}\sim \mathrm{Hyper}(N,Ny,n), \\
            \Pr[\hat{Y}\leq n x]\leq\epsilon\},
        \end{split}
    \end{equation}
    if the target set of the minimization is non-empty, and via $\mathcal{H}^{\mathrm{U}}_{N,n,\epsilon}(x)=(N+1)/N$ otherwise.
\end{theorem}

%% file: journal/sections/appendix_channel_model.tex
\section{Channel model}
\label{sec:channel_appendix}

In this appendix, we describe the satellite-to-ground optical channel model that determines the transmission efficiency $\eta_{\mathrm{sys}}$ and background noise parameters entering the security analysis.

For this, we consider the most relevant physical phenomena that impact signal transmission. This includes beam diffraction due to finite aperture sizes, pointing errors, turbulence effects, atmospheric absorption and scattering,  coupling losses between the receiver and detection setup and background noise at the receiver.

\begin{figure}
  \centering
  \includegraphics[width=0.8\linewidth]{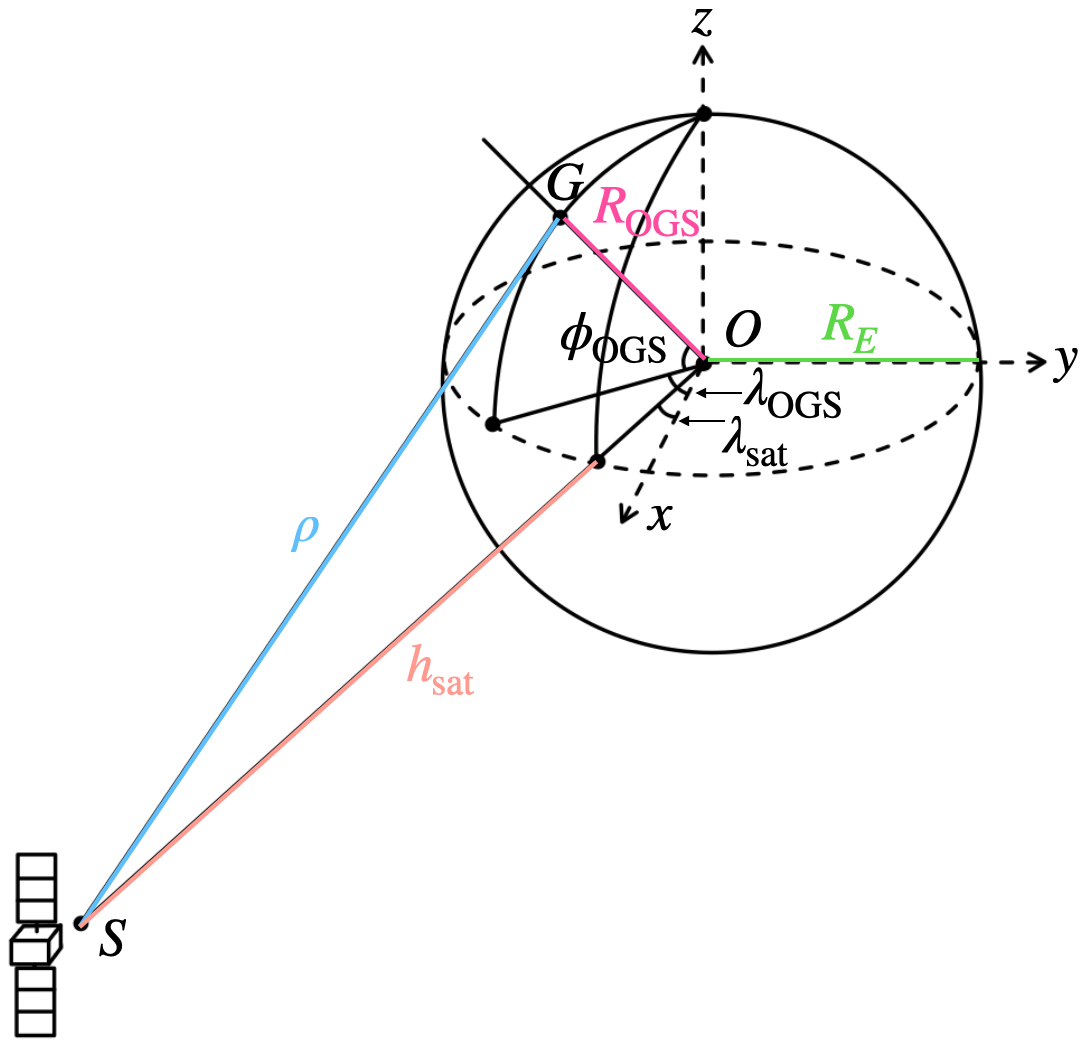}
  \caption{Geometry of the satellite-to-ground downlink scenario.
    The Earth is modeled as a sphere of radius $R_E$ centered at the origin $O$.
    The OGS is located at point $G$ with latitude $\phi_{\mathrm{OGS}}$
    and longitude $\lambda_{\mathrm{OGS}}$, at a radial distance $R_{\mathrm{OGS}}$ from the center ($h_{\mathrm{OGS}}$ is negligible compared to $R_E$, hence $R_{\mathrm{OGS}}\approx{}R_E$).
    The geostationary satellite $S$ is positioned at altitude $h_{\mathrm{sat}}$ and
    longitude $\lambda_{\mathrm{sat}}$. The slant range $\rho$ is the line-of-sight length from the
    satellite to the OGS.}
  \label{fig:earthsection}
\end{figure}

\begin{figure}
  \centering
  \includegraphics[width=0.8\linewidth]{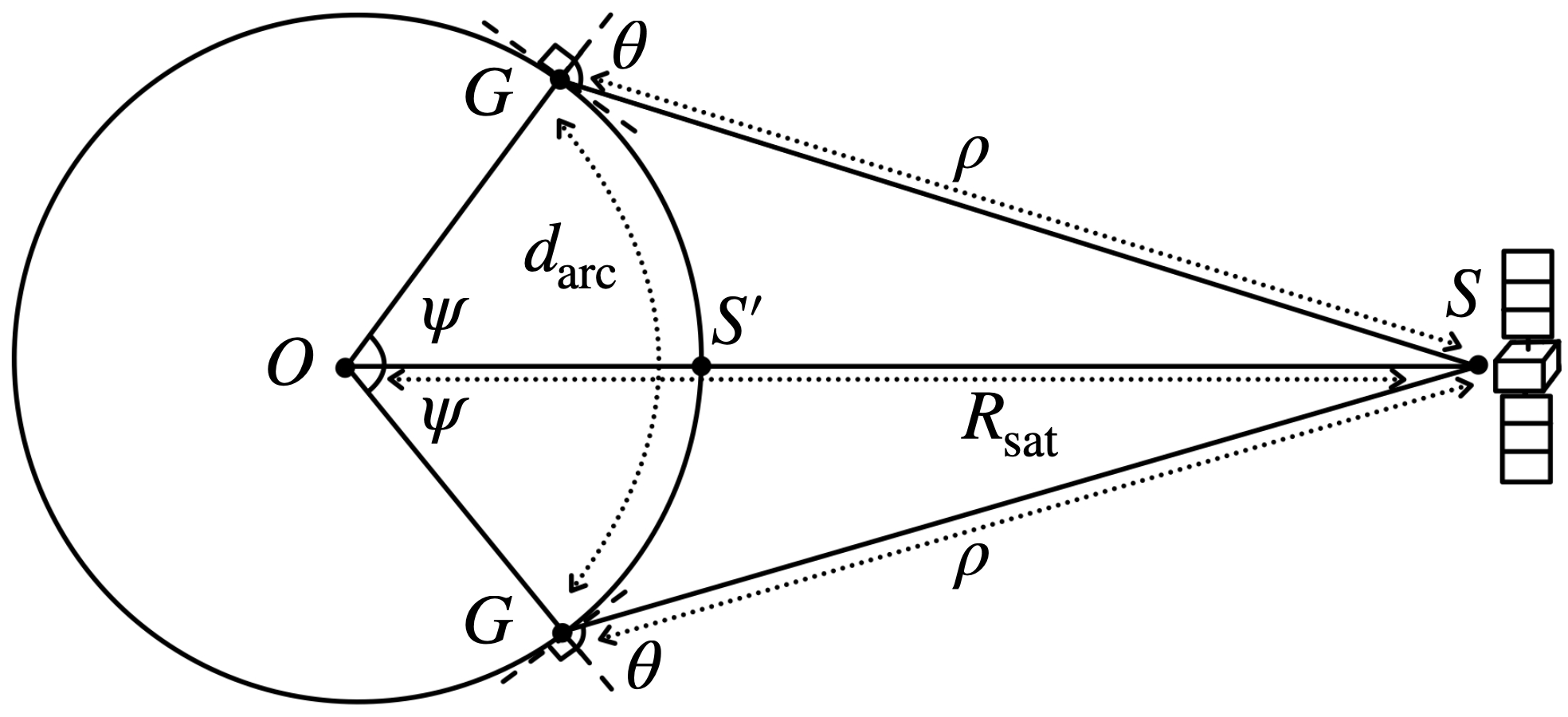}
  \caption{Cross-sectional geometry used to derive the maximum ground separation distance $d_{\mathrm{arc}}$ between two OGSs positioned symmetrically on the Earth's surface relative to the sub-satellite point $S'$, such that they both observe the geostationary satellite $S$ at a fixed zenith angle $\theta$. The central angle subtended by each station from the Earth's center $O$ is denoted by $\psi$.}
  \label{fig:arcgeom}
\end{figure}

\label{sec:channel_appendix_geometry}
We begin by modeling the geometry of the satellite-to-ground link, specifically calculating the zenith and azimuth viewing angles and the slant range \cite{soler1994determination,ilcev2016satellite}. We assume a perfectly spherical Earth of radius $R_E$, and an optical ground station (OGS) located at latitude $\phi_{\mathrm{OGS}}$, longitude $\lambda_{\mathrm{OGS}}$, and altitude $h_{\mathrm{OGS}}$. The geostationary satellite is orbiting above the equator at a longitude $\lambda_{\mathrm{sat}}$, and altitude $h_{\mathrm{sat}}$ (see Figure~\ref{fig:earthsection}). We define the parameters

\begin{equation}
  R_{\mathrm{OGS}} = R_E + h_{\mathrm{OGS}},
  \quad
  R_{\mathrm{sat}} = R_E + h_{\mathrm{sat}}.
\end{equation}
The unit vectors toward the OGS and the satellite in the Earth-centered frame are
\begin{equation}
  \hat{\bm r}_{\mathrm{OGS}} =
  \begin{bmatrix}
    \cos\phi_{\mathrm{OGS}}\cos\lambda_{\mathrm{OGS}} \\[3pt]
    \cos\phi_{\mathrm{OGS}}\sin\lambda_{\mathrm{OGS}} \\[3pt]
    \sin\phi_{\mathrm{OGS}}
  \end{bmatrix},
  \quad
  \hat{\bm r}_{\mathrm{sat}} =
  \begin{bmatrix}
    \cos\lambda_{\mathrm{sat}} \\[3pt]
    \sin\lambda_{\mathrm{sat}} \\[3pt]
    0
  \end{bmatrix},
\end{equation}
with position vectors,
\begin{equation}
  \bm r_{\mathrm{OGS}} = R_{\mathrm{OGS}}\hat{\bm r}_{\mathrm{OGS}},
  \qquad
  \bm r_{\mathrm{sat}} = R_{\mathrm{sat}}\hat{\bm r}_{\mathrm{sat}}.
\end{equation}

The central angle $\psi$ between the OGS and the sub-satellite point (i.e., the locations on the Earth's surface directly beneath the satellite) is given by (see Figure~\ref{fig:arcgeom})
\begin{equation}
  \begin{split}
    \cos\psi & = \hat{\bm r}_{\mathrm{OGS}}\!\cdot\!\hat{\bm r}_{\mathrm{sat}}              \\
             & = \cos\phi_{\mathrm{OGS}}\cos(\lambda_{\mathrm{sat}}-\lambda_{\mathrm{OGS}}) \\
             & = \cos\phi_{\mathrm{OGS}}\cos\Delta\lambda,
  \end{split}
\end{equation}
where we use the identity
$\cos(a-b)=\cos(a)\cos(b)+\sin(a)\sin(b)$.
Here, $\Delta\lambda = \lambda_{\mathrm{sat}} - \lambda_{\mathrm{OGS}}$.

The line-of-sight vector from the OGS to the satellite is (see Figure~\ref{fig:earthsection})
\begin{equation}
  \label{eq:vecLOS}
  \bm\rho = \bm r_{\mathrm{sat}} - \bm r_{\mathrm{OGS}}.
\end{equation}

\subsection{Zenith angle}

The vertical component magnitude of $\bm\rho$  is
\begin{equation}
  \label{eq:defu}
  u = \bm\rho\cdot\hat{\bm r}_{\mathrm{OGS}} = R_{\mathrm{sat}}\cos\psi - R_{\mathrm{OGS}},
\end{equation}
and the horizontal component magnitude is
\begin{equation}
  \label{eq:defh}
  h = \sqrt{\rho^2 - u^2} = R_{\mathrm{sat}}\sin\psi.
\end{equation}
The zenith angle is then given by
\begin{equation}
  \tan \theta = \frac{h}{u}
  = \frac{R_{\mathrm{sat}}\sin\psi}{R_{\mathrm{sat}}\cos\psi - R_{\mathrm{OGS}}}.
\end{equation}

\subsection{Slant range }

As already mentioned, the slant range is the length $\rho$ of the line-of-sight vector from the satellite to the OGS. The zenith angle $\theta$ at the OGS satisfies
\begin{equation}
  \rho\cos\theta = u, \qquad \rho\sin\theta = h.
\end{equation}
Using the relations for $u$ and $h$ given by Eqs.~(\ref{eq:defu})-(\ref{eq:defh}), we have that
\begin{equation}
  \begin{split}
    R_{\mathrm{sat}}^2 & = (R_{\mathrm{sat}}\cos\psi)^2 + (R_{\mathrm{sat}}\sin\psi)^2     \\
                       & = (\rho\cos\theta+R_{\mathrm{OGS}})^2 + (\rho\sin\theta)^2        \\
                       & = \rho^2 + 2\rho R_{\mathrm{OGS}}\cos\theta + R_{\mathrm{OGS}}^2.
  \end{split}
\end{equation}
This quadratic expression in $\rho$ yields
\begin{equation}
  \rho^2+2\rho R_{\mathrm{OGS}}\cos\theta+\big(R_{\mathrm{OGS}}^2-R_{\mathrm{sat}}^2\big)=0,
\end{equation}
whose relevant root is
\begin{equation}
  \label{eq:slant_range}
  \rho
  = \sqrt{R_{\mathrm{sat}}^2 + R_{\mathrm{OGS}}^2(\cos^2\theta - 1)} - R_{\mathrm{OGS}}\cos\theta.
\end{equation}
This parameter gives the geometric propagation distance between the OGS and the satellite, neglecting atmospheric refraction and Earth's oblateness~\cite{fowler1990solid}.
In practice, the true optical path deviates slightly from this quantity due to refractive bending. For example, one can integrate the ray curvature through a stratified atmospheric model to correct the geometric path length~\cite{vasylyev2019satellite}.
For a GEO satellite at a nominal altitude of $35{,}786$~km, the additional path length due to atmospheric refraction is typically less than $10$~km, corresponding to $<0.03\,\%$ of the total propagation distance. In addition, Earth's oblateness and orbital uncertainties introduce further small deviations. Specifically, the Earth's equatorial radius exceeds the polar radius by approximately $21$~km, corresponding to an oblateness of $\sim 0.3\,\%$, which modifies the ground-to-satellite range by less than $\pm15$~km depending on the OGS latitude~\cite{hofmann2005physical}.
Typical orbit-control limits for a well-maintained GEO spacecraft are $\pm0.05^{\circ}$ in longitude and inclination, producing slant range variations of $\pm30$--$40$~km~\cite{wertz2011space}.
Together, these non-refractive uncertainties contribute less than $0.2\,\%$ of the total optical path. Given their negligible magnitude relative to the total path length, these effects are omitted in the present analysis for simplicity.

\subsection{Azimuth angle}
Let us now derive the azimuth angle of the satellite from the OGS. To keep the discussion compact, we use the standard local east-north-up (ENU) basis at the OGS~\cite{soler1994determination,navipedia_ecef_enu,ilcev2016satellite}. In the Earth-centered frame, the local east and north unit vectors at latitude $\phi_{\mathrm{OGS}}$ and longitude $\lambda_{\mathrm{OGS}}$ are
\begin{align}
  \hat{\bm e}=
  \begin{bmatrix}
    -\sin\lambda_{\mathrm{OGS}}           \\[2pt]
    \phantom{-}\cos\lambda_{\mathrm{OGS}} \\[2pt]
    0
  \end{bmatrix},
  \qquad
  \hat{\bm n}=
  \begin{bmatrix}
    -\sin\phi_{\mathrm{OGS}}\cos\lambda_{\mathrm{OGS}} \\[2pt]
    -\sin\phi_{\mathrm{OGS}}\sin\lambda_{\mathrm{OGS}} \\[2pt]
    \phantom{-}\cos\phi_{\mathrm{OGS}}
  \end{bmatrix},
\end{align}
which fulfill
% As already discussed, the line-of-sight vector is given by Eq.~\eqref{eq:vecLOS}.
% Since $\hat{\bm e}$ and $\hat{\bm n}$ are tangent to the Earth at the OGS, they are orthogonal to $\bm r_{\mathrm{OGS}}$, so that
\begin{equation}
  \bm r_{\mathrm{OGS}}\cdot\hat{\bm e}=0,
  \qquad
  \bm r_{\mathrm{OGS}}\cdot\hat{\bm n}=0.
\end{equation}
Figure~\ref{fig:azimuth_proof} shows the decomposition of the line-of-sight vector $\rho$ onto the local ENU basis.

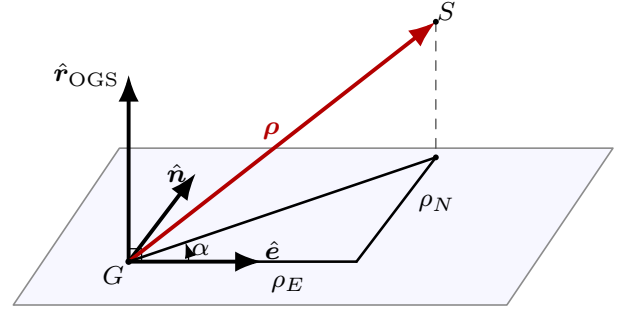
\begin{figure}[t]
  \centering
  \resizebox{0.92\columnwidth}{!}{%
    \begin{tikzpicture}[
        >=Latex,
        line cap=round,
        line join=round,
        font=\footnotesize,
        basis/.style={->, very thick},
        los/.style={->, very thick, red!70!black},
        patch/.style={draw=black!45, line width=0.5pt},
        proj/.style={draw=black!70, dashed},
        every node/.style={inner sep=1.2pt}
      ]
      \coordinate (P1) at (0.20,0.45);
      \coordinate (P2) at (5.55,0.45);
      \coordinate (P3) at (6.70,2.15);
      \coordinate (P4) at (1.35,2.15);
      \fill[blue!3] (P1) -- (P2) -- (P3) -- (P4) -- cycle;
      \draw[patch] (P1) -- (P2) -- (P3) -- (P4) -- cycle;

      \coordinate (G)  at (1.45,0.92);
      \coordinate (GE) at (2.90,0.92);
      \coordinate (GN) at (2.18,1.88);
      \coordinate (P)  at (3.92,0.92);
      \coordinate (H)  at (4.78,2.05);
      \coordinate (U)  at (1.45,2.95);
      \coordinate (S)  at (4.78,3.52);

      \draw[basis] (G) -- (GE) node[above right=-1pt] {$\hat{\bm e}$};
      \draw[basis] (G) -- (GN) node[left=1pt] {$\hat{\bm n}$};
      \draw[basis] (G) -- (U) node[left=1pt] {$\hat{\bm r}_{\mathrm{OGS}}$};
      \draw ($(G)+(0.14,0)$) -- ($(G)+(0.14,0.14)$) -- ($(G)+(0,0.14)$);

      \draw[los] (G) -- (S) node[midway, above left=-1pt] {$\bm\rho$};
      \draw[proj] (S) -- (H);
      \draw[thick] (G) -- (H);
      \draw[thick] (G) -- (P);
      \draw[thick] (P) -- (H);

      \node[below=3pt] at ($(G)!0.70!(P)$) {$\rho_E$};
      \node[right=4pt] at ($(P)!0.56!(H)$) {$\rho_N$};

      \fill (G) circle (1pt) node[below left] {$G$};
      \fill (H) circle (0.9pt);
      \fill (S) circle (0.9pt) node[above right=-1pt] {$S$};

      \pic[
        draw,
        ->,
        "$\alpha$",
        angle eccentricity=1.22,
        angle radius=6.5mm
      ] {angle = GE--G--H};
    \end{tikzpicture}%
  }
  \caption{Local ENU geometry used in the azimuth derivation. The line-of-sight vector $\bm\rho$ is decomposed on the local horizontal basis $(\hat{\bm e},\hat{\bm n})$, yielding the east and north components $\rho_E$ and $\rho_N$. The dashed segment indicates the projection of the satellite position onto the local horizontal plane. The azimuth angle $\alpha$ is measured counterclockwise from east through north.}
  \label{fig:azimuth_proof}
\end{figure}

To obtain the azimuth angle  $\alpha$ (counterclockwise from east through north), we project $\bm\rho$ onto the local horizontal basis
$(\hat{\bm e},\hat{\bm n})$:
\begin{align}
  \rho_E & = \bm\rho\cdot\hat{\bm e}
  = \bm r_{\mathrm{sat}}\cdot\hat{\bm e}
  = R_{\mathrm{sat}}\sin\Delta\lambda,
  \\[4pt]
  \rho_N & = \bm\rho\cdot\hat{\bm n}
  = \bm r_{\mathrm{sat}}\cdot\hat{\bm n}
  = -R_{\mathrm{sat}}\sin\phi_{\mathrm{OGS}}\cos\Delta\lambda.
\end{align}
The azimuth angle $\alpha$ is then given by
\begin{equation}
  \begin{split}
    \tan \alpha & =\frac{\rho_N}{\rho_E}
    =-\frac{\sin\phi_{\mathrm{OGS}}\cos\Delta\lambda}{\sin\Delta\lambda}.
  \end{split}
\end{equation}

\subsection{Ground separation distance}
\label{sec:ground_separation}
The long-distance benchmark in Fig.~\ref{fig:plob_bounds_wavelength_grid} requires converting a GEO viewing zenith angle into the maximum ground separation that can be served by two OGSs using the same satellite.
The locus of all ground points on a spherical Earth observing the GEO satellite at a fixed zenith angle $\theta$ forms a circle on the Earth's surface centered at the sub-satellite point. The maximum separation between two such points occurs when they are diametrically opposite to  each other relative to the sub-satellite point.
Due to the rotational symmetry of the spherical model, this distance is invariant with respect to the azimuth. For simplicity, we derive its value by considering two OGSs located along the same meridian as the sub-satellite point (see Figure~\ref{fig:arcgeom}). The resulting surface distance, denoted by \(d_{\mathrm{arc}}\) below, is the end-to-end ground separation served by the GEO link and
provides the horizontal-axis distance used in the main-text long-distance benchmark figure comparing the GEO
results with the PLOB~\cite{pirandola2017fundamental} and repeater-assisted fiber bounds~\cite{pirandola2019end}.

Let the sub-satellite point lie at latitude \( \phi=0 \) and longitude \( \lambda_{\mathrm{sat}} \). Two OGSs are positioned symmetrically at latitudes \( \pm\phi_{\mathrm{OGS}} \), both at longitude \( \lambda_{\mathrm{sat}} \) and altitude $h_{\mathrm{OGS}}$.
Applying the law of sines in the triangle given by the points $O, G,$ and $S$ (see Fig.~\ref{fig:arcgeom}), we have that
\begin{equation}
  \frac{\sin\psi}{\rho} = \frac{\sin\theta}{R_{\mathrm{sat}}},
\end{equation}
where $\psi$  is again the central angle, which is given by
\begin{equation}
  \label{eq:central_angle}
  \psi = \arcsin\!\left(\frac{\rho\sin\theta}{R_{\mathrm{sat}}}\right).
\end{equation}

This means that the corresponding arc distance from the sub-satellite point to one OGS at altitude $h_{\mathrm{OGS}}$ is given by
\begin{equation}
  \label{eq:one_ground_distance}
  R_{\mathrm{OGS}}\,\psi
  = R_{\mathrm{OGS}}\,\arcsin\!\left(\frac{\rho\sin\theta}{R_{\mathrm{sat}}}\right).
\end{equation}

Because the two stations are symmetrically located around the sub-satellite point, their total separation is twice this value:
\begin{equation}
  \label{eq:ground_separation}
  d_{\mathrm{arc}}
  = 2R_{\mathrm{OGS}}\,\arcsin\!\left(\frac{\rho\sin\theta}{R_{\mathrm{sat}}}\right).
\end{equation}

\subsection{Diffraction}

The transmitted optical beam from the satellite is modeled as a collimated Gaussian field with waist radius $\omega_0$ at the aperture $a_T$ of the transmitting telescope.
At a propagation distance $\rho$, the spot radius expands due to diffraction according to \cite{siegman1986lasers}
\begin{equation}
  \label{diffwidth}
  \omega_{\mathrm{d}} = \omega_0
  \sqrt{1 + \left(\frac{M^2 \rho}{z_R}\right)^2},
\end{equation}
Here, the spot radius $\omega_{\mathrm{d}}$ denotes the Gaussian beam radius in the receiver plane, defined as the transverse distance from the beam axis at which the irradiance has fallen to $1/e^2$ of its on-axis value.
The parameter $M^2 \ge 1$ is the beam quality factor that quantifies any deviation from an ideal TEM$_{00}$ Gaussian beam \cite{siegman1993defining}, and $z_R = \pi \omega_0^2 n / \lambda$ is the Rayleigh range in a medium of refractive index $n$. Throughout this work we will assume a refractive index $n=1$, since nearly the entire propagation path is in vacuum and the near-ground atmospheric refractive index at optical wavelengths deviates from unity only at the $\sim 10^{-4}$ level.
Representative commercial laser sources and their specified beam-quality factors are summarized in Table~\ref{tab:M2}.

\begin{table}[h]
  \centering
  \caption{Representative commercial laser sources with specified beam-quality factor $M^2$.}
  \begin{tabularx}{\columnwidth}{>{\raggedright\arraybackslash}Xcc}
    \toprule
    \textbf{Source} & \textbf{Wavelength [nm]} & \textbf{Specified $M^2$} \\
    \midrule
    CrystaLaser \cite{crystalaser_red_dpss}
                    & $656.5 \pm 3$            & $< 1.1$                  \\
    SolsTiS \cite{solstis_tisapph}
                    & $700$--$1000$            & $< 1.1$                  \\
    Menhir Photonics \cite{menhir_1550_series}
                    & $1545$--$1565$           & $< 1.05$                 \\
    \bottomrule
  \end{tabularx}
  \label{tab:M2}
\end{table}

Reducing diffraction-induced broadening can in principle be achieved by
(i) increasing the launch waist $\omega_0$,
(ii) decreasing the wavelength $\lambda$,
or (iii) minimizing $M^2$.
However, $\omega_0$ cannot be increased arbitrarily because the
physical telescope aperture truncates the Gaussian profile.
This clipping modifies the far-field divergence from ideal Gaussian
behavior.

This truncation can be  characterized by the ratio $a_T / \omega_0$.
If it is  too small, a significant portion of the Gaussian tail is clipped, resulting in power loss, diffraction sidelobes, and an expanded far-field pattern; if it is too large, the beam becomes unnecessarily narrow relative to the aperture, increasing divergence and reducing on-axis intensity at the receiver \cite{pirandola2021limits}.
An intermediate value therefore provides the best balance between transmitted power and beam quality \cite{stutzman2012antenna,andrews2005laser}. In this work, we choose
\begin{equation}
  \omega_0 = \frac{a_T}{4}.
  \label{eq:apertrun}
\end{equation}
This value preserves nearly all of the Gaussian power while avoiding the divergence penalty associated with overfilling the aperture, ensuring an effectively Gaussian far-field distribution and high coupling efficiency at the receiver.
The residual diffraction-induced broadening is then fully captured by Eq.~\eqref{diffwidth} through the beam quality factor $M^2$, which, as shown in Table~\ref{tab:M2}, is typically close to unity for laser sources. We choose  a conservative value of $M^2=1.2$ for our simulations.

\subsection{Turbulence}
\label{sec:turb}

\begin{figure}
  \centering
  \includegraphics[width=\linewidth]{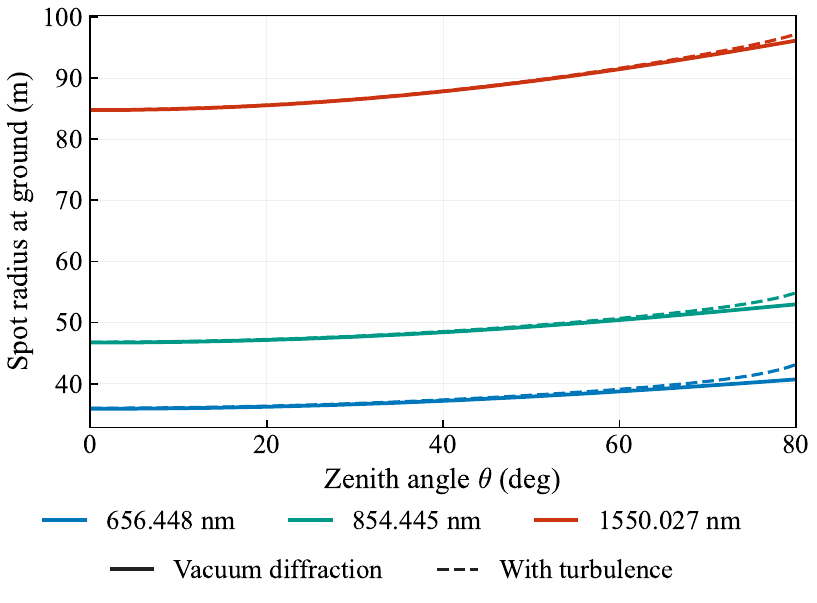}
  \caption{
    Spot radius at the ground (y axis) versus zenith angle (x axis) for three optical communication wavelengths ($\lambda = 656.448$~nm, $854.445$~nm, and $1550.027$~nm), assuming a coastal atmospheric model and a transmitter aperture $a_T = 1$~m.
    Solid lines represent the beam spreading in the absence of turbulence, while dashed lines represent the long-term spot broadening due to atmospheric turbulence. Both lines essentially overlap for most values of the zenith angle.}
  \label{fig:spot_width_analysis}
\end{figure}

This refers to irregular variations in the atmospheric refractive index, driven by temperature gradients, wind shear, pressure and humidity fluctuations, which perturb the propagation of optical waves.
These effects increase the apparent beam radius at the receiver beyond its diffraction-limited value \cite{andrews2005laser}. In a satellite-to-ground downlink scenario, the optical beam travels through vacuum before encountering the atmosphere only in the final few kilometers of the path. Consequently, atmospheric turbulence does not have sufficient propagation distance to significantly expand the physical beam footprint on the ground.
Therefore, beam broadening due to turbulence is expected to be very limited, and the effective spot size is mainly dominated by the diffraction limit of the satellite aperture. In any case, for completeness, we have included its effect in our simulations.

Specifically, the long-term average beam radius at the receiver, denoted by  $\omega_t$, can be expressed as
\begin{equation}
  \omega_t = \omega_d \sqrt{1 + T},
  \label{eq:beam_turb}
\end{equation}
where $\omega_d$ is the diffraction-limited beam radius in vacuum given by Eq.~\eqref{diffwidth}, and $T$ is a dimensionless turbulence broadening parameter that accounts for cumulative phase perturbations along the slant path of the optical signal~\cite{andrews2005laser}.
For a downlink geometry, $T$ is given by
\begin{equation}
  \begin{split}
    T = 4.35
    \left(\frac{2 \rho}{k\omega_d^2}\right)^{5/6}
    k^{7/6}
    \left(H_{\mathrm{atm}} - h_{0}\right)^{5/6}
    \sec^{11/6}\theta \\
    \times \int_{h_{0}}^{H_{\mathrm{atm}}}
    C_n^2(h)
    \left(\frac{h - h_{0}}{H_{\mathrm{atm}} - h_{0}}\right)^{5/3}
    \mathrm{d}h,
  \end{split}
  \label{eq:Turbulence_integral}
\end{equation}
where $k = 2\pi/\lambda$, $\theta$ is the zenith angle, $h_0$ is the height of the receiver above ground and $H_{\mathrm{atm}}$ is the effective top of the turbulent atmospheric layer (taken to be $\sim 20$~km) \cite{andrews2005laser}.
The integral incorporates the vertical turbulence profile function $C_n^2(h)$, which quantifies the local turbulence strength as a function of altitude.
The factor $T$ increases with turbulence strength and optical path length and decreases with larger transmitter beam waist $\omega_d$, illustrating that wide initial beams are less sensitive to atmospheric spreading.

The function $C_n^2(h)$ can be modeled using the generalized Hufnagel--Valley (HV) parameterization \cite{hufnagel1964modulation,valley1980isoplanatic,hardy1998adaptive,vasylyev2017free}:
\begin{equation}
  \begin{split}
    C_n^2(h)
    = A \exp\!\left(-\frac{h}{100}\right)
    + B \exp\!\left(-\frac{h}{1500}\right) \\
    + C h^{10} \exp\!\left(-\frac{h}{1000}\right),
  \end{split}
  \label{eq:Cn2_profile}
\end{equation}
where the coefficients $A$, $B$, and $C$ represent turbulence contributions from the surface layer, the troposphere, and the tropopause, respectively. Typical parameter sets correspond to the HV~5--7, HV~10--10, and HV~15--12 profiles, describing conditions ranging from strong (sea level) to excellent (astronomical) seeing \cite{hufnagel1964modulation,pugh2020adaptive}. Representative values for the turbulence parameters considered in this work are listed in Table~\ref{tab:location_params}.

Figure~\ref{fig:spot_width_analysis} illustrates the spot size  at the ground for three wavelengths as a function of the zenith angle $\theta$ for a coastal atmosphere and a transmitter aperture of 1~m.

In free-space optical communication, the power collected by the receiver is fundamentally limited by the spatial spread of the beam relative to the receiver aperture and by the stochastic misalignment of the beam centroid. We analyze these two distinct attenuation mechanisms separately below.

\subsection{Geometric loss}

Geometric loss, often referred to as aperture truncation loss, occurs because the optical beam spreads due to diffraction and atmospheric turbulence as it propagates. Even with perfect alignment, a significant fraction of the transmitted beam power may fall outside the finite receiver aperture.

Assuming a circularly symmetric Gaussian intensity distribution, the geometric efficiency $\eta_{\mathrm{geo}}$ represents the fraction of the total optical power intercepted by the receiver aperture in the absence of pointing errors \cite{andrews2005laser}:
\begin{equation}
  \eta_{\mathrm{geo}} = 1 - \exp\!\left(-\frac{a_R^2}{2\omega_t^2}\right),
  \label{eq:geom_loss}
\end{equation}
where $a_R$ is the receiver aperture and $\omega_t$ is the turbulence-broadened beam radius at the receiver plane given by Eq.~\eqref{eq:beam_turb}. When $a_R \gg \omega_t$, we have that $\eta_{\mathrm{geo}} \to 1$, and truncation loss becomes negligible; conversely, when $a_R \ll \omega_t$, a significant portion of the beam is not captured, subsequently incurring high loss.

\subsection{Pointing loss}
Pointing loss arises from the residual mechanical jitter of the transmitter, which causes random displacements of the beam centroid relative to the center of the receiver aperture \cite{trinh2022statistical}. This misalignment reduces the effective power coupled into the system.

Modeling the pointing error as a Gaussian random variable with radial standard deviation $\sigma_p$, with the receiver at a distance $\rho$, we approximate the average pointing efficiency $\eta_{\mathrm{p}}$ as \cite{andrews2005laser,kaushal2017optical}:
\begin{equation}
  \eta_{\mathrm{p}} = \frac{\omega_t^2}{\omega_t^2 + 4(\sigma_p \rho)^2}.
  \label{eq:pointing_loss}
\end{equation}

This approximation relies on assumptions of small angular deviations and circular symmetry, which normally hold within the typical operating regime of satellite optical links. However, the pointing efficiency $\eta_{\mathrm{p}}$ degrades rapidly as the jitter $\sigma_p$ becomes significant relative to the beam radius $\omega_t$. In our simulations, we consider three different scenarios for the angular jitter $\sigma_p$: $0.5\,\mu\mathrm{rad}$ (low pointing error), $1.0\,\mu\mathrm{rad}$ (moderate), and $2.0\,\mu\mathrm{rad}$ (high).

\subsection{Coupling loss}

The coupling loss quantifies the fraction of transmitted optical power that falls outside the receiver field of view (FOV) and is not coupled into the detection system. In a free-space optical (FSO) receiver, the FOV defines the angular region over which the detector collects the incoming light. The choice of FOV directly determines the trade-off between signal collection efficiency and susceptibility to background noise from stray  photons \cite{born1999principles,goodman2005fourier,tyson2022principles}. A narrow FOV minimizes background photon flux but risks truncating the received signal, whereas a wide FOV improves signal collection at the cost of increased background noise \cite{lanning2021quantum}. Thus, the receiver FOV defines the optimal operating point for
background-limited FSO links
\cite{tyson2022principles,lanning2021quantum,hardy1998adaptive}.

For a circular, diffraction-limited telescope aperture $a_R$ observing at wavelength $\lambda$, the on-axis point-spread function (PSF) follows the Airy distribution \cite{born1999principles},
\begin{equation}
  I(\gamma_{\text{FOV}}) = I_0 \left[\frac{2 J_1(u_{\text{FOV}})}{u_{\text{FOV}}}\right]^2,
  \qquad u_{\text{FOV}} = \frac{ \pi a_R \sin\gamma_{\text{FOV}}}{\lambda},
  \label{eq:airy-pattern}
\end{equation}
where $I(\gamma_{\text{FOV}})$ is the intensity at an off-axis angle $\gamma_{\text{FOV}}$ relative to the optical axis, $I_0$ is the peak intensity at the center of the focal spot, and $J_1(\cdot)$ denotes the first-order Bessel function of the first kind. The Airy pattern represents the fundamental angular response of an ideal, unobstructed circular aperture.

The cumulative energy, $L_{\mathrm{FOV}}$, collected within a receiver FOV of half-angle $\gamma_{\mathrm{FOV}}$ is obtained by integrating the PSF over  $0$ to $\gamma_{\mathrm{FOV}}$. The result is the following analytic expression \cite{born1999principles,schroeder2000astronomical}:

\begin{equation}
  \label{eq:encircled-energy}
  L_{\mathrm{FOV}}(\gamma_{\mathrm{FOV}}) = 1 - \left[J_0(u_{\mathrm{FOV}})\right]^2 - \left[J_1(u_{\mathrm{FOV}})\right]^2,
\end{equation}
where $J_0(\cdot)$ is the zeroth-order Bessel function of the first kind.
The solid angle subtended by a receiver FOV of half-angle $\gamma_{\mathrm{FOV}}$ is
\begin{equation}
  \Omega_{\mathrm{FOV}} = 2\pi \left(1 - \cos \gamma_{\mathrm{FOV}}\right).
  \label{eq:solid-angle}
\end{equation}

The first minimum of the Airy pattern occurs at
\begin{equation}
  \gamma_{\mathrm{Airy}} \approx \frac{1.22\lambda}{a_R},
  \label{eq:airy-angle}
\end{equation}
which corresponds to the angular radius of the first dark ring.
Substituting $\gamma_{\mathrm{Airy}}$ into Eq.~\eqref{eq:encircled-energy}
yields under small angle approximation,
\begin{equation}
  L_{\mathrm{FOV}}(\gamma_{\mathrm{Airy}}) \approx 0.838.
\end{equation}

This means that approximately $83.8\%$ of the total optical power from an
on-axis diffraction-limited point source is contained within the first
Airy disk.
This defines the conventional diffraction-limited collection efficiency
used as a benchmark in FSO system design
\cite{born1999principles,goodman2005fourier,andrews2005laser}.
For narrower FOVs, we have that $L_{\mathrm{FOV}} < 0.838$, reflecting truncation of
the PSF core; for larger FOVs, it follows that $L_{\mathrm{FOV}} \to 1$, though
background photon collection increases with $\Omega_{\mathrm{FOV}}$.

In realistic atmospheric propagation conditions, random refractive-index fluctuations impose phase distortions across the telescope pupil, broadening the PSF and reducing the fraction of light remaining in the  core. The degree of degradation is commonly quantified by the \emph{Strehl ratio} $S$, defined as the ratio of the on-axis intensity in the aberrated image to that of a perfect diffraction-limited system. An empirical form for $S$ frequently adopted in FSO and astronomical optics is \cite{lanning2021quantum,gruneisen2016adaptive,sasiela2012electromagnetic}
\begin{equation}
  S = \left[1 + \left(\frac{a_R}{r_0}\right)^{5/3}\right]^{-6/5},
  \label{eq:strehl}
\end{equation}
where $r_0$ is the Fried coherence parameter, given by
\begin{equation}
  r_{0} = \left[ 0.423 \, k^{2} \, \sec(\theta)
    \int_{h_{0}}^{H_{\mathrm{atm}}} C_{n}^{2}(h) \, dh
    \right]^{-3/5},
\end{equation}
where $k$ is the wavenumber, $\theta$ is the zenith angle, $H_{\mathrm{atm}}$ is the effective height of the turbulent atmospheric layer, $h_{0}$ is the height of the OGS above ground, and $C_n^2(h)$ is given by Eq.~\eqref{eq:Cn2_profile}.
As turbulence strength increases, $S$ decreases, representing the fraction of optical power that remains coherently focused.

Combining the independent effects of turbulence-induced phase
degradation and finite-FOV truncation, the mean coupling efficiency can
be written as \cite{tyson2022principles,roddier1981effects}
\begin{equation}
  \eta^{\mathrm{FS}}_{\mathrm{cpl}}(\gamma_{\mathrm{FOV}}) = L_{\mathrm{FOV}}(\gamma_{\mathrm{FOV}}) \, S,
  \label{eq:eta_cpl_general}
\end{equation}
where the superscript FS is used to denote that we assume free-space coupling.

The collected signal could also be relayed into single-mode fiber (SMF) before detection. The SMF acts as a spatial filter by admitting only the fundamental transverse mode, thereby suppressing incoherent background in the received field.  In this scenario, the overall coupling efficiency quantifies the fraction of optical power transferred from the telescope focus into the guided mode. This quantity depends on the overlap between the focal-plane field Airy pattern and the fiber’s Gaussian mode \cite{shaklan1988coupling,coude1994integrated,jovanovic2017efficient}. Unlike a detector with a hard-edged angular acceptance, the SMF couples light according to its Gaussian amplitude weighting, and can be characterized by the effective FOV half-angle \cite{scriminich2022optimal,acosta2023quantum}

\begin{equation}
  \gamma_{\mathrm{SMF}}\approx\frac{2.24\lambda}{\pi a_R}.
\end{equation}
We note that here $\gamma_{\mathrm{SMF}}$ is an effective angular acceptance used for background/throughput modeling via Eq.~\eqref{eq:solid-angle}, and not a hard-edged truncation angle that can be  substituted into Eq.~\eqref{eq:encircled-energy}.

Even without turbulence or alignment errors, the intrinsic mismatch between the focused Airy pattern and the fundamental fiber mode limits the achievable coupling. For optimal coupling one should choose a Gaussian matched to the size of the Airy core, rejecting all the outer rings. Under this Gaussian-mode approximation, the maximum efficiency is typically quoted as $\eta_{\mathrm{max}} \approx 0.81$ \cite{shaklan1988coupling,jovanovic2017efficient}. Using the exact eigenmode of a step-index fiber reduces the theoretical limit to $\eta_{0} = 0.786$ \cite{fidler2008application}, which is what we choose for the simulations. Atmospheric turbulence further reduces the coherent fraction of the incoming wavefront, decreasing fiber coupling proportionally to the Strehl ratio $S$. Thus, we model the mean coupling efficiency as,
\begin{equation}
  \eta^{\mathrm{SMF}}_{\mathrm{cpl}} = \eta_0 \, S.
  \label{eq:eta_smf}
\end{equation}

\subsection{Adaptive optics compensation}

Adaptive optics (AO) systems are implemented to reduce turbulence-induced phase distortions by actively correcting phase aberrations across the receiver aperture \cite{tyson2022principles}.
In this framework, AO correction is modeled as an effective enhancement of the Fried coherence value from its uncorrected value $r_0$ to a closed-loop corrected value $r_0^{\mathrm{CL}} > r_0$.
This improvement directly increases the Strehl ratio defined previously, thereby reducing coupling loss.

Following the model of Lanning \textit{et al.}~\cite{lanning2021quantum}, the AO-compensated Fried parameter is expressed as
\begin{equation}
  \begin{split}
    r_0^{\mathrm{CL}}
    = 1.03^{3/5} a_R
    \bigg[
      \left(\frac{f_G}{f_c}\right)^{5/3}
      +\left(\frac{\pi}{2}\frac{f_{\mathrm{TG}}}{f_{tc}}\right)^{2}
      \bigg]^{-3/5},
  \end{split}
  \label{eq:r0cl}
\end{equation}
where  $f_c$ and $f_{tc}$ are the closed-loop correction bandwidths of the deformable mirror and tip--tilt subsystems in the AO system, respectively.
The parameters $f_G$ and $f_{TG}$, on the other hand, are the Greenwood and tracking--Greenwood frequencies \cite{hardy1998adaptive,tyson2022principles}, that characterize the temporal dynamics of atmospheric turbulence, defined as
\begin{align}
  \begin{split}
    f_G & = \bigg[\,0.1022\,k^{2}\sec\theta                                                                        \\
        & \quad \times \int_{h_{\mathrm{OGS}}}^{h_{\mathrm{atm}}} C_n^2(h)\,V^{5/3}(h)\,\mathrm{d}h\,\bigg]^{3/5},
  \end{split}
  \label{eq:fG} \\[6pt]
  \begin{split}
    f_{\mathrm{TG}} & =
    5.268\times10^{-2}\,a_R^{-1/6}\,k                 \\
                    & \quad \times \left[\,\sec\theta
      \int_{h_{\mathrm{OGS}}}^{h_{\mathrm{atm}}} C_n^2(h)\,V^{2}(h)\,\mathrm{d}h\,\right]^{1/2},
  \end{split}
  \label{eq:fTG}
\end{align}
where  $V(h)$ is the wind speed profile. This function is commonly represented using
the Greenwood wind model \cite{greenwood1977bandwidth}, which assumes a
Gaussian-shaped velocity peak near the tropopause.
The model captures the transition between lower-altitude boundary-layer
winds and strong high-altitude jet streams. In doing so, we have that $V(h)$ can be expressed as
\cite{greenwood1977bandwidth,trinh2022statistical,andrews2005laser}
\begin{equation}
  V(h) = v_{\mathrm{OGS}} + 30
  \exp\!\left[-\left(\frac{h +h_{\mathrm{OGS}}-12448}{4800}\right)^{2}\right],
  \label{eq:greenwood_wind}
\end{equation}
where $v_{\mathrm{OGS}}$ is the wind speed at the OGS and $h_{\mathrm{OGS}}$ is the altitude of the OGS above sea level.
For the case of a GEO downlink, the satellite has negligible apparent motion relative to the ground observer, and therefore the pseudo-wind term associated with satellite motion is omitted from the model.

Equations~\eqref{eq:r0cl}--\eqref{eq:fTG} show that increasing the AO control bandwidths $f_c$ and $f_{tc}$ increases $r_0^{\mathrm{CL}}$, improving wavefront coherence and raising on-axis intensity. Following~\cite{lanning2021quantum}, our simulations consider $f_{tc}=60~\mathrm{Hz}$ and $f_c\in\{130,200,500\}\,\mathrm{Hz}$, corresponding to low, moderate, and advanced AO correction.

\subsection{Absorption/scattering loss}
To model the channel loss arising from atmospheric absorption and scattering processes, we employ a radiative transfer code, via the \texttt{libRadtran} package \cite{emde2016libradtran,mayer2005libradtran}. We note that the MODTRAN radiative transfer code \cite{berk2014modtran} provides an alternative, widely used implementation with comparable capabilities. Both codes offer suites of standardized atmospheric profiles that differ in vertical temperature distributions and trace-gas abundances, which strongly affect infrared absorption. Among these, the widely adopted 1976 U.S. Standard Atmosphere \cite{united1976us} serves as a median climatological reference specifying altitude-dependent pressure, temperature, and constituent concentrations, which is what we use in this work.

\texttt{libRadtran} incorporates aerosol contributions to the losses using Shettle--Fenn type models \cite{shettle1979models,shettle1990models}. These models parameterize losses in terms of visibility, and prescribe distinct optical and microphysical properties for rural, urban, and maritime air masses \cite{shettle1979models}, our three representative aerosol types. For each case, the aerosol loading is controlled via the \texttt{libRadtran} visibility parameter, which we set to $5$, $10$, and $23$ km for coastal, urban, and rural conditions, respectively. These settings span a practical range of operational scenarios, from hazy coastal channels to clear-air rural links. The full set of aerosol types and visibility parameters adopted in this work is summarized in Table~\ref{tab:location_params}. Operationally, these visibility regimes may be related to prevailing conditions described by METAR and TAF aviation weather reports \cite{metarandtaf}.

For each atmospheric and aerosol scenario, we use \texttt{libRadtran} to compute the wavelength-dependent atmospheric transmittance $\eta_{\mathrm{atm}}$ for different zenith angles $\theta$. The results inherently account for the geometric path length through the atmosphere as well as altitude-dependent extinction.
% To provide intuition, one can consider the standard secant approximation, which assumes that the effective optical path length scales as $\sec\theta$ \cite{P21},
% \begin{equation}
%   \eta_{\rm atm}
%   = \bigl(\eta_{\rm atm}^{\rm zen}\bigr)^{\sec\theta},
% \end{equation}
% where $\eta_{\rm atm}^{\rm zen}$ is the transmittance at zenith. This relation gives a useful qualitative sense of how the transmittance decreases with increasing zenith angle. 

\texttt{libRadtran} also simulates downlink transmittance under cloudy skies
by incorporating clouds as layers with specified optical thickness.
The required cloud optical thickness (COT) data can be obtained from climate satellites, which we discuss below.

\subsubsection{Clouds}
\label{sec:clouds}
\begin{figure}[t]
  \centering
  \includegraphics[width=\linewidth]{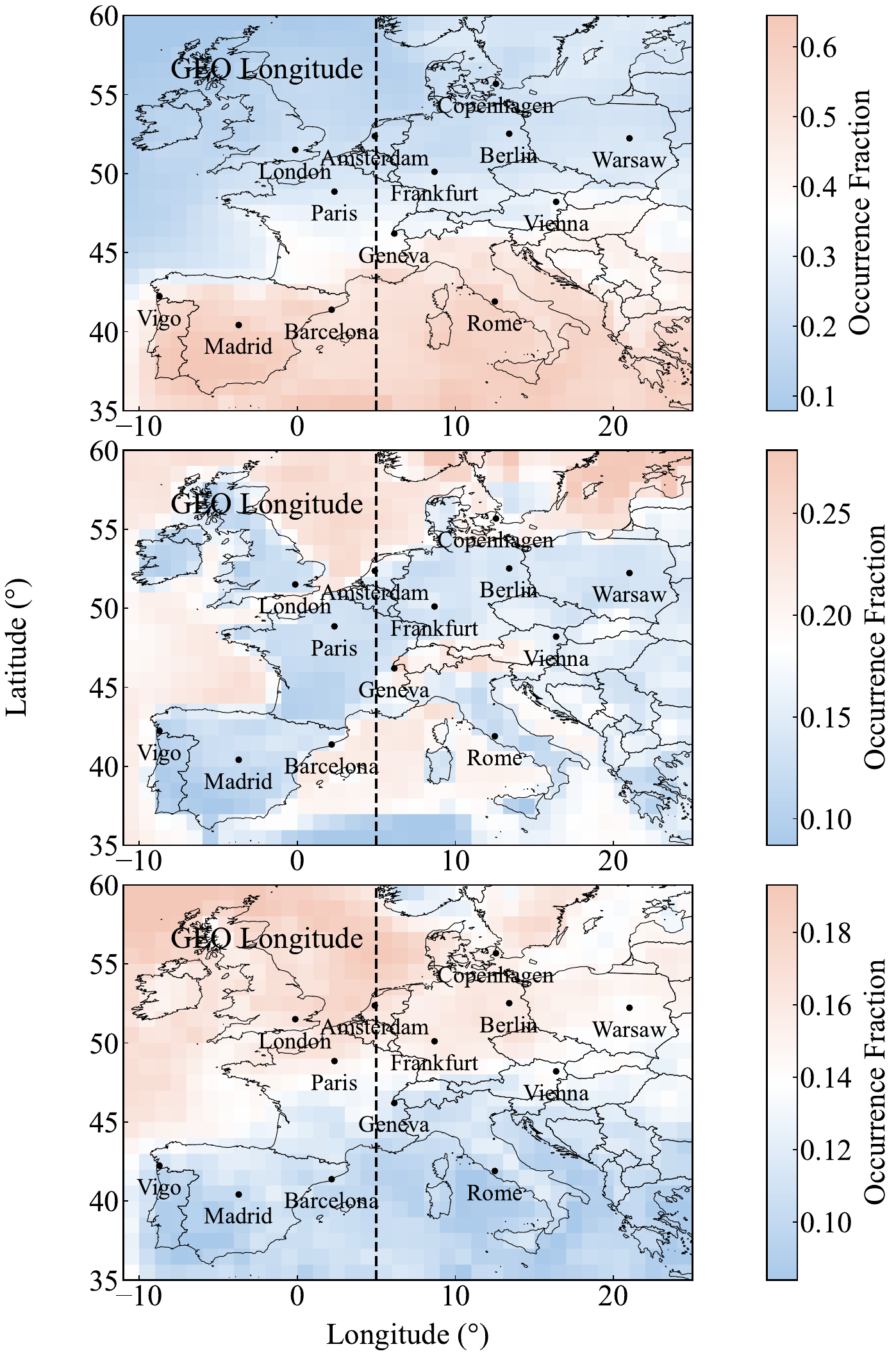}

  \caption{Climatological mean fraction of time for different sky conditions across Europe, derived from CRAAS satellite data. The maps display the occurrence probability for clear-sky conditions (top), thin-cloud conditions (center), and thick-cloud conditions (bottom). The color bars indicate the fractional occurrence for each regime, highlighting regional variations in optical transmission availability.}
  \label{fig:clouds}
\end{figure}

To account for cloud-induced effects in the QKD downlink analysis, we employ the CRAAS dataset based on the CLAAS-2.1 climate data record~\cite{tzallas2022craas,EUMETSAT93:online}. CRAAS provides a 14-year (2004--2017) satellite-derived climatology of cloud properties. These data come from the Spinning Enhanced Visible and InfraRed Imager (SEVIRI) onboard the geostationary Meteosat Second Generation (MSG) satellites MSG-1, MSG-2, and MSG-3, with 15-minute temporal and $1^\circ \times 1^\circ$ spatial resolution, spanning the European region.

For each grid cell, the dataset provides the mean fractional occurrence of six cloud optical thickness (COT) categories. In this study, we use the average coverage fractions, from the 2017 dataset, for the clear sky, first COT bin, and second COT bin, which we call clear-sky, thin-cloud, and thick-cloud regimes, respectively, as shown in Fig.~\ref{fig:clouds}. The losses of all remaining COT categories are prohibitive for GEO-QKD and therefore we exclude them from the study. The first and second COT bins have midpoints at 0.65 and 2.45, respectively. The clear-sky fraction of a given grid cell, $f_{\text{clear}}$, is defined as
\begin{equation}
  f_{\text{clear}} = 1 - \sum_{i=1}^6 f_{\text{COT},i},
\end{equation}
where $f_{\text{COT},i}$ is the mean fractional coverage of COT bin $i$ at the given grid cell.

For cloud-attenuated scenarios, the optical thickness of the first two COT bins is passed to \texttt{libRadtran} for calculating the zenith angle dependent transmission losses, using the default cloud settings in this package. The atmospheric and cloud input files used for these calculations are archived in Ref.~\cite{mannalath2026atmospheric}.

\subsection{Overall system loss}
Combining all the effects detailed so far, the overall transmittance $\eta_{\mathrm{sys}}$ of the satellite downlink system is given by

\begin{equation}
  \label{eq:totloss}
  \eta_{\mathrm{sys}}=\eta_{\mathrm{geo}} \eta_{\mathrm{p}}\eta_{\mathrm{cpl}} \eta_{\mathrm{atm}} \eta_{\mathrm{R}} \eta_{\mathrm{D}},
\end{equation}
where $\eta_{\mathrm{R}}$ denotes the total internal transmittance of the receiver subsystem excluding the single-photon detectors. This term captures the aggregate efficiency of the telescope optics (primary and secondary mirrors), relay mirrors, and the narrowband spectral filter used for background suppression. Throughout this work we adopt a nominal receiver optical loss of 4~dB, representative of typical ground-based receivers utilized in space-QKD studies \cite{liorni2019satellite,dirks2021geoqkd,lu2022micius,sidhu2023finite}. Finally, $\eta_{\mathrm{D}}$ represents the quantum efficiency of the single-photon detectors, which we assume is the same for all the detectors in a receiver module.

Figure~\ref{fig:zenith_loss_components_grid_sm} decomposes the total attenuation into the physical terms that drive the earlier contour plots. The solid curve $\eta_{\mathrm{sys}}$ sits far above any individual constant offset, showing that the link is dominated by the large geometric baseline loss together with the zenith-dependent growth of atmospheric and coupling losses. Each column represents one bundled urban system specification (Systems 1--3) that combines detector choice with different pointing-jitter assumptions, aperture sizes, and adaptive-optics strength. To keep the visual comparison compact, the plotted receiver-side term combines the fixed $4$ dB receiver-optics contribution with the detector-efficiency loss as $\eta_{\mathrm{R}}\eta_{\mathrm{D}}$.
\begin{figure}[htbp]
  \includegraphics[width=\linewidth]{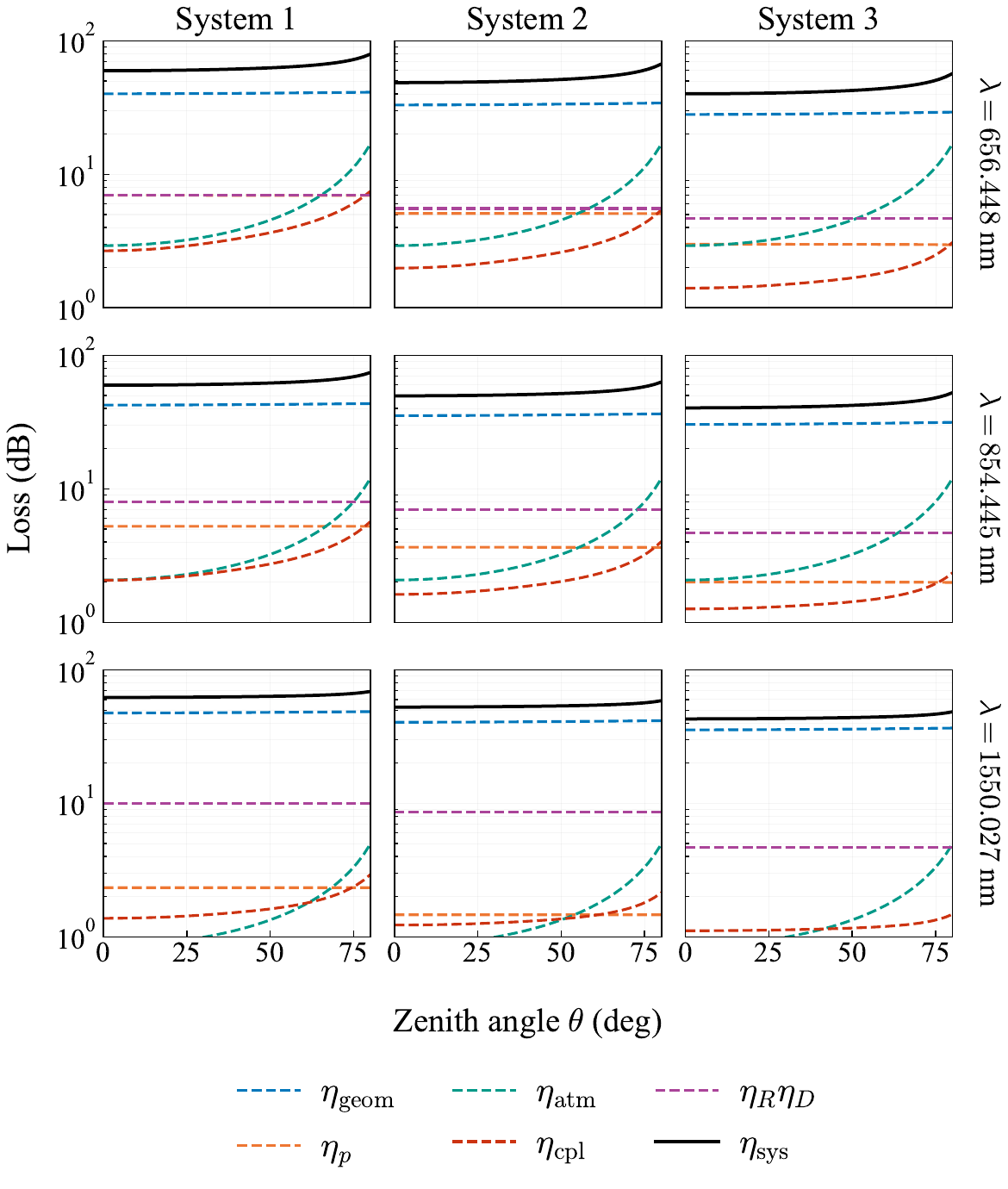}
  \caption{Loss-budget decomposition with loss in dB (y axis, logarithmic scale) versus zenith angle $\theta$ (x axis) for an urban location baseline, using SMF coupling and a fixed 4~dB receiver optics/filter loss. Dashed colored curves show the dB losses corresponding to the component transmission terms $\eta_{\mathrm{geo}}$ (geometric collection), $\eta_{\mathrm{p}}$ (pointing), $\eta_{\mathrm{atm}}$ (atmospheric transmission), $\eta_{\mathrm{cpl}}$ (receiver coupling), and $\eta_{\mathrm{R}}\eta_{\mathrm{D}}$ (receiver optics times detector efficiency), while the solid black curve gives the total system loss corresponding to $\eta_{\mathrm{sys}}$ (overall system transmission). Rows correspond to $\lambda=656.448$ nm, $\lambda=854.445$ nm, and $\lambda=1550.027$ nm. Columns correspond to bundled urban system specifications: System 1 [detector: APD Spec A; $a_T=0.50$ m; $a_R=1.00$ m; pointing jitter: high; AO: low ($f_c=130$ Hz)], System 2 [detector: APD Spec B; $a_T=0.75$ m; $a_R=1.50$ m; pointing jitter: moderate; AO: moderate ($f_c=200$ Hz)], and System 3 [detector: SNSPD Spec A; $a_T=1.00$ m; $a_R=2.00$ m; pointing jitter: low; AO: high ($f_c=500$ Hz)].}
  \label{fig:zenith_loss_components_grid_sm}
\end{figure}

\subsection{Background noise}
\label{sec:background_noise}

This subsection gives the explicit calculation of the mean background-photon level $\bar{n}_{\mathrm{B}}$ and the derived noise-click probability $p_{\mathrm{noise}}$ that define the nighttime and daytime noise scenarios used in Sec.~\ref{sec:feas}.

For a receiver with angular field of view $\Omega_{\mathrm{FOV}}$, using a detector within a temporal window $\Delta t$ and a spectral filter of bandwidth $\Delta \lambda$ centered around wavelength $\lambda$ (see Table~\ref{tab:filters}), the average number of background photons received within the detection window is given by \cite{liorni2019satellite,pirandola2021limits,P21}.

\begin{table}[ht]
  \centering
  \small
  \begin{tabularx}{\columnwidth}{>{\raggedright\arraybackslash}X@{\hspace{1em}}cccc}
    \toprule
    \textbf{Reference} & \textbf{\shortstack{Cntr.                                     \\(nm)}} & \textbf{\shortstack{BW\\(GHz)}} & \textbf{\shortstack{BW\\(nm)}} & \textbf{\shortstack{Loss\\(dB)}} \\
    \midrule
    Thorlabs FPQFA~\cite{thorlabs_fpqfa_fpqsa_fabry_perot_filters_2025}
                       & 656                       & 30        & $\sim$0.043 & $<$0.97 \\
                       & 854                       & 30        & $\sim$0.072 & $<$0.97 \\
                       & 1550                      & 30        & $\sim$0.24  & $<$0.97 \\
    \midrule
    Alluxa ULTRA~\cite{alluxa-6564-0p1-od4,alluxa_854_0.75_od6_ultra_narrow_bandpass,alluxa_1550.6_1.8_od4_2025}
                       & 656                       & $\sim$ 70 & 0.1         & $<$1    \\
                       & 854                       & $\sim$310 & 0.75        & $<$0.7  \\
                       & 1550                      & $\sim$225 & 1.8         & $<$0.46 \\
    \midrule
    OptiGrate BragGrate~\cite{optigrate_quantum_tech_2025,lumeau2010ultra}
                       & 656                       & 10        & $\sim$0.014 & $<$0.46 \\
                       & 854                       & 10        & $\sim$0.024 & $<$0.46 \\
                       & 1550                      & 10        & $\sim$0.08  & $<$0.46 \\
    \midrule
    Exail IXC-FBG-PS ~\cite{Exail_UltraNarrowBandwithFilters}
                       & 1550                      & 1         & $\sim$0.006 & $<$1    \\
    \bottomrule
  \end{tabularx}
  \caption{Comparative specifications of  bandpass filters. Loss in dB is reported or computed from stated peak transmission: $L=-10\log_{10}(T)$. Bandwidths were converted from frequency to length using   $\,\Delta\lambda=(\lambda^2/c)\Delta\nu\,$. In the table, Cntr. is the central wavelength and BW is the bandwidth}
  \label{tab:filters}
\end{table}

\begin{equation}
  \label{eq:bgphot}
  \bar{n}_{\mathrm{B}} = H_\lambda^{\mathrm{sky}} \, \Gamma_{\mathrm{R}},
\end{equation}
where the effective geometric-spectral-temporal acceptance factor is given by
\begin{equation}
  \Gamma_{\mathrm{R}} := \Delta \lambda \, \Delta t \, \Omega_{\mathrm{FOV}} \, a_R^2/4.
\end{equation}
We note that representative narrowband filter options used to parameterize $\Delta\lambda$ are listed in Table~\ref{tab:filters}.

The parameter $H_\lambda^{\mathrm{sky}}$ represents the photon flux spectral density from the sky background. It is expressed as
\begin{equation}
  H_\lambda^{\mathrm{sky}} := \frac{\lambda \, \tilde{H}_\lambda^{\mathrm{sky}} 10^6}{2 \hbar c},
\end{equation}
where $\hbar$ is the reduced Planck’s constant, $c$ is the speed of light in vacuum, and $\tilde{H}_\lambda^{\mathrm{sky}}$ denotes the spectral radiance of the sky (units of mW m$^{-2}$ nm$^{-1}$ sr$^{-1}$). To obtain realistic values for $\tilde{H}_\lambda^{\mathrm{sky}}$ under different observing conditions, we employ radiative transfer and astronomical sky background models for daytime and nighttime conditions.

\subsubsection{Daytime conditions}

\begin{figure}[htbp]
  \includegraphics[width=\linewidth]{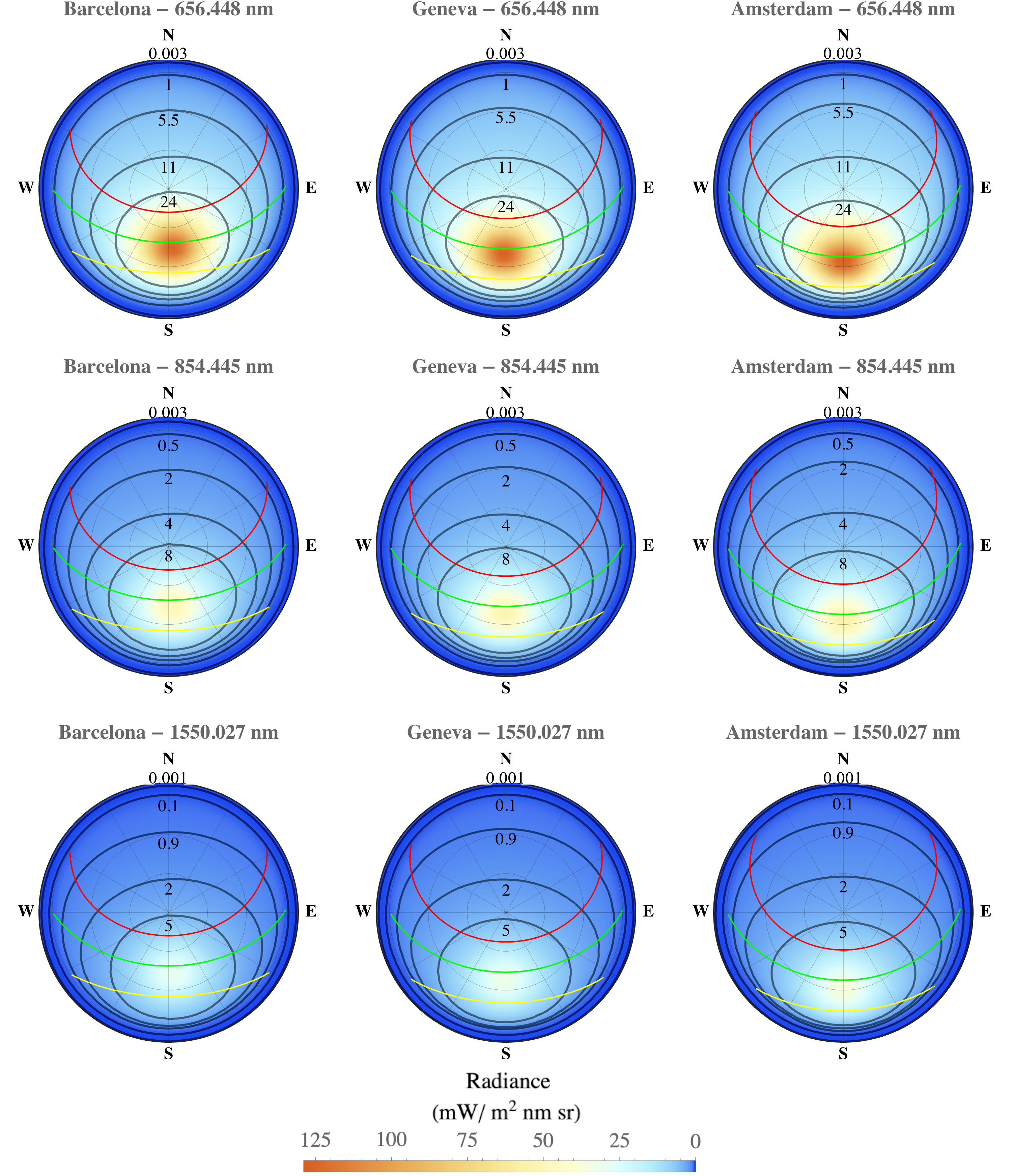}
  \caption{Sky-radiance solar position maps for Barcelona, Geneva, and Amsterdam with the line of sight fixed toward a GEO satellite at longitude $5^\circ$E. Solar paths for summer solstice (red), equinox (green), and winter solstice (yellow) are overlaid on radiance contours. The similarity in radiance distribution across latitudes justifies the use of a unified background noise model for annual key estimation.}
  \label{fig:sky_radiance}
\end{figure}

For the daytime analysis, we use the sky-radiance maps shown in Fig.~\ref{fig:sky_radiance}, simulated with \texttt{libRadtran} \cite{emde2016libradtran} software package with the atmospheric and aerosol models mentioned before. The solar source is parameterized by an extraterrestrial solar spectrum data \cite{kurucz1994synthetic}, and the solar zenith and azimuth angles relative to the receiver field of view are input to determine the sky radiance. The line of sight is fixed toward a GEO satellite at $5^\circ$E, while the solar tracks for Barcelona ($41^\circ$N), Geneva ($46^\circ$N), and Amsterdam ($52^\circ$N) show how the Sun moves relative to that fixed pointing direction over the solstices and equinox. Although the solar arc shifts with latitude, the relative Sun-satellite geometry remains qualitatively similar across these representative European cities. Each case contains a bright inner region where the Sun passes sufficiently close to the line of sight to drive the sky radiance to its highest values, together with broader surrounding regions of lower daytime radiance.
We use a unified European daytime partition in which the innermost high-radiance region is treated as an exclusion zone where key generation is paused, while the remaining usable daylight interval is represented by the low-, moderate-, and high-background scenarios associated with the outer radiance contours.

Daytime background is modeled for three representative solar zenith angles: $0^{\circ}$, $60^{\circ}$, and $90^{\circ}$, with azimuth fixed toward the east. These correspond to high, moderate, and low daylight conditions, respectively. The receiver is assumed to observe the satellite at a zenith angle of 60$^{\circ}$ and southward azimuth, representative of a mid-elevation downlink.

The diffuse sky spectral radiance,
$\tilde{H}_{\lambda}^{\mathrm{sky}}$,
is computed using \texttt{libRadtran} for each environment
(rural, urban, coastal) and wavelength
(656.448\,nm, 854.445\,nm, and 1550.027\,nm).
These radiances directly feed into the photon-flux expression used to
compute the mean number of background photons $\bar{n}_{\mathrm{B}}$.

\subsubsection{Nighttime conditions}
\label{sec:noise_extrapolation}

\begin{figure}[h]
  \centering
  \includegraphics[width=\linewidth]{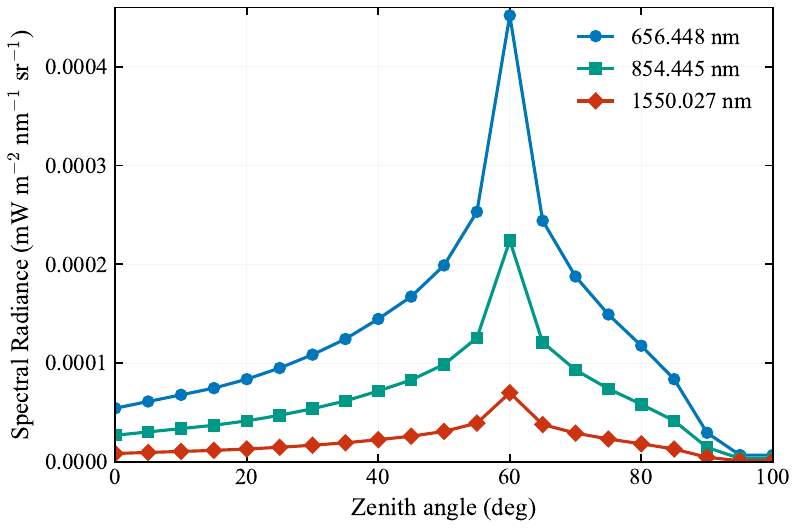}
  \caption{Spectral radiance (y axis) versus the full moon's zenith angle (x axis), simulated using ESO SkyCalc for wavelengths of 656.448 nm, 854.445 nm, and 1550.027 nm. The receiver is fixed at a zenith angle of $60^\circ$, resulting in a peak radiance when the moon aligns with the receiver's line of sight. For the nighttime radiance analysis presented in this work, the values at a lunar zenith angle of $30^\circ$ are used as a baseline and extrapolated to varying atmospheric conditions.}
  \label{fig:lunar_radiance}

\end{figure}

Nighttime background noise under full Moon is evaluated using the ESO SkyCalc tool~\cite{ESO_SkyCalc_2025} for Cerro Paranal, adopted as a representative of a rural dark-sky site~\cite{noll2013cerro,jones2013advanced}. This model accounts for multiple background contributions including airglow continuum and emission lines, scattered starlight, zodiacal light, and reflected moonlight. The tool allows configuration of the lunar phase, the zenith and azimuth angles of the Moon relative to the receiver’s pointing direction. The Moon is placed at the same azimuth as the satellite line of sight, and the satellite is observed at a zenith angle of $60^{\circ}$. Figure~\ref{fig:lunar_radiance} shows the resulting diffuse sky spectral radiance, $\tilde{H}_{\lambda,\mathrm{rural}}^{\mathrm{moon}}(z_{\mathrm{moon}})$, as a function of the Moon zenith angle $z_{\mathrm{moon}}$.

For the simulation, a single representative case is adopted at
$z_{\mathrm{moon}} = 30^{\circ}$.
This simplification is justified because nighttime illumination levels
remain well below the critical noise regime for QKD, and finer
discretization has negligible impact on system performance.

Urban and coastal values are obtained by scaling the rural full-Moon radiances using the ratios of the daytime solar radiances between locations at the same wavelength. We define the scaling factor $f_{\mathrm{loc}}$ as:
\begin{align}
  \tilde{H}_{\lambda,\mathrm{loc}}^{\mathrm{moon}}(30^{\circ})
   & =
  \tilde{H}_{\lambda,\mathrm{rural}}^{\mathrm{moon}}(30^{\circ})\,
  f_{\mathrm{loc}},
  \\
  f_{\mathrm{loc}}
   & =
  \frac{\tilde{H}_{\lambda,\mathrm{loc}}^{\mathrm{sun}}(30^{\circ})}
  {\tilde{H}_{\lambda,\mathrm{rural}}^{\mathrm{sun}}(30^{\circ})},
\end{align}
where $\mathrm{loc} \in \{\mathrm{urban},\mathrm{coastal}\}$. The corresponding scaling factors calculated for each wavelength and environment are listed in Table~\ref{tab:moon_scaling_factors}.

\begin{table}[ht]
  \centering
  \caption{Full-Moon radiance scaling factors used to obtain urban and coastal nighttime background radiances from the rural baseline.}
  \label{tab:moon_scaling_factors}
  \begin{tabular}{ccc}
    \toprule
    \textbf{Wavelength [nm]} & \textbf{$f_{\mathrm{urban}}$} & \textbf{$f_{\mathrm{coastal}}$} \\
    \midrule
    656.448                  & 1.03                          & 2.15                            \\
    854.445                  & 1.17                          & 2.86                            \\
    1550.027                 & 1.58                          & 7.55                            \\
    \bottomrule
  \end{tabular}
\end{table}

This scaling approximation is physically motivated by the fact that both sunlight and moonlight are broadband extraterrestrial sources interacting with identical atmospheric constituents. Since the scattering phase functions and extinction properties of the aerosols dominating the urban and coastal contributions are independent of the source intensity, the relative enhancement in sky radiance due to local aerosol loading is expected to remain consistent across both illumination regimes. We tabulate the radiance values used for the simulations in Table~\ref{tab:combined_radiance}.

\section{Simulation parameters}
\label{sec:params_appendix}

In this appendix, we specify the parameter values used in the numerical simulations. These parameters span the orbital geometry, atmospheric conditions, hardware specifications, and protocol settings that collectively define the GEO-QKD system under study.

\subsection{Locations}\label{sec:locations}

To assess the feasibility of GEO-QKD under diverse operational conditions, we consider three representative  locations for the OGS that span a practical range of atmospheric environments. These scenarios are characterized by their altitude, visibility conditions, and atmospheric turbulence profiles, as summarized in Table~\ref{tab:location_params}.

The three locations cover different altitude configurations that influence the atmospheric path length and aerosol distribution. The coastal scenario is situated at sea level (0 km altitude), maximizing the atmospheric column thickness and aerosol interaction. The moderate urban location is positioned at 0.2 km altitude, representing a slightly elevated OGS that reduces the lower-altitude aerosol layer traversed by the optical beam. The rural scenario is located at 0.4 km altitude, which significantly decreases the interaction with dense near-surface aerosols.

Visibility conditions (VIS) vary substantially across the three scenarios, reflecting different aerosol loading environments. The coastal site has a severely degraded visibility of 5 km, typical of near-urban coastal scenario with high particulate matter concentrations due to human activity as well as aerosols from the ocean. The moderate urban location exhibits an improved visibility of 10 km, representing urban environments with moderate aerosol content. The rural site features an excellent visibility of 23 km, characteristic of clear rural atmospheres with minimal aerosol contamination. For reduced visibility conditions (i.e., when \(2 < \mathrm{VIS} < 10\) km), the aerosol extinction coefficient is modeled as constant up to approximately 1 km altitude, while for clear conditions (i.e., when \(23 < \mathrm{VIS} < 50\) km), an exponential vertical distribution with rural-type scale heights is employed \cite{shettle1979models,metarandtaf}.

Atmospheric turbulence is described by site-specific Hufnagel--Valley (HV) $C_n^2(h)$ profiles. The wind field is taken from Eq.~\eqref{eq:greenwood_wind}, which sets the adaptive-optics temporal dynamics. The coastal, urban, and rural scenarios use the $(A,B,C)$ coefficients listed in Table~\ref{tab:location_params}, with ground-wind values $v_{\mathrm{OGS}}=25$, $10$, and $5$~m/s, respectively. The coastal scenario is represented by a stronger turbulence profile than the urban baseline, while the rural one adopts a weaker HV 10--10-like profile representative of average astronomical sites~\cite{hufnagel1964modulation,valley1980isoplanatic,hardy1998adaptive,pugh2020adaptive}. In this implementation, the AO-related Greenwood frequencies are therefore set by the combination of the location-dependent $C_n^2(h)$ profile and the Bufton-type high-altitude wind enhancement in Eq.~\eqref{eq:greenwood_wind}, rather than by a single fixed high-altitude wind speed.

\subsection{Wavelengths}\label{sec:wavelengths}

We focus on three representative wavelengths: 656.448~nm (H$\alpha$ line), 854.445~nm (Ca~II line), and 1550.027~nm (C-band representative). The first two lie on prominent Fraunhofer lines in the solar spectrum, where the Sun itself is significantly dimmer because a substantial fraction of the radiation at those wavelengths is absorbed before it leaves the solar atmosphere. This creates narrow depressions in the daylight spectrum at Earth (see Fig.e~\ref{fig:radandtrans}), which can be exploited for daylight free-space QKD when combined with narrowband filtering. Ref.~\cite{abasifard2024ideal} identifies such Fraunhofer depressions over 400--1700~nm; particularly H$\alpha$ and Ca~II, as local optima when one accounts for solar background, atmospheric transmission and  realistic detectors. Indeed, the H$\alpha$ window around 656~nm was one of the earliest Fraunhofer-based proposals for daylight QKD~\cite{rogers2006daylight}.

As for atmospheric transmittance, 656 nm and 854 nm lie in comparatively transparent windows
%  656 nm is below the O$_2$ B-band (686–690 nm) and away from H$_2$O features near 718–740 nm; 854 nm sits beyond the O$_2$ A-band (759–770 nm) and short of the strong H$_2$O complex at 930–960 nm, with only weak H$_2$O structure around 820–840 nm. 
across the 650–900 nm window~\cite{tomaello2011link}. Moreover, these wavelengths remain within the effective detection range of Si APDs. Daylight free-space demonstrations of QKD at visible and near-infrared wavelengths have routinely operated in the 700--900~nm region precisely because it balances high Si APD efficiency with manageable atmospheric loss and background~\cite{hughes2002practical, peloso2009daylight,ko2018experimental,cai2024free}. This alignment with the Fraunhofer-line strategy justifies the inclusion of these two wavelengths for detailed modeling and comparison.

The third wavelength, 1550.027~nm, is chosen as a representative telecom C-band point. We chose this specific value because it is the nearest tabulated wavelength to 1550~nm in the extraterrestrial solar spectrum used for the radiance calculations \cite{kurucz1994synthetic}. The C-band is extensively used in modern optical communications and has become increasingly relevant to satellite-based QKD~ \cite{avesani2021full}. Operation in this band leverages high atmospheric transmittance and reduced molecular (Rayleigh) scattering compared to visible bands~\cite{hearne2025wavelength, tomaello2011link}.  Rayleigh scattering decreases as $\lambda^{-4}$ \cite{bodhaine1999rayleigh,liou2002atmospheric}, while aerosol scattering efficiencies also diminish with increasing wavelength for the particle size distributions described in the Shettle--Fenn models \cite{shettle1979models}. From a systems perspective, it also offers superior single-mode fiber coupling into the OGS back-end infrastructure and full compatibility with existing telecom hardware.

\begin{figure}[h]
  \centering
  \includegraphics[width=\linewidth]{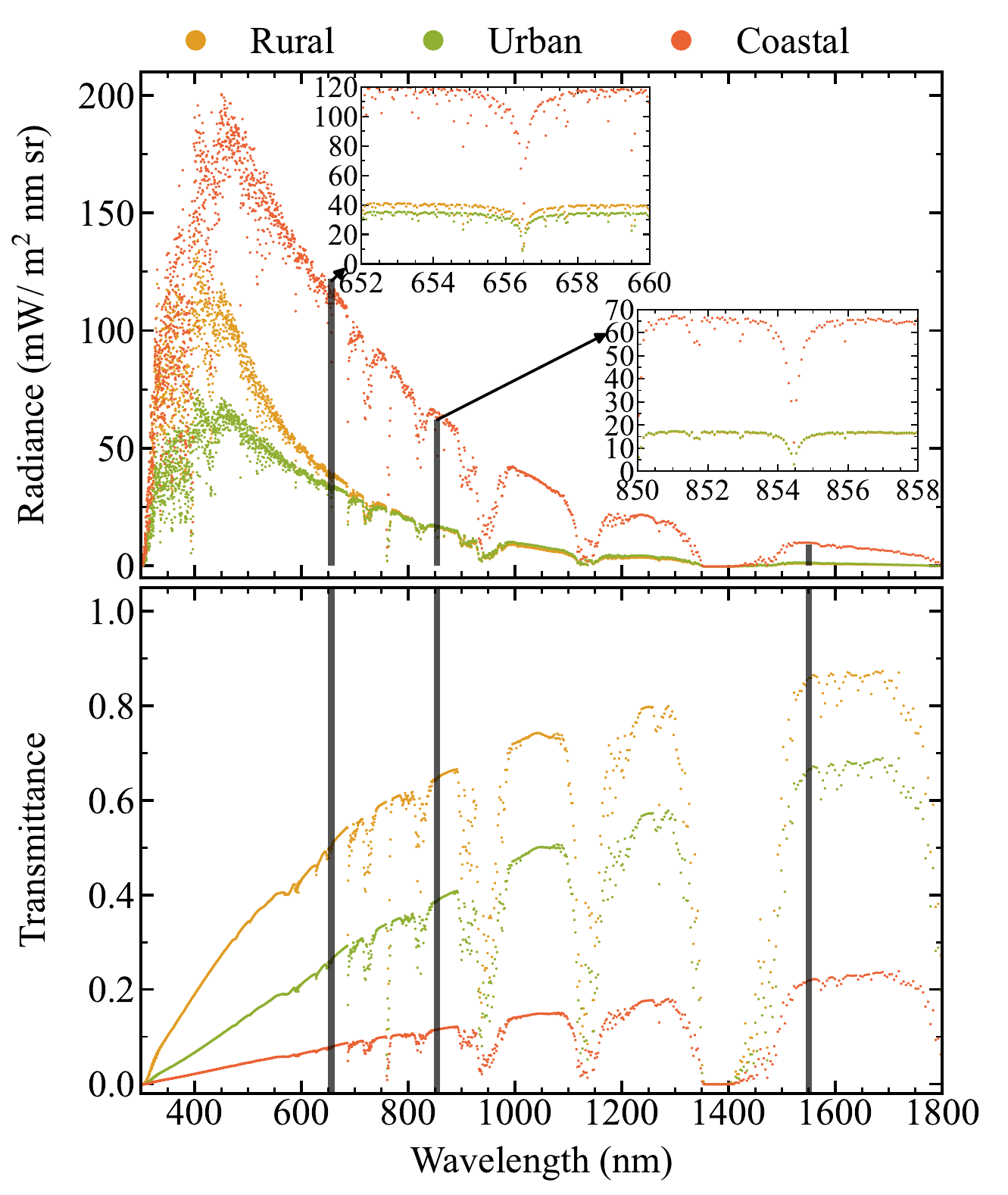}
  \caption{Spectral radiance (top, y axis) and atmospheric transmittance (bottom, y axis) versus wavelength (x axis) for three site types: Rural, Urban, and Coastal. The simulations were performed using the \texttt{libRadtran} software package. The receiver is pointed at $60^\circ$ from zenith, and the solar zenith angle is $45^\circ$. Black vertical lines mark the three  wavelengths considered in this work; 656.448 nm (H$\alpha$), 854.445 nm (Ca~II), and 1550.027 nm (C-band representative). Insets in the radiance panel magnify the depressions near H$\alpha$ and Ca~II absorption lines. Radiance units: mW m$^{-2}$ nm$^{-1}$ sr$^{-1}$; transmittance is unitless.}
  \label{fig:radandtrans}
\end{figure}

\subsection{Single-photon detection}

The performance of a satellite-to-ground QKD link critically depends on the receiver's detection hardware. We consider  Si APDs for visible and near-infrared wavelengths, InGaAs/InP APDs for telecom wavelengths, and SNSPDs across the spectrum. A comparative survey of commercially available modules and state-of-the-art research devices is summarized in Table~\ref{tab:masteRAT_pd}.

\begin{table*}[t]
  \centering
  \caption{Comparative specifications of single-photon detectors. In the table: DCR, dark count rate; Temp., temperature. }
  \label{tab:masteRAT_pd}
  \scriptsize
  \setlength{\tabcolsep}{4pt}
  \renewcommand{\arraystretch}{0.95}
  \begin{adjustbox}{center,max width=\textwidth}
    \begin{tabularx}{\textwidth}{>{\raggedright\arraybackslash}Xcccccc}
      \toprule
      \textbf{Detector }                                                          & \textbf{Efficiency}       & \textbf{DCR (Hz)}    & \textbf{Dead Time}   & \textbf{Afterpulse} & \textbf{Jitter} & \textbf{Temp.}  \\
      \midrule
      \multicolumn{5}{l}{\textbf{Si APDs}}                                                                                                                                                                            \\
      Laser Components COUNT Series \cite{lasercomponents_count_series_2023}
                                                                                  & $\sim$70\% @656\,nm       & 10--250              & 45\,ns               & 0.2\%               & 1000 ps         & 10--40$^\circ$C \\
                                                                                  & $\sim$40\% @854\,nm       & \multicolumn{5}{c}{}                                                                                  \\
      \addlinespace
      Laser Components COUNT T-Series  \cite{lasercomponents_count_t_series_2023} & \quad $\sim$70\% @656\,nm & 100--250             & 45\,ns               & 1\%                 & 350 ps          & 10--40$^\circ$C \\
                                                                                  & $\sim$40\% @854\,nm       & \multicolumn{5}{c}{}                                                                                  \\
      \addlinespace
      Laser Components COUNT NIR Series \cite{lasercomponents_count_nir_2024}     & \quad$\sim$65\% @656\,nm  & 50--500              & 45\,ns               & 0.2\%               & 1000 ps         & 10--40$^\circ$C \\
                                                                                  & $\sim$50\% @854\,nm       & \multicolumn{5}{c}{}                                                                                  \\
      \addlinespace
      Excelitas SPCM-NIR \cite{excelitas_spcm_nir_2019}                           & $\sim$75\% @656\,nm       & 10--1500             & 22--35\,ns           & 1\%                 & 350 ps          & 5--70$^\circ$C  \\
                                                                                  & $\sim$55\% @854\,nm       & \multicolumn{5}{c}{}                                                                                  \\
      \addlinespace
      Excelitas SPCM-AQRH \cite{excelitas_spcm_aqrh_2019}                         & $\sim$70\% @656\,nm       & 25--1500             & 22--42\,ns           & 0.5\%               & 350 ps          & 5--70$^\circ$C  \\
                                                                                  & $\sim$55\% @854\,nm       & \multicolumn{5}{c}{}                                                                                  \\
      \addlinespace
      Thorlabs SPDMH \cite{thorlabs_spdmh2f_2025}                                 & $\sim$70\% @656\,nm       & 100                  & 45\,ns               & 0.2\%               & 1000 ps         & 10--40$^\circ$C \\
                                                                                  & $\sim$55\% @854\,nm       & \multicolumn{5}{c}{}                                                                                  \\
      \addlinespace
      PicoQuant tau-SPAD \cite{ picoquant_tau_spad_discontinued_2014}             & $\sim$65\% @656\,nm       & $<20$                & 70\,ns               & $<1$\%              & 350-800ps       & 10--40$^\circ$C \\
                                                                                  & $\sim$55\% @854\,nm       & \multicolumn{5}{c}{}                                                                                  \\
      \addlinespace
      IDQ ID120~\cite{ID120fea43:online}                                          & $\sim$70\% @656\,nm       & $<300$               & 1\,$\mu$s            & ---                 & $<$400\,ps      & 10--40$^\circ$C \\
                                                                                  & $\sim$60\% @854\,nm       & \multicolumn{5}{c}{}                                                                                  \\
      \addlinespace
      Anisimova \emph{et\,al.} \cite{ anisimova2021low}                           & $\sim$50\% @854\,nm       & $<5$                 & 500\,ns              & 0.35\%              & 500-1000 ps     & -100$^\circ$C   \\
      \addlinespace
      \midrule
      \multicolumn{7}{l}{\textbf{InGaAs/InP APDs (1550\,nm)}}                                                                                                                                                         \\
      IDQ ID230 \cite{idquantique_id230_2023}                                     & 20\%                      & 200                  & 2--100\,$\mu$s       & ---                 & 200ps           & 10--25$^\circ$C \\
      \addlinespace
      IDQ ID Qube NIR \cite{idquantique_qube_nir_fr_2025}                         & 25\%                      & $1200$               & 0.1--80\,$\mu$s      & ---                 & 200 ps          & 10--35$^\circ$C \\
      \addlinespace
      Korzh \emph{et\,al.}  \cite{korzh2014free}                                  & 11\%                      & 1                    & 20\,$\mu$s           & 2.2\%               & 400 ps          & -110$^\circ$C   \\
      \addlinespace
      Comandar \emph{et\,al.} \cite{comandar2015gigahertz}                        & 55\%                      & ---                  & 10\,ns               & 10.2\%              & $\sim$100 ps    & 20$^\circ$C     \\
      \addlinespace
      Fang \emph{et\,al.} \cite{fang2020ingaas}                                   & 40\%                      & $3000$               & 88\,ns               & 5.5\%               & ---             & -20$^\circ$C    \\
      \addlinespace
      Xu \emph{et\,al.}  \cite{xu2024compact}                                     & 40\%                      & $2300$               & 10\,$\mu$s           & 8\%                 & 49 ps           & -80$^\circ$C    \\
      \addlinespace
      Yan \emph{et\,al.} \cite{yan2012ultra}                                      & 10\%                      & 100                  & ---                  & ---                 & 30 ps           & -60$^\circ$C    \\
      \addlinespace
      Signorelli \emph{et\,al.} \cite{signorelli2021ingaas}                       & 17\%                      & $2000$               & 1\,$\mu$s            & 4.5\%               & 319 ps          & -48$^\circ$C    \\
      \addlinespace
      Bouzid \emph{et\,al.} \cite{bouzid2014ingaas}                               & 20\%                      & 120                  & ---                  & 1.2\%               & ---             & -40$^\circ$C    \\
      \addlinespace
      Li \emph{et\,al.} \cite{li2022ultra}                                        & 20\%                      & 320                  & 1\,$\mu$s            & 0.57\%              & ---             & -40$^\circ$C    \\
      \addlinespace
      AUREA  \cite{aurea_spd_a_nir_2019}
                                                                                  & 30\%                      & $<800$               & 100\,ns--1\,ms       & $<0.1$\%            & 150\,ps         & 10--30$^\circ$C \\
      \addlinespace
      He \emph{et\,al.} \cite{he2022high_pde_ingaas}
                                                                                  & 55.4\%                    & 43800                & ---                  & ---                 & ---             & -26$^\circ$C    \\
      \addlinespace
      Ye \emph{et\,al.} \cite{ye2025room_temp_inp_ingaas}
                                                                                  & 20.1\%                    & 6820                 & ---                  & 0.15\%              & 120\,ps         & 20$^\circ$C     \\
      \midrule
      \multicolumn{7}{l}{\textbf{SNSPDs}}                                                                                                                                                                             \\
      \addlinespace
      Hu \emph{et\,al.}  \cite{hu2020detecting}                                   & \quad ~95\%               & 100                  & 42\,ns               & ---                 & 66\,ps          & -271$^\circ$C   \\
      \addlinespace
      IDQ ID281 Pro SNSPD \cite{idquantique_id281_pro_2025}                       & $>$90\% @810\,nm          & 1 - 5 @810\,nm       & 1--30\,ns            & ---                 & 3--40\,ps       & -270$^\circ$C   \\
                                                                                  & ~~$>$90\% @1550\,nm       & 1 - 100 @1550\,nm    & \multicolumn{4}{c}{}                                                           \\
      \addlinespace
      Marsili \emph{et\,al.} \cite{marsili2013detecting}
                                                                                  & $>90\% $                  & $\sim$1              & 40\,ns               & ---                 & 150\,ps         & -271$^\circ$C   \\
      \addlinespace
      Single Quantum Interleaved SNSPDs \cite{singlequantum_interleaved_2025}
                                                                                  & $>90\% $                  & $<100$               & 4\,ns                & ---                 & $<20$\,ps       & -270.7$^\circ$C \\
      \bottomrule
    \end{tabularx}
  \end{adjustbox}
\end{table*}

% Based on the survey, we observe that Si-APDs typically provide detection efficiencies in the range of 40--75\% with dark count rates (DCR) varying from 10 to 500\,Hz. In the telecom band, InGaAs APDs generally exhibit lower efficiencies (10--30\%) and higher noise (DCR up to several kHz), although active cooling and gating can improve performance and bring DCR down to hundreds of Hz. Conversely, SNSPDs consistently demonstrate superior performance across all wavelengths, with efficiencies exceeding 90\% and DCRs often below 10\,Hz, albeit at the cost of cryogenic operation.

In our simulations, we define two performance tiers for each detector class; denoted as \textit{Spec A} and \textit{Spec B} (see Table~\ref{tab:detector_params}). These tiers are designed to span the range from standard commercial off-the-shelf components to state-of-the-art research devices.
% For SNSPDs, Spec A represents a high-performance system with 85\% detection efficiency and a DCR of 10\,Hz, and Spec B is  an ultra-low-noise device with 95\% efficiency and a DCR of just 0.1\,Hz. Both tiers are assumed to exhibit negligible afterpulsing, a characteristic feature of superconducting nanowire technology. The specifications for APDs are wavelength-dependent; we consider  Si-APDs for the 656\,nm and 854\,nm bands and InGaAs/InP APDs for the 1550\,nm telecom band. For this technology, Spec A aligns with standard commercial modules featuring moderate efficiencies (25\%--50\%), higher dark counts (100--300\,Hz), and afterpulsing probabilities in the 2\%--3\% range. Spec B models higher-performance, optimized modules with improved efficiencies (35\%--70\%), reduced DCRs (10--50\,Hz), and minimized afterpulsing (0.5\%--1\%). 

Figure~\ref{fig:noise_probability_detector_scenarios_sm} shows the background noise as a function of the sites, detectors, and illumination conditions. At night, the lowest points are set by detector dark counts, so the SNSPD cases sit well below the APD cases, but as the daylight background rises the curves bunch together and the location dependence becomes dominant, with the coastal points consistently highest. The figure identifies the regimes in which better detectors materially help and when the sky background becomes the real bottleneck.

\begin{figure}[htbp]
  \includegraphics[width=\linewidth]{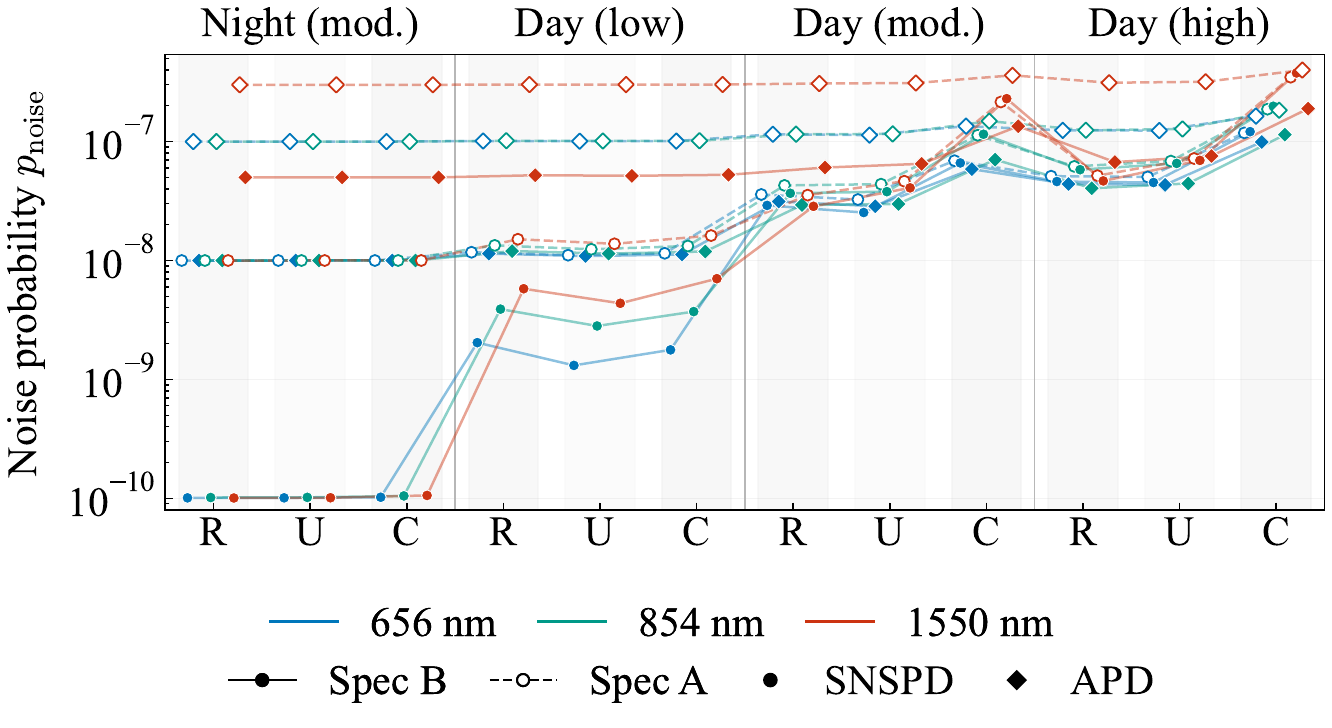}
  \caption{Noise-click probability $p_{\mathrm{noise}}$ per detector per gate (y axis) for different location, detector, and illumination scenarios (x axis). Precisely, the four scenario blocks correspond to nighttime (moderate), daytime low, daytime moderate, and daytime high background conditions. Within each block the x-axis labels R, U, and C denote rural, urban, and coastal sites. Curve color denotes wavelength $\{656.448,\,854.445,\,1550.027\}$ nm, marker shape distinguishes SNSPDs (circles) from APDs (diamonds), and line/marker fill distinguishes Spec A (dashed/hollow) from Spec B (solid/filled).}
  \label{fig:noise_probability_detector_scenarios_sm}
\end{figure}

\subsection{Misalignment}
\label{sec:misalignment}
Extensive experimental and theoretical work has demonstrated effective approaches for mitigating polarization misalignment in satellite-based QKD systems.
In the \textit{Micius} missions, active polarization control was implemented using reference laser pulses to monitor the evolving polarization state of the downlink channel, enabling real-time compensation of rotation caused by satellite motion and atmospheric birefringence~\cite{liao2017satellite,han2020polarization}.
This closed-loop stabilization maintained the received polarization frame to within a fraction of a degree and achieved an average intrinsic misalignment error of approximately $e_{\mathrm{mis}}\!\approx\!0.5$--$0.7\%$~\cite{liao2017satellite}.
Subsequent polarization-maintenance studies in free-space and satellite QKD systems have shown that polarization-maintaining optical trains, periscope-type telescope designs, and multi-waveplate compensation can sustain polarization extinction ratios on the order of $10^2$–$10^3$, implying intrinsic misalignment errors at or below the half-percent level~\cite{han2020polarization,Wu2020,Wu2022,Tan2024}.
At the network scale, long-baseline experiments, have combined trajectory-predicted polarization-frame rotation with automatic compensation using motorized waveplates, reporting overall intrinsic channel errors in the sub-percent regime~\cite{liao2017satellite,khmelev2024eurasian}.

%% file: journal/sections/appendix_plots_reduced.tex
\section{Supplementary plots}
\label{sec:appendix_plots}

This appendix first presents supplementary plots not present in the main text and then, Sec.~\ref{sec:appendix_plots_expanded} collects the additional panels from selected main-text figure families, excluding panels already shown in the main text.
Unless a quantity is explicitly swept on an axis or across panels, the remaining parameters follow the defaults summarized in Table~\ref{tab:sim_params2}.

Figure~\ref{fig:channel_loss_grid_sm} plots the system loss as a function of receiver aperture and pointing error for different wavelengths and transmitter apertures. On the fixed urban baseline, larger receiver apertures always move the operating point toward lower loss, while pointing jitter drives it back toward the high-loss regime. Increasing the transmitter aperture relaxes that trade throughout the grid, but the column comparison shows that the wavelength dependence is not one-dimensional: $656.448$ and $854.445$ nm reach the lowest best-case losses in the low-jitter corner, whereas $1550.027$ nm sacrifices some best-case performance in exchange for contours that are less steep against pointing error and therefore more tolerant once jitter grows.

\begin{figure}[htbp]
    \includegraphics[width=\linewidth]{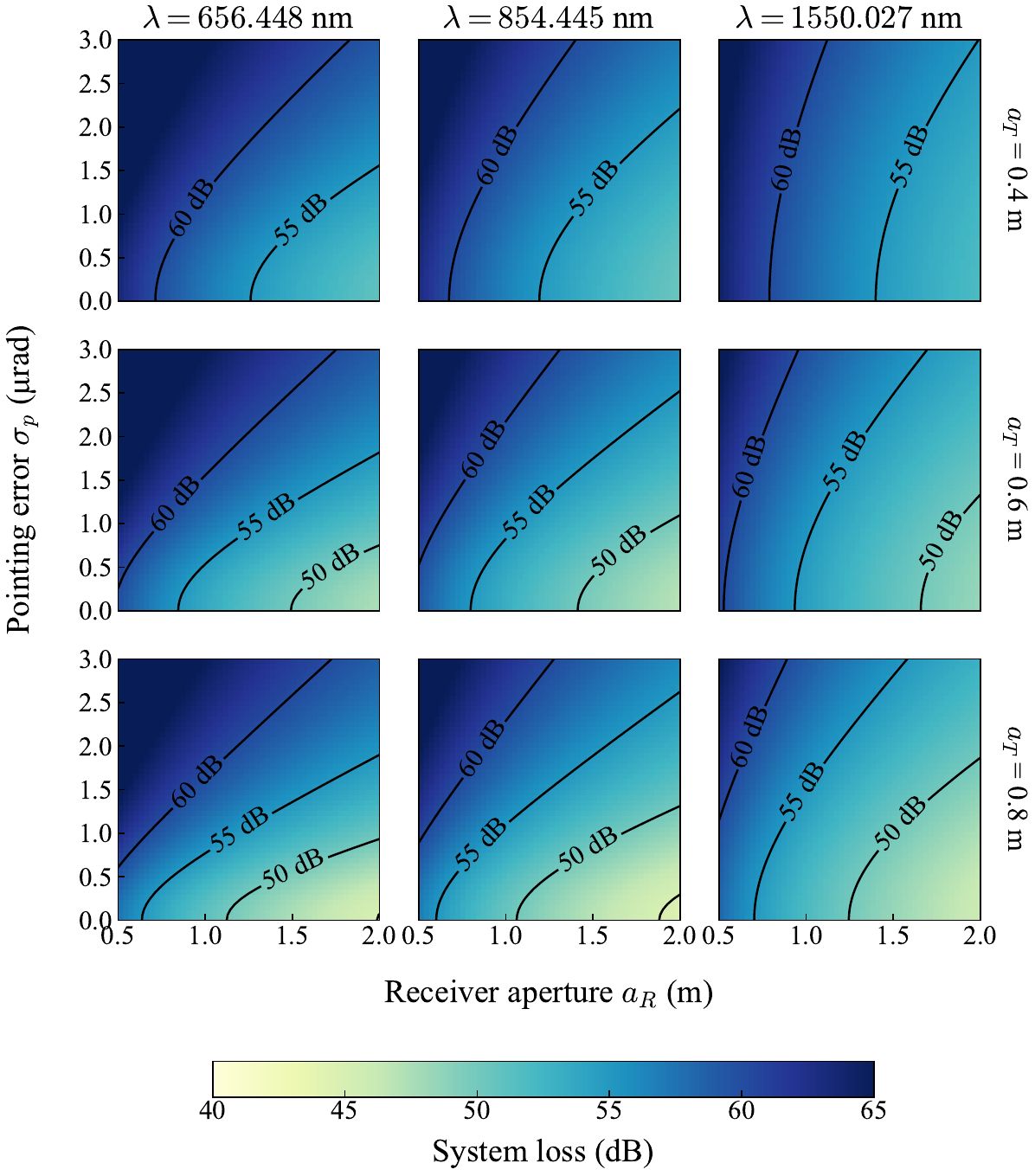}
    \caption{System loss (dB) as a function of receiver aperture $a_R$ (x axis) and pointing error $\sigma_p$ (y axis) on the fixed urban baseline with $\theta=60^\circ$, low AO correction ($f_c=130$ Hz), SMF coupling, and the default SNSPD Spec B receiver. The color map and labeled black contours indicate iso-loss levels. Columns correspond to $\lambda=656.448$ nm, $\lambda=854.445$ nm, and $\lambda=1550.027$ nm, while rows correspond to transmitter apertures $a_T$ of $0.4$ m, $0.6$ m, and $0.8$ m.}
    \label{fig:channel_loss_grid_sm}
\end{figure}

Figure~\ref{fig:zenith_rate_noise_scenarios_sm} fixes the site and receiver aperture and varies daytime noise and transmitter aperture. For 656.448 and 854.445 nm, increasing the daytime background mainly acts by pulling the zenith cutoff leftward; the rate level at small zenith angles changes less dramatically than the accessible-angle range. Increasing the transmitter aperture from $0.75$ m to $1.0$ m recovers part of that lost margin, but not enough to offset the high-noise penalty completely. The 1550.027 nm row degrades more gently across the same columns, consistent with the narrower telecom filter bandwidth assumed in the baseline model and higher transmission at large zenith angles.

\begin{figure}[htbp]
    \includegraphics[width=\linewidth]{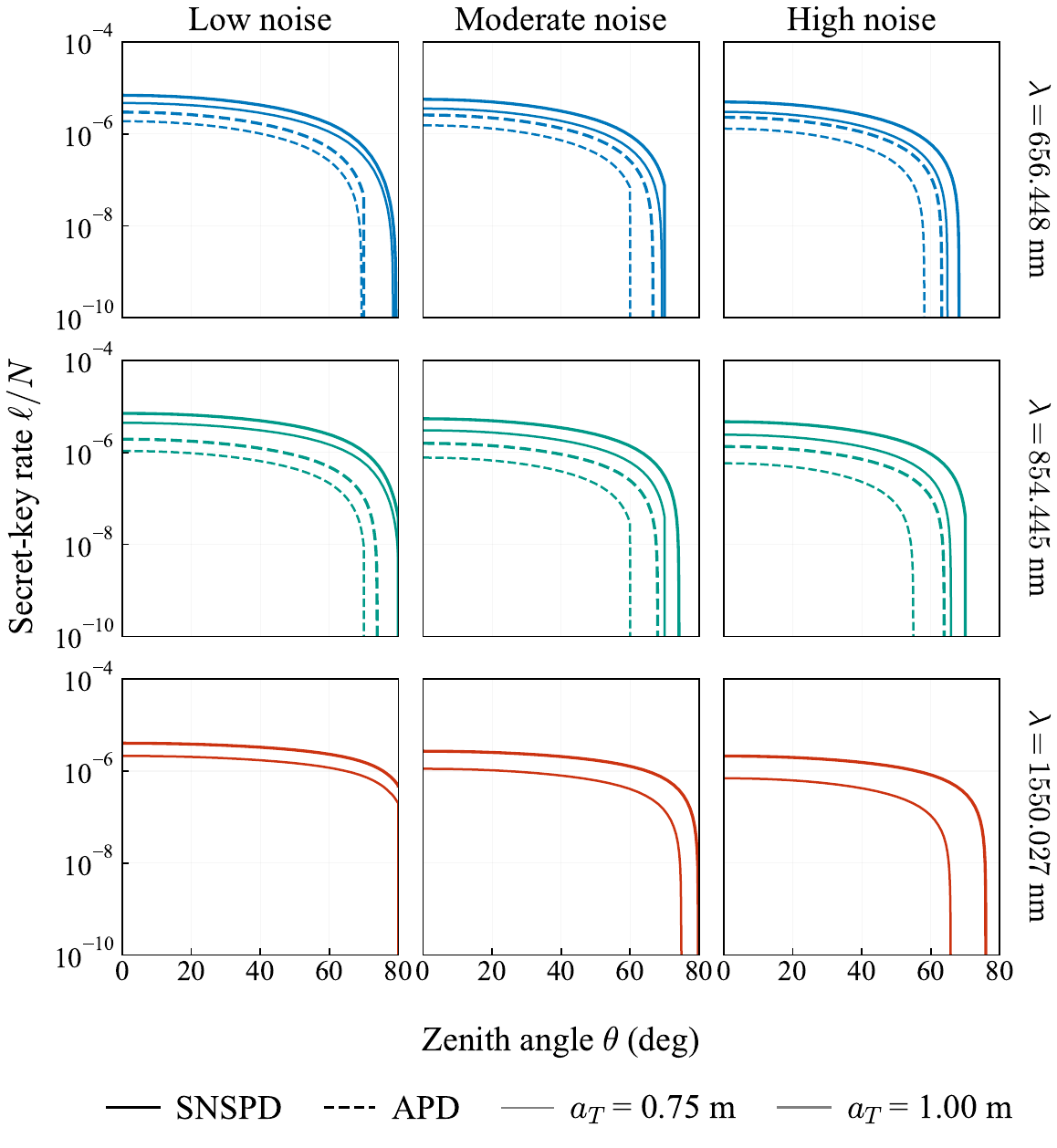}
    \caption{Finite-size secret-key rate $\ell/N$ (y axis) versus zenith angle (x axis) under daytime background-noise conditions for the default urban geometry with receiver aperture $a_R=1.5$ m, low AO correction ($f_c=130$ Hz), SMF coupling, and Spec B detectors. Columns correspond to low, moderate, and high noise, and rows to $\lambda=656.448$, $\lambda=854.445$, and $\lambda=1550.027$ nm. Solid curves denote SNSPDs, dashed curves denote APDs, and line width distinguishes transmitter apertures $a_T=0.75$ m and $1.0$ m.}
    \label{fig:zenith_rate_noise_scenarios_sm}
\end{figure}

Figure~\ref{fig:zenith_rate_ao_impact_sm}  shows that the improvement from no AO to $f_c=130$ Hz opens most of the feasible region, whereas the further increase to $500$ Hz mainly adds margin. Without AO, the urban and coastal cases are nearly impossible. Once AO is active, the rural and urban rows recover broad SNSPD operating windows, while coastal links remain the hardest regime and therefore benefit the most from stronger correction.

\begin{figure}[htbp]
    \includegraphics[width=\linewidth]{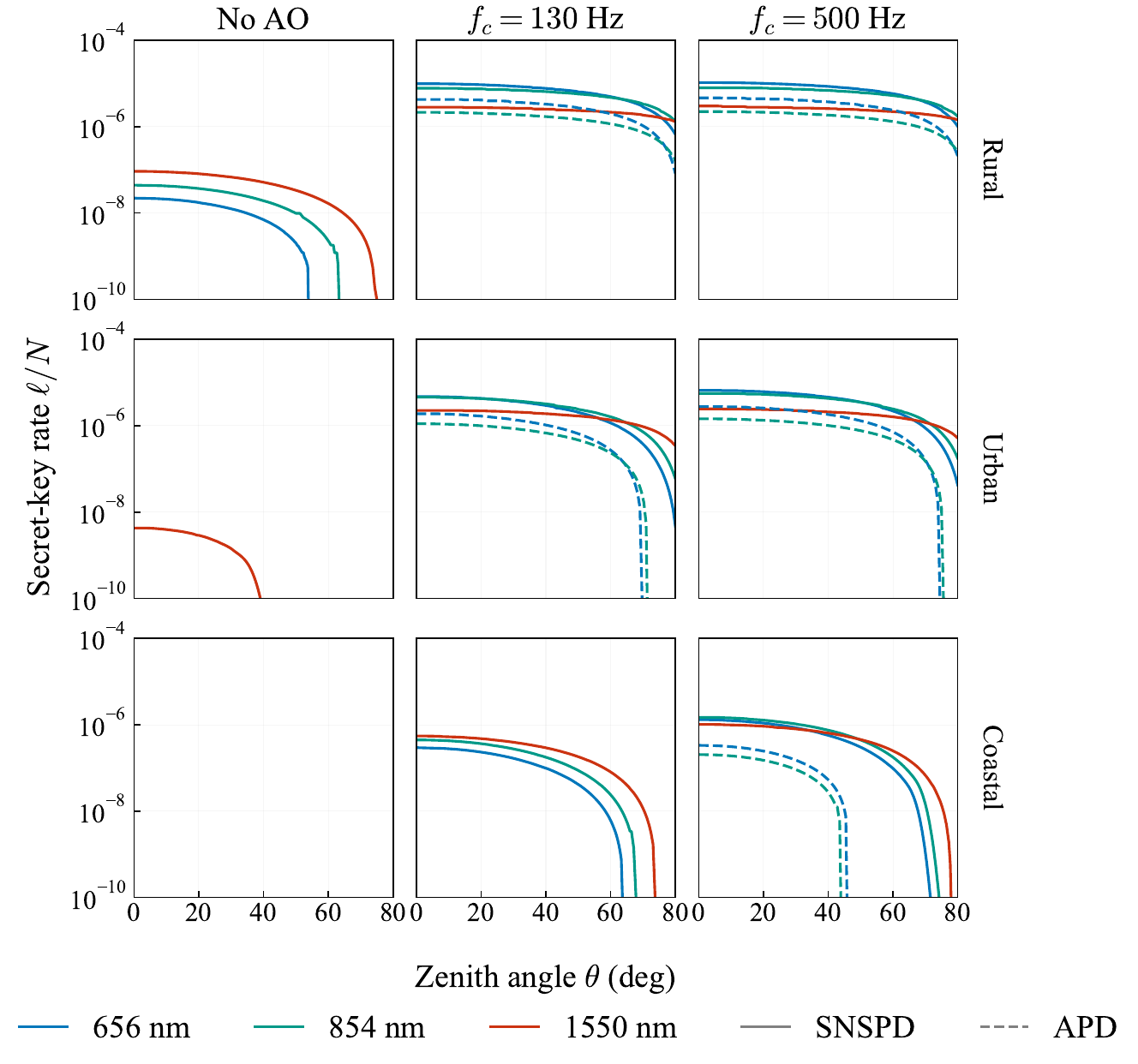}
    \caption{Finite-size secret-key rate $\ell/N$ (y axis) versus zenith angle (x axis) as a function of the bandwidth of the AO under nighttime operation at fixed apertures $a_T=0.75$ m and $a_R=1.5$ m, with SMF coupling and Spec B detectors. Rows correspond to rural, urban, and coastal locations, and columns compare no AO, low-bandwidth AO ($f_c=130$ Hz), and strong AO ($f_c=500$ Hz). Curve color denotes wavelength $\{656.448,\,854.445,\,1550.027\}$ nm, and line style distinguishes SNSPDs (solid) from APDs (dashed).}
    \label{fig:zenith_rate_ao_impact_sm}
\end{figure}

Figure~\ref{fig:fried_aperture_rate_contour_sm} illustrates the impact of different AO correction bandwidths in terms of the effective Fried parameter, making the dependence on atmospheric seeing explicit. The overlaid black dashed lines translate the closed-loop bandwidths $f_c=130$, $200$, and $500$ Hz into the corrected $r_0$ reached at each receiver aperture. The intersections of these lines with the white key-rate contours show how much aperture can be traded for stronger correction. In rural conditions the curves lie well inside the positive-key region, so increasing the AO bandwidth yields only a modest reduction in the required receiver aperture, with the gain somewhat more visible at longer wavelengths. In coastal conditions, by contrast, the AO curves move much more strongly across the contour map, showing that higher-bandwidth correction can substantially relax the aperture requirement.

\begin{figure}[htbp]
    \includegraphics[width=\linewidth]{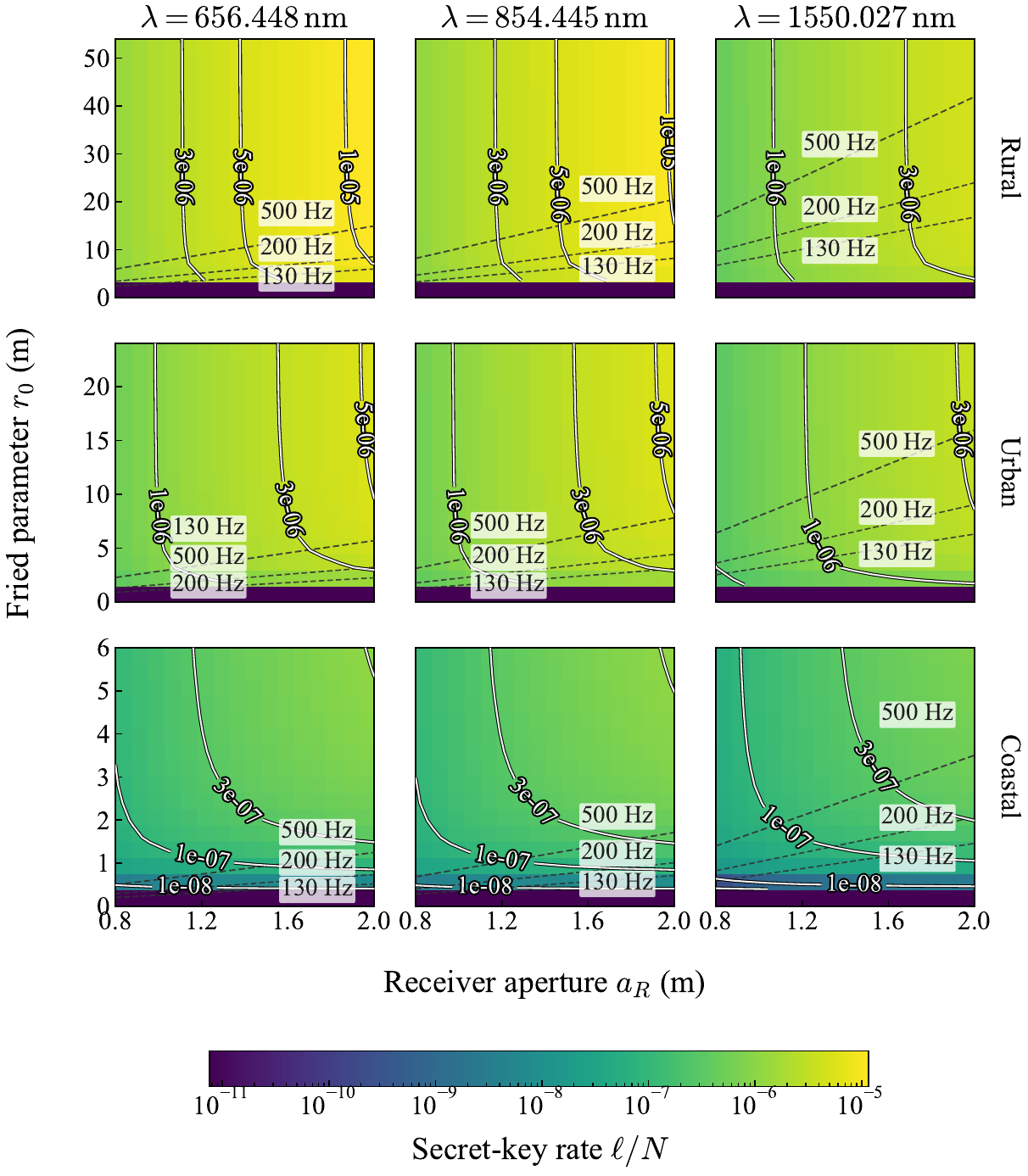}
    \caption{Finite-size secret-key rate $\ell/N$ as a function of receiver aperture $a_R$ (x axis) and Fried parameter $r_0$ (y axis) under nighttime operation, for fixed transmitter aperture $a_T=0.75$ m, SMF coupling, and the default SNSPD Spec B receiver. The color map and white contour labels show the finite-size secret-key rate, while the overlaid black dashed lines mark the effective AO-corrected Fried-parameter conditions for closed-loop bandwidths $f_c=130$, $200$, and $500$ Hz at fixed zenith angle $60^\circ$. Rows correspond to rural, urban, and coastal environments, and columns correspond to $\lambda=656.448$ nm, $\lambda=854.445$ nm, and $\lambda=1550.027$ nm.}
    \label{fig:fried_aperture_rate_contour_sm}
\end{figure}

Figure~\ref{fig:city_keyrate_aperture_lines_sm} shows the receiver-aperture dependence for a fixed cloud model and link parameters. This figure suggests that changing the aperture changes the spread: the curves fan out as $a_R$ increases, meaning that favorable sites extract a disproportionate benefit from larger terminals.
\begin{figure}[htbp]
    \centering
    \includegraphics[width=\linewidth]{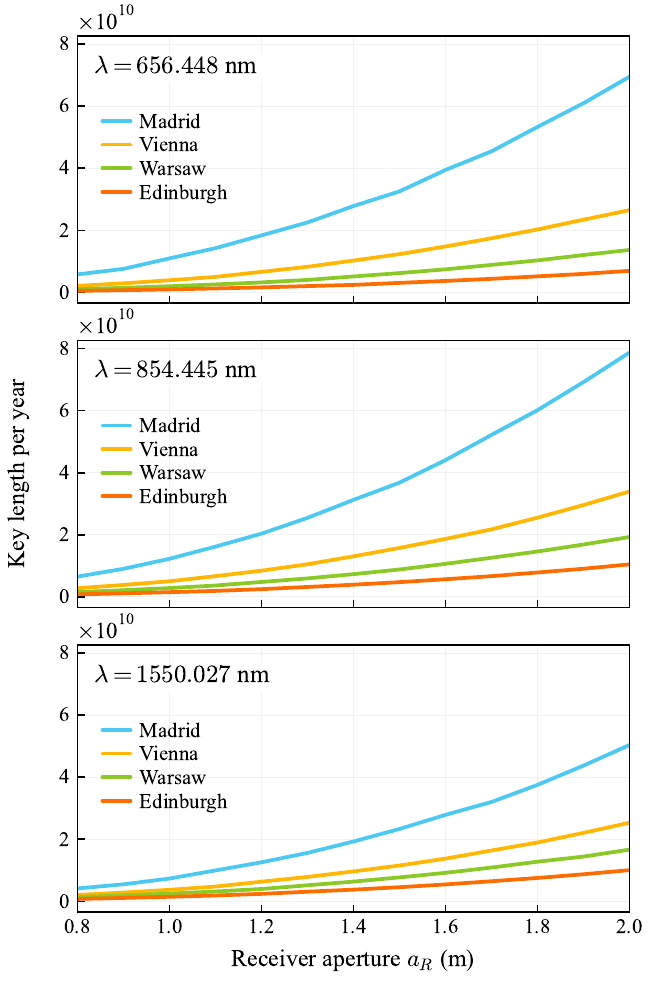}
    \caption{Annual finite-size secret-key volume (y axis) versus receiver aperture $a_R$ (x axis) for four representative cities using the cloud-weighted annual key model with the default finite-key settings, fixed transmitter aperture $a_T=0.75$ m, the urban atmospheric baseline, and the SNSPD Spec B receiver model. Rows correspond to $\lambda=656.448$ nm, $\lambda=854.445$ nm, and $\lambda=1550.027$ nm. Curve color distinguishes Madrid, Vienna, Warsaw, and Edinburgh.}
    \label{fig:city_keyrate_aperture_lines_sm}
\end{figure}

\clearpage
\subsection{Index of Expanded Main-Text Figures}
\label{sec:appendix_plots_expanded}

Finally, in Figs.~\ref{fig:loss_noise_rate_grid_sm}--\ref{fig:min_block_size_key_rate_grid_sm}, we present an expanded version of some of the figures presented in the main text. In particular, Figure~\ref{fig:loss_noise_rate_grid_sm} extends Fig.~\ref{fig:loss_noise_rate_grid} by adding a $1\%$ misalignment column and a $1\%$ afterpulsing row.
Figure~\ref{fig:aperture_loss_contour_sm} extends Fig.~\ref{fig:aperture_loss_contour} by adding a $656.448$ nm row and a $130$ Hz AO column.
Figure~\ref{fig:aperture_rate_contour_sm} extends Fig.~\ref{fig:aperture_rate_contour} by adding a $656.448$ nm row and a daytime low-noise column.
Figure~\ref{fig:zenith_aperture_loss_sm} extends Fig.~\ref{fig:zenith_aperture_loss} by adding a $656.448$ nm column and a daytime high-noise row.
Figure~\ref{fig:zenith_rate_finite_key_sm} extends Fig.~\ref{fig:zenith_rate_finite_key} by adding a coastal row and a $a_R=1.6$ m column. Finally,
Figure~\ref{fig:min_block_size_key_rate_grid_sm} extends Fig.~\ref{fig:min_b} by adding a nighttime row and a $656.448$ nm column.

\begin{figure}[htbp]
    \includegraphics[width=\linewidth]{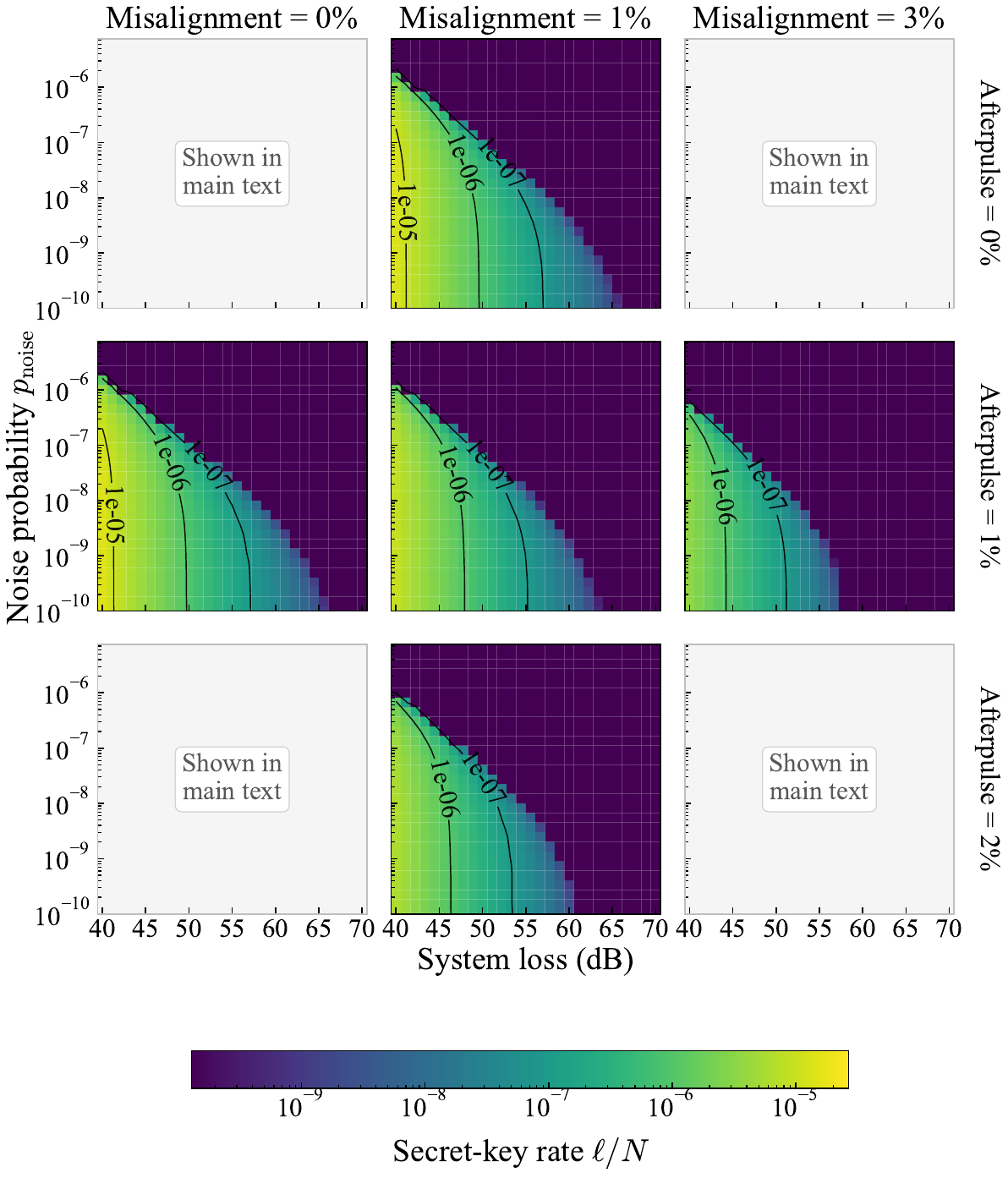}
    \caption{Finite-size secret-key rate $\ell/N$ as a function of total system loss (x axis, including receiver optics/filter and detector efficiency) and noise-click probability $p_{\mathrm{noise}}$ per detector per gate (y axis, log scale) for an asymmetric active receiver assuming the default finite-key settings, in particular $N=10^{12}$ and $\epsilon_{\mathrm{tot}}=10^{-8}$. The color map and black contour labels show the finite-size secret-key rate. Columns correspond to misalignment values of $0\%$, $1\%$, and $3\%$, while rows correspond to afterpulsing probabilities of $0\%$, $1\%$, and $2\%$. Subfigures already shown in the main text are omitted here.}
    \label{fig:loss_noise_rate_grid_sm}
\end{figure}

\begin{figure}[htbp]
    \includegraphics[width=\linewidth]{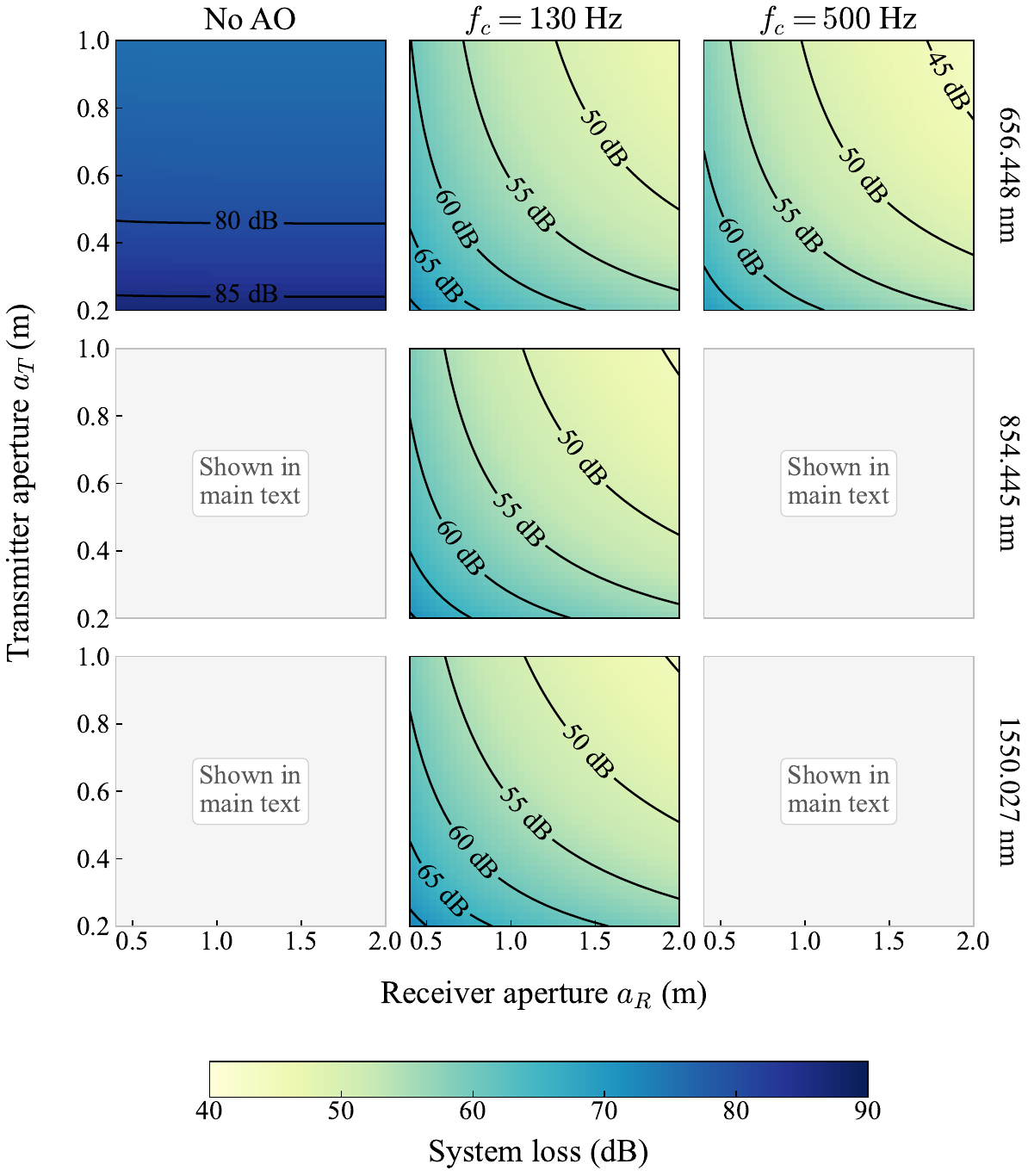}
    \caption{System loss (dB) as a function of the receiver aperture $a_R$ (x axis) and the transmitter aperture $a_T$ (y axis) on the fixed urban baseline with $\theta=60^\circ$, low pointing jitter ($0.5\,\mu\mathrm{rad}$), SMF coupling, and the default SNSPD Spec B receiver. The heat map and contour labels show iso-loss levels. Rows correspond to $\lambda=656.448$ nm, $\lambda=854.445$ nm, and $\lambda=1550.027$ nm, while columns compare no AO, low-bandwidth AO ($f_c=130$ Hz), and strong AO ($f_c=500$ Hz). Subfigures already shown in the main text are omitted here.}
    \label{fig:aperture_loss_contour_sm}
\end{figure}

\begin{figure}[htbp]
    \includegraphics[width=\linewidth]{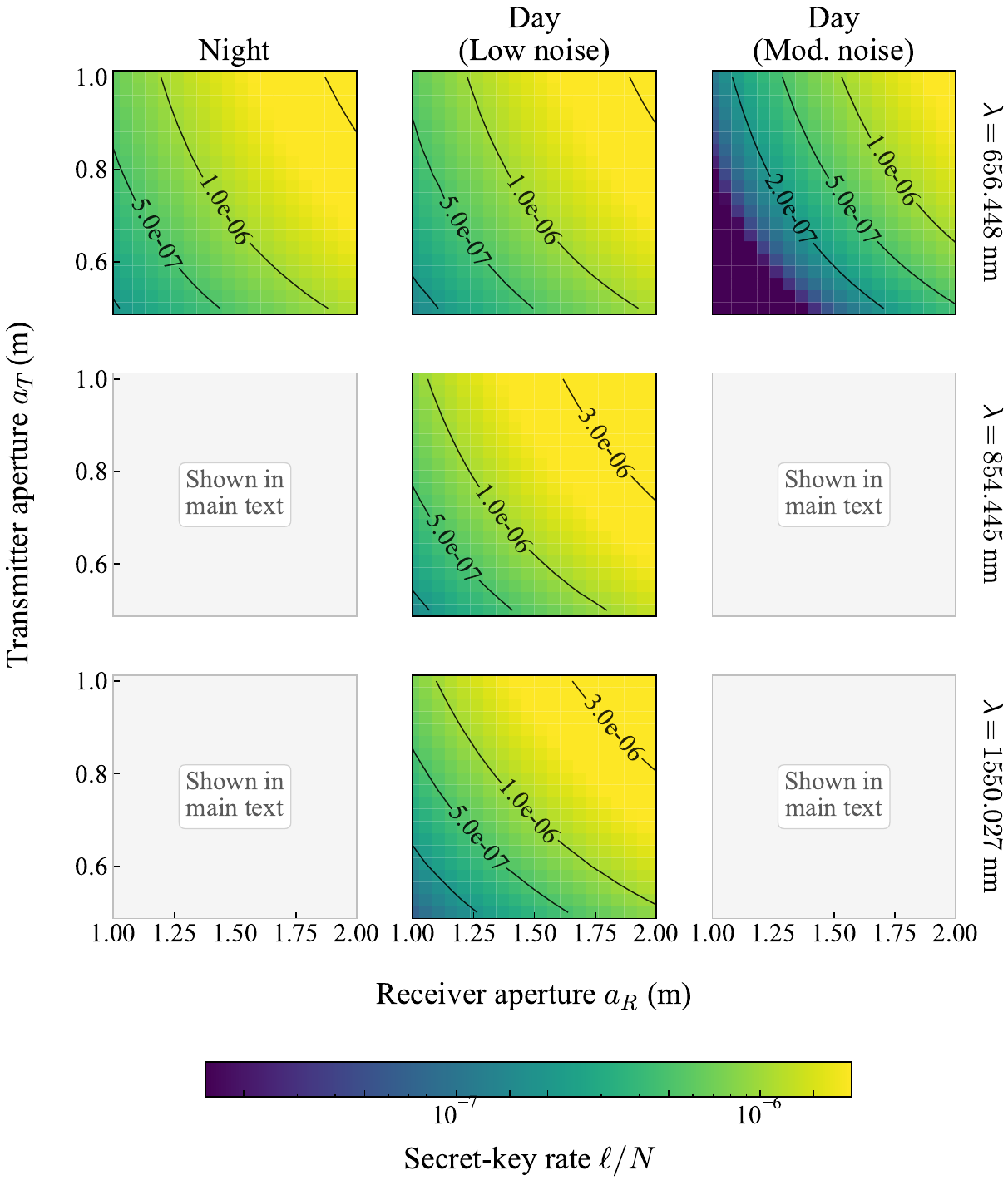}
    \caption{Finite-size secret-key rate $\ell/N$ as a function of receiver aperture $a_R$ (x axis) and transmitter aperture $a_T$ (y axis) for the default urban geometry with $\theta=60^\circ$, low AO correction ($f_c=130$ Hz), SMF coupling, and an SNSPD Spec B receiver. The color map and black contour labels show the finite-size secret-key rate. Rows correspond to $\lambda=656.448$ nm, $\lambda=854.445$ nm, and $\lambda=1550.027$ nm, while columns compare nighttime operation with daytime low-noise and daytime moderate-noise conditions. Subfigures already shown in the main text are omitted here.}
    \label{fig:aperture_rate_contour_sm}
\end{figure}

\begin{figure}[htbp]
    \includegraphics[width=\linewidth]{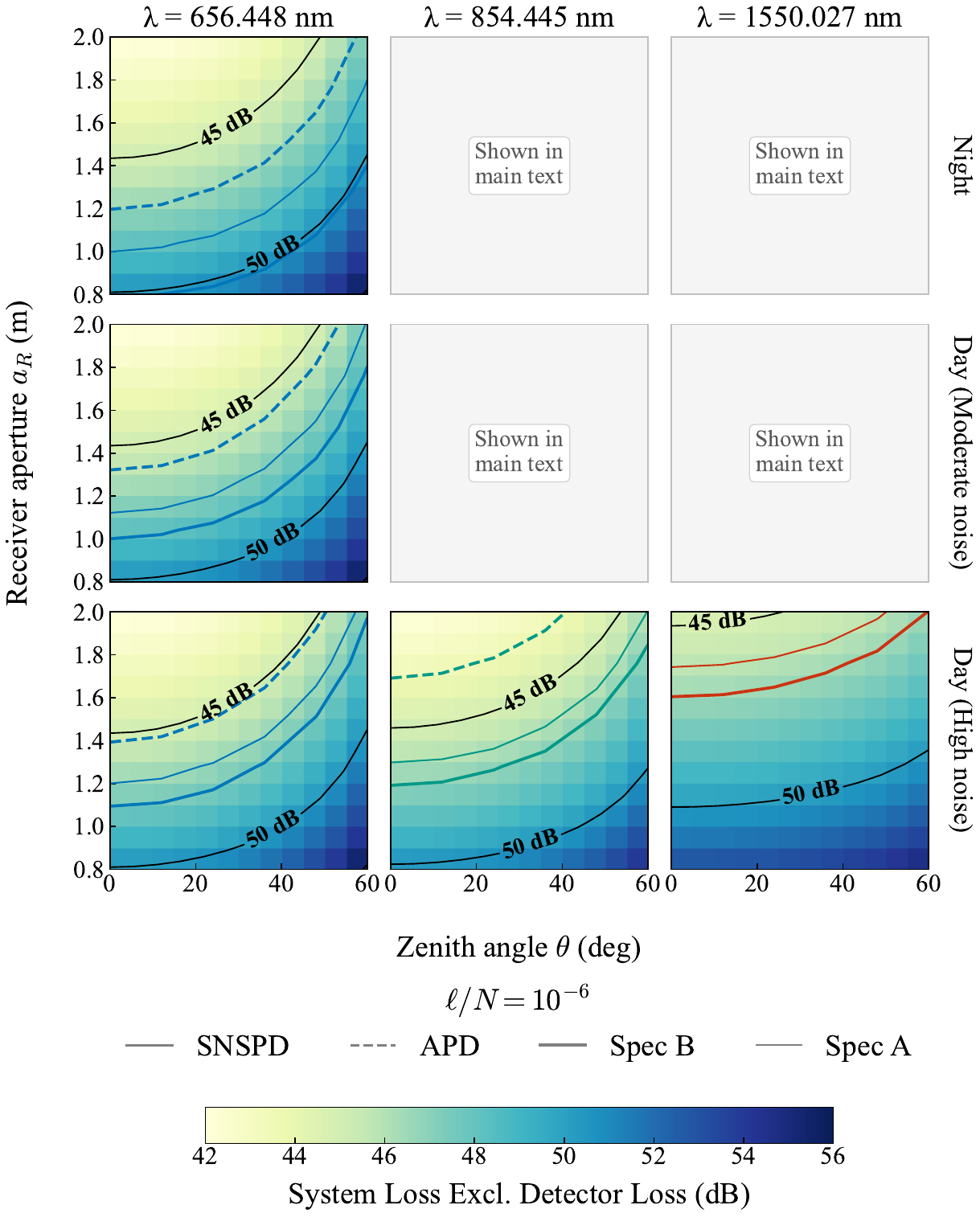}
    \caption{Combined loss and performance map as a function of zenith angle (x axis) and receiver aperture $a_R$ (y axis) at the fixed transmitter aperture $a_T=0.75$ m for the urban low-AO baseline with $f_c=130$ Hz, low pointing jitter ($0.5\,\mu\mathrm{rad}$), and SMF coupling. The background heat map and black contours show pre-detector system loss in dB, i.e., channel plus receiver optics/filter loss excluding detector efficiency, and the overlaid boundary curves mark the $\ell/N=10^{-6}$ threshold for the different detector models. Solid curves denote SNSPDs, dashed curves denote APDs, and thicker curves denote Spec B. Columns correspond to $\lambda=656.448$ nm, $\lambda=854.445$ nm and $\lambda=1550.027$ nm, while rows compare nighttime, daytime moderate-noise, and daytime high-noise operation. Subfigures already shown in the main text are omitted here.}
    \label{fig:zenith_aperture_loss_sm}
\end{figure}

\begin{figure}[htbp]
    \includegraphics[width=\linewidth]{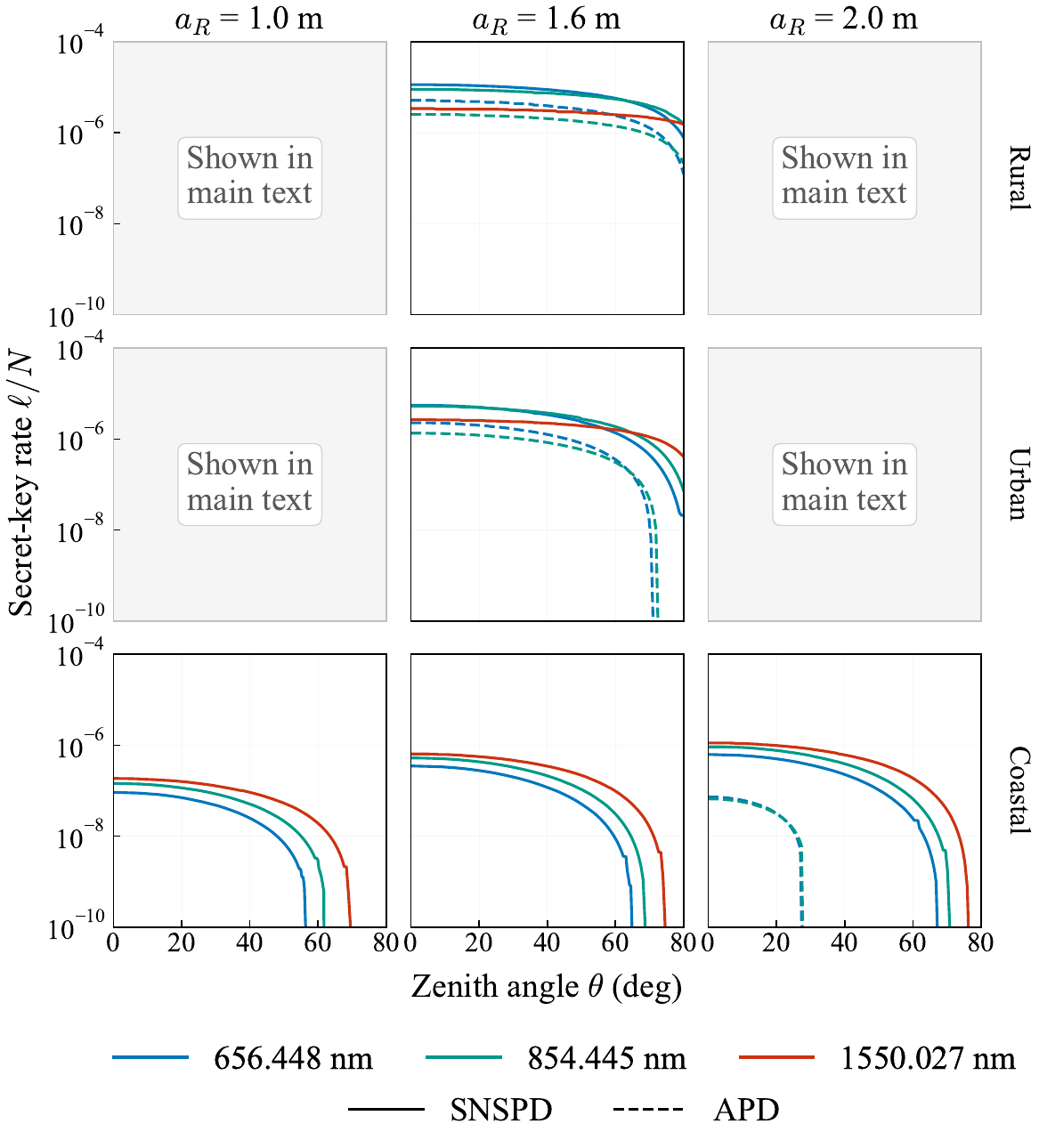}
    \caption{Finite-key rate $\ell/N$ (y axis) versus zenith angle (x axis) for different receiver apertures and operating environments under nighttime operation, with fixed transmitter aperture $a_T=0.75$ m, low AO correction ($f_c=130$ Hz), SMF coupling, and Spec B detectors. Rows correspond to rural, urban, and coastal conditions, while columns correspond to receiver apertures $a_R$ of $1.0$ m, $1.6$ m, and $2.0$ m. Curve color denotes wavelength $\{656.448,\,854.445,\,1550.027\}$ nm, and line style distinguishes SNSPDs (solid) from APDs (dashed). Subfigures already shown in the main text are omitted here.}
    \label{fig:zenith_rate_finite_key_sm}
\end{figure}

\begin{figure}[htbp]
    \includegraphics[width=\linewidth]{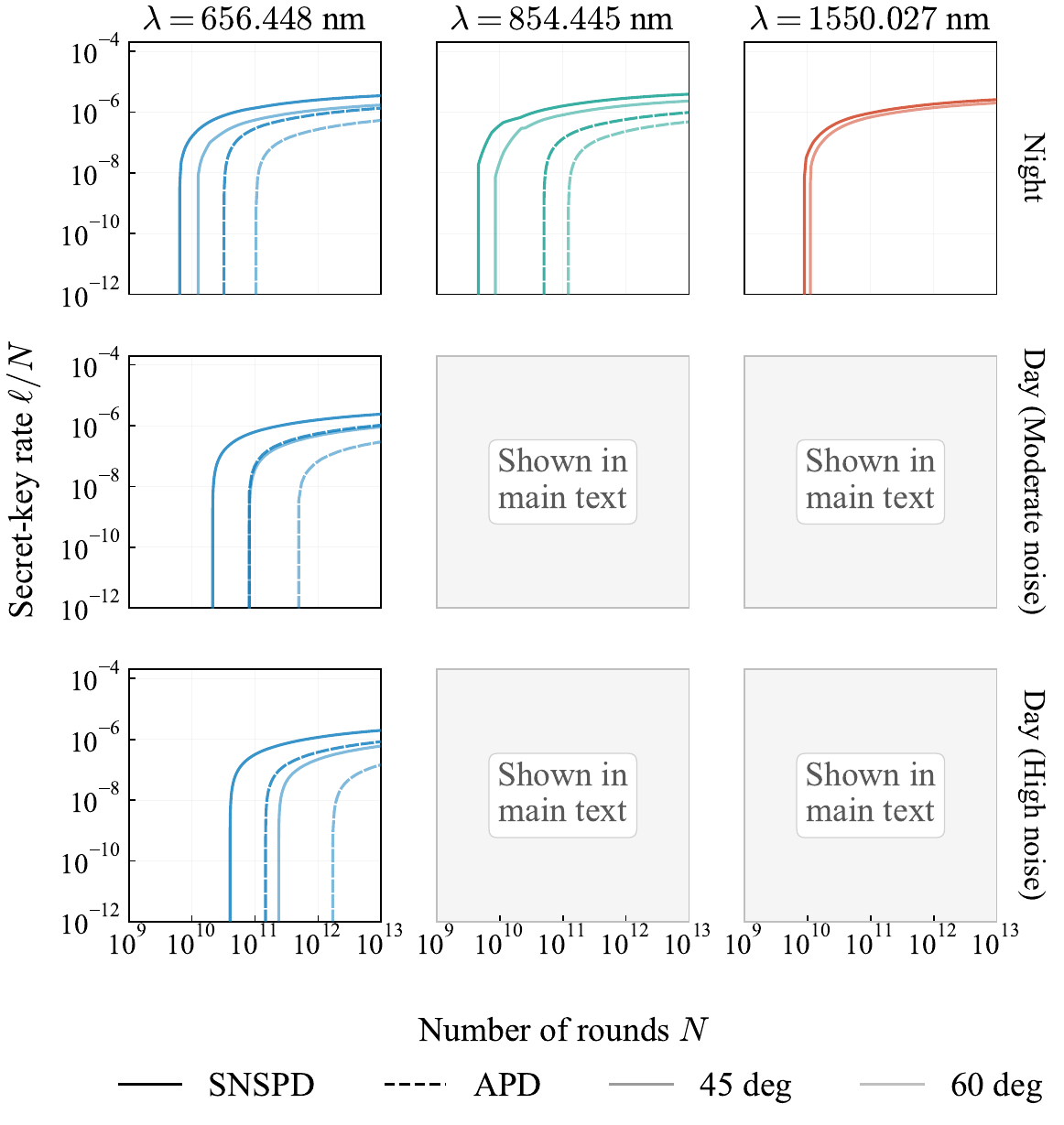}
    \caption{Finite-key rate $\ell/N$ (y axis) versus the number of transmission rounds $N$ (x axis), with both axes on logarithmic scales, for the default urban link with fixed transmitter aperture $a_T=0.75$ m, receiver aperture $a_R=1.5$ m, low AO correction ($f_c=130$ Hz), SMF coupling, and Spec B detectors. Columns correspond to $\lambda=656.448$ nm, $\lambda=854.445$ nm, and $\lambda=1550.027$ nm, while rows compare nighttime, daytime moderate-noise, and daytime high-noise operation. Solid curves denote SNSPD receivers, dashed curves denote APD receivers, and line opacity distinguishes zenith angles of $45^\circ$ and $60^\circ$. Subfigures already shown in the main text are omitted here.}
    \label{fig:min_block_size_key_rate_grid_sm}
\end{figure}